\documentclass[11pt]{article}
\usepackage[letterpaper, left=0.6in, right=0.6in, top=1in, bottom=1in]{geometry}
\usepackage{graphicx}
\usepackage{physics}
\usepackage{amsmath}
\usepackage{amssymb}
\usepackage{amsthm}
\usepackage{enumitem}
\usepackage{framed}
\usepackage{caption}
\usepackage{subfig}
\usepackage{adjustbox}
\usepackage{multirow}
\usepackage{multicol}
\usepackage{listings}
\usepackage{tikz}
\usepackage{xcolor}
\usepackage{bm}
\usepackage{mathtools}
\usepackage{algorithm}
\usepackage{algpseudocode,eqparbox}
\usepackage[T1]{fontenc}
\usepackage{wrapfig}
\usetikzlibrary{arrows,decorations.pathmorphing,decorations.footprints,decorations.pathreplacing,fadings,calc,trees,mindmap,shadows,decorations.text,patterns,positioning,shapes,matrix,fit}
\usetikzlibrary{shapes.misc}

\mathtoolsset{showonlyrefs}
\allowdisplaybreaks
\usepackage[colorlinks,citecolor=red,urlcolor=blue,bookmarks=false,hypertexnames=true]{hyperref}

\usepackage{hyperref}

\usepackage{letltxmacro}
\LetLtxMacro{\originaleqref}{\eqref}
\renewcommand{\eqref}{Eq.~\originaleqref}

\definecolor{codegreen}{rgb}{0,0.6,0}
\definecolor{codegray}{rgb}{0.5,0.5,0.5}
\definecolor{codepurple}{rgb}{0.58,0,0.82}
\definecolor{backcolour}{rgb}{0.95,0.95,0.92}

\definecolor{channelcolor}{rgb}{0.67,0.88,0.69}
\definecolor{infocolor}{rgb}{0.82,0.62,0.91}
\definecolor{cqcolor}{rgb}{0.99,0.56,0.67}
\lstdefinestyle{mystyle}{
   backgroundcolor=\color{backcolour},   
   commentstyle=\color{codegreen},
   keywordstyle=\color{magenta},
   numberstyle=\tiny\color{codegray},
   stringstyle=\color{codepurple},
   basicstyle=\ttfamily\footnotesize,
   breakatwhitespace=false,         
   breaklines=true,                 
   captionpos=b,                    
   keepspaces=true,                 
   numbers=left,                    
   numbersep=5pt,                  
   showspaces=false,                
   showstringspaces=false,
   showtabs=false,                  
   tabsize=2
}

\usepackage{stmaryrd}

\usepackage[style=ieee,backend=biber,doi=true,url=false,eprint=true]{biblatex}

\AtEveryBibitem{%
   \iffieldundef{doi}
      {}
      {\clearfield{eprint}%
       \clearfield{eprinttype}%
       \clearfield{eprintclass}}%
}

\bibliography{ref}

\usepackage{ifthen}

\newif\ifarxiv
\arxivtrue

\newif\ifisit
\isitfalse

\newif\ifextra
\extrafalse

\usetikzlibrary{shapes.geometric, arrows}
\tikzstyle{startstop} = [rectangle,  minimum width=3cm, minimum height=1cm,text centered,text width=10cm, draw=black ,fill=gray!20]
\tikzstyle{process} = [rectangle, minimum width=3cm, minimum height=1cm, text centered,text width=10cm, draw=black,fill=orange!20]
\tikzstyle{arrow} = [thick,->,>=stealth]
\tikzstyle{state}=[shape=circle,draw=blue!50,fill=blue!20]
\tikzstyle{observation}=[shape=rectangle,draw=orange!50,fill=orange!20]
\tikzstyle{lightedge}=[<-,dotted]
\tikzstyle{mainstate}=[state,thick]
\tikzstyle{mainedge}=[<-,thick]
\definecolor{bitcolor}{rgb}{1,0.84314,0}
\definecolor{checkcolor}{rgb}{0.52941,0.80784,1}
\usepackage{pgfplots}
\usepackage{longtable}
\pgfplotsset{compat=1.18}
\renewcommand{\epsilon}{\varepsilon}

\newtheorem{theorem}{Theorem}

\newtheorem{lem}[theorem]{Lemma}
\newtheorem{defn}[theorem]{Definition}
\newtheorem{corol}[theorem]{Corollary}

\newcommand{\cA}{\mathcal{A}}

\newcommand{\cC}{\mathcal{C}}

\newcommand{\cG}{\mathcal{G}}
\newcommand{\cH}{\mathcal{H}}

\newcommand{\cM}{\mathcal{M}}
\newcommand{\cN}{\mathcal{N}}

\newcommand{\cX}{\mathcal{X}}
\newcommand{\cY}{\mathcal{Y}}

\newcommand{\bA}{\bm{A}}

\newcommand{\ba}{\bm{a}}

\newcommand{\bc}{\bm{c}}

\newcommand{\be}{\bm{e}}

\newcommand{\bz}{\bm{z}}

\newcommand{\mI}{\mathbb{I}}

\renewcommand{\triangleq}{\coloneqq}

\newcommand{\sym}[1]{S_N}

\newcommand{\E}{\ensuremath{\mathbb{E}}}
\renewcommand{\Pr}{\ensuremath{\mathbb{P}}}

\global\long\def\Pr{\mathbb{P}}%

\newcommand{\vnop}{\varoast}

\newcommand{\cnop}{\boxast}

\newcommand{\blambda}{\bm{\lambda}}

\newcommand{\perr}{\Pr_{\text{err}}}
\newcommand{\pblk}{\Pr_{\text{block}}}

\newcommand{\lambdaa}{\lambda^{(1)}}
\newcommand{\lambdab}{\lambda^{(2)}}

\newcommand{\bmu}{\bm{\mu}}

\newcommand{\mone}{\bm{1}}
\newcommand{\mzero}{\bm{0}}

\newcommand{\err}{\text{Err}}

\usepackage[title]{appendix}
\usepackage{authblk}

\title{Quantum Message Passing Convergence and Vanishing Block-Error Probability for Random LDPC Codes}
\author[1,2]{Avijit Mandal\footnote{avijit.mandal@duke.edu}}
\author[4]{Christophe Piveteau}
\author[5]{Joseph M. Renes}
\author[1,2,3]{Henry D. Pfister}
\hypersetup{
   pdftitle={Quantum Message Passing Convergence and Vanishing Block-Error Probability for Random LDPC Codes},
   pdfauthor={Avijit Mandal, Christophe Piveteau, Joseph M. Renes, and Henry D. Pfister},
   pdfsubject={Belief propagation with quantum messages for random LDPC codes},
   pdfkeywords={BPQM, LDPC codes, block-error probability, pure-state channels, density evolution}
}
\affil[1]{Department of Electrical and Computer Engineering, Duke University}
\affil[2]{Duke Quantum Center, Duke University}
\affil[3]{Department of Mathematics, Duke University}
\affil[4]{IBM Research Europe -- Zurich}
\affil[5]{Institute for Theoretical Physics, ETH Zurich}

\date{}

\begin{document}
\onecolumn

\maketitle

\begin{abstract}
Belief propagation with quantum messages (BPQM) is a quantum algorithm that decodes classical codes
transmitted over classical--quantum channels. It realizes optimal decoding on
tree factor graphs over pure-state classical-quantum channels.  However, this tree-based analysis does not ensure vanishing block-error probability for LDPC Tanner
graphs with cycles.  
In this work, we construct a two-stage BPQM decoder for random $q$-ary
LDPC codes over symmetric $q$-ary pure-state channels, where $q$ is prime, and
prove that its ensemble-average block-error probability vanishes as the
blocklength $N$ tends to infinity.  For regular ensembles with $d_v\geq3$,
fidelity bounds yield double-exponential decay of the average symbol-error
probability throughout the BPQM success region.  We apply depth-$\ell$ BPQM to
coordinates with tree neighbourhoods and treat the remaining coordinates as
erasures.  With a suitable $\ell=\Theta(\log\log N)$, a noncommutative union bound controls
the BPQM decoding errors, while the minimum-distance property guarantees
erasure recovery.  We also extend the analysis to finite-support irregular
ensembles.  These results are relevant to quantum algorithms based on Regev's reduction,
where coherent decoding uncomputes a codeword register. Decoded quantum interferometry (DQI) uses a closely related Fourier-based framework that reduces sparse max-LINSAT optimization problems to LDPC decoding problems on pure-state channels. Our results justify the use of BPQM in the decoding step of DQI and of coding-theoretic algorithms based on Regev's reduction whenever the code is drawn from one of the random LDPC ensembles analyzed here and the induced memoryless symmetric pure-state channel lies in the BPQM success region. Curiously, our numerical results indicate that DQI+BPQM achieves a satisfaction ratio that closely matches that of simulated annealing.
% Density evolution (DE) tracks asymptotic symbol-error performance of BPQM decoding by propagating message distributions through the tree. 
\end{abstract}
\section{Introduction}

Decoding linear codes over classical--quantum channels is an important
problem in quantum communication
\cite{wilde2013towards,Renes-njp17}
and quantum algorithms
\cite{chailloux_quantumdecoding_2023,jordan2025optimization,shutty2026lqd}.
The study of coding for classical-quantum channels began with foundational
results by Holevo and by Schumacher and Westmoreland, who showed that rates
governed by Holevo information can be attained through joint measurements of
block outputs~\cite{holevo1998capacity,schumacher1997sending}. Subsequent work
investigated whether specific code families could attain the
Holevo information on classical-quantum channels. Wilde and Guha established channel polarization
for binary-input classical--quantum channels and proved that polar codes
achieve the symmetric Holevo information under quantum successive-cancellation
decoding~\cite{wilde2012polar}.

Here we consider symmetric \(q\)-ary pure-state channels, a
special class of classical--quantum channels relevant in algorithmic applications. Let
\begin{align*}
   W\colon x\longmapsto \ket{\psi_x}
\end{align*}
be such a channel with $x\in \mathbb{F}_q$, and let
\(\cC\subseteq\mathbb F_q^N\) be an \([N,K]\) linear code. When a codeword
$\bc\in\cC$ is transmitted,
the decoder receives $N$ quantum systems in the state 
\begin{align*}
   \ket{\Psi_{\bc}}
   \coloneqq
   \bigotimes_{i=0}^{N-1}\ket{\psi_{c_i}}
\end{align*}
and attempts to recover $\bc$. In general, the states associated with
different codewords are nonorthogonal, so decoding requires a joint
measurement that discriminates among the codeword states. Although the
optimal measurement is well defined, its implementation can require
operations whose complexity grows exponentially with the blocklength. This
motivates decoders that exploit the algebraic and graphical structure of
the code.

Such decoders have application in quantum algorithms, for example in Regev's reduction, a foundational result in lattice-based cryptography that
relates worst-case lattice problems to the learning-with-errors problem through
a quantum procedure~\cite{regev_reduction_2009}. This reduction was later
adapted to linear codes by Debris-Alazard, Remaud, and Tillich, who showed that
a decoder for a random linear code \(\cC\) can be used in a quantum algorithm
to find short codewords in its dual code
\(\cC^\perp\)~\cite{thomas_reduction_2024}. Chen, Liu, and Zhandry subsequently
considered a quantum version of learning with errors in which the errors are
provided in superposition~\cite{chen_quantumalgo_2022}. Chailloux and Tillich
adapted this formulation to codes through the quantum decoding problem, which
asks for the recovery of an unknown codeword from a superposition of its noisy
versions~\cite{chailloux_quantumdecoding_2023}. In the
coding-theoretic form of Regev's reduction, a decoder that recovers the
codeword exactly can be used to uncompute the codeword register before applying
a quantum Fourier transform. The quantum Fourier transform then produces a
state supported on a coset of the dual code. Related work has developed further connections among decoding, Fourier
sampling, and quantum algorithms for coding and algebraic
problems~\cite{yamakawa_quantumadvantage_2024,
chailloux_quantumadvanteage_2025,chailloux_opixsoftdecoders_2025,
briaud_quantumadvantage_2025}.

Jordan et al.\ introduced decoded quantum interferometry (DQI), a quantum
algorithmic framework that uses the quantum Fourier transform to reduce
classical optimization problems to decoding
problems~\cite{jordan2025optimization}. DQI associates an optimization instance
with a linear code and uses a decoder within the quantum computation to
increase the probability of sampling high-quality solutions. For sparse
max-LINSAT instances, the associated codes are LDPC codes determined by the
problem instances. Shutty et al.\ subsequently developed locally quantum
decoders for the LDPC codes arising from regular sparse max-XORSAT
instances~\cite{shutty2026lqd}. 

In DQI, and more generally in the coding-theoretic formulation of Regev's reduction, the decoder is used to coherently erase information produced during an intermediate stage of the algorithm.
To describe this uncomputation step more precisely, let \(f:\mathbb F_q^N\to\mathbb R_{\geq0}\) satisfy
\(\sum_{\be\in\mathbb F_q^N}|f(\be)|^2=1\), and define the quantum states
\begin{align}
   \ket{\Phi_{\bc}^{f}}
   \coloneqq
   \sum_{\be\in\mathbb F_q^N}
   f(\be)\ket{\bc+\be}.
\end{align}
For \(v\in\mathbb F_q^N\) and
\(\omega=e^{2\pi\mathrm{i}/q}\), the objective is to implement the transformation
\begin{align}
   \frac{1}{\sqrt{q^K}}
   \sum_{\bc\in\cC}
   \omega^{v\cdot\bc}
   \ket{\Phi_{\bc}^{f}}\ket{\bc}
   \longmapsto
   \frac{1}{\sqrt{q^K}}
   \sum_{\bc\in\cC}
   \omega^{v\cdot\bc}
   \ket{\Phi_{\bc}^{f}}\ket{\mzero}.
\end{align}
Clearly, this transformation can be achieved given a circuit that acts as $\ket{\Phi_{\bc}^{f}}\ket{\bc} \mapsto \ket{\Phi_{\bc}^{f}}\ket{0}$ for all $\bc\in\cC$.
This task can be viewed as a coherent version of a state discrimination problem: given a state $\ket{\Phi_{\bc}^{f}}$, the goal is to identify the corresponding codeword $\bc\in\cC$.
This problem is precisely the quantum decoding problem.
Perfect decoding is generally impossible, since the ideal transformation above need not be unitary.
Nevertheless, Regev's reduction remains successful provided that the decoder achieves a sufficiently large success probability~\cite{chailloux_opixsoftdecoders_2025}.

When the noise amplitude factors as
\(f(\be)=\prod_{i=0}^{N-1}a(e_i)\) for some
\(a:\mathbb F_q\to\mathbb R_{\geq0}\) satisfying
\(\sum_{z\in\mathbb F_q}|a(z)|^2=1\), the received state satisfies
\begin{align}
   \ket{\Phi_{\bc}^{f}}
   =
   \bigotimes_{i=0}^{N-1}\ket{\psi_{c_i}},
   \qquad
   \ket{\psi_x}
   \coloneqq
   \sum_{z\in\mathbb F_q}a(z)\ket{x+z}.
\end{align}
Thus, under the i.i.d.\ amplitude assumption, the quantum decoding problem
becomes the decoding of \(\cC\) over a memoryless symmetric \(q\)-ary
pure-state channel. Instances of algorithms based on Regev's reduction and DQI
in this regime therefore require reliable algorithms for decoding the codes
determined by the underlying problem instances over the induced pure-state
channels. Specifically, they require the decoder to exhibit a small block-error probability.

In this paper, we consider belief propagation with quantum messages
(BPQM), a quantum analogue of classical belief propagation (BP) for decoding classical LDPC codes transmitted over classical-quantum channels. Classical BP represents the information about a
code symbol by a likelihood function and passes such functions along the
edges of a factor graph. BPQM replaces these classical messages by quantum
systems. The information contained in the channel outputs is combined
recursively according to the variable-node and check-node constraints of
the code, and a measurement of the final quantum message produces an
estimate of the desired code symbol.  When the factor graph is a tree, this recursive decomposition exactly
follows the conditional-independence structure of the decoding problem.
BPQM can then compress the incoming quantum messages into a quantum
sufficient statistic for the root symbol. The
resulting measurement achieves the optimal symbol-error probability, and
the same operations can be incorporated into a sequential procedure for
recovering the codeword~\cite{Renes-njp17,piveteau2022quantum}.

The asymptotic performance of BPQM on computation trees can be studied by
density evolution.  Density evolution tracks the distribution of the effective
channel for estimating the root symbol after each BPQM iteration on a tree
Tanner graph.  For a nonbinary
LDPC ensemble, the effective channel depends on the nonzero edge
coefficients.  For an irregular ensemble, it also depends on the random node
degrees.  Averaging over these variables determines the BPQM success region
for the specified LDPC ensemble.

Earlier BPQM density-evolution analyses characterize symbol-error performance
on computation trees~\cite{brandsen2022belief,mandal2026belief}, but they do
not establish vanishing block-error probability for finite LDPC codes. Finite
LDPC Tanner graphs contain cycles, and the classical computation-tree
interpretation cannot be implemented exactly in the quantum setting.
Classical BP reuses a channel observation whenever the corresponding variable
node appears more than once in the computation tree. The analogous quantum
construction would require several copies of the same nonorthogonal channel
output, which is impossible by the no-cloning
theorem~\cite{wootters1982single}. Approximate-cloning and node-merging
approaches have been proposed to handle repeated quantum
messages~\cite{piveteau2022quantum,piveteau2025efficient}, but they make a
rigorous analysis on general LDPC Tanner graphs more difficult. In this work,
we prove vanishing ensemble-average block-error probability for codes drawn
from random $q$-ary LDPC ensembles over symmetric pure-state channels
in the BPQM success region. Such codes arise in coherent-decoding formulations
of Regev's reduction and DQI when the induced parity-check matrices are drawn
from the corresponding ensembles. The BPQM node operations and their adjoints
permit a coherent implementation of the decoder. For these instances, the
vanishing block-error guarantee therefore ensures that the codeword register
is erased with probability approaching one before the subsequent
Fourier-processing step.

\paragraph{Related work on quantum message passing.}

BPQM was introduced in 2017 for decoding binary linear codes over pure-state
channels~\cite{Renes-njp17}.  Later it was shown  that BPQM is simultaneously
optimal for bit- and block-error probability on tree factor graphs and 
the BPQM based decoding complexity can be reduced using a quantum reliability
register~\cite{piveteau2022quantum}. Subsequent work developed density
evolution for binary symmetric classical--quantum
channels~\cite{brandsen2022belief}, BPQM polar decoders achieving the Holevo
bound on binary PSCs~\cite{mandal2023belief,mandal2024polar}, and trellis-based
BPQM constructions for convolutional and turbo
codes~\cite{piveteau2025efficient}. BPQM was then generalized to prime
\(q\)-ary alphabets and abelian
groups~\cite{mandal2026belief,mandal2026qmp}. These generalizations established
the \(q\)-ary node operations, the heralded-mixture message representation,
eigen-list update rules, and fidelity bounds for the resulting channel
combinations. We extend these tools to random-coefficient LDPC ensembles and use them to
establish the block-error results summarized below.

\paragraph{Contributions.}
The main technical contributions are as follows.

\begin{itemize}
   \item Starting from the \(q\)-ary BPQM combining rules and fidelity bounds
   of~\cite{mandal2026belief,mandal2026qmp}, we derive the exact
   density-evolution recursion for random-coefficient regular computation
   trees with variable-node degree \(d_v\) and check-node degree \(d_c\).
   In Theorem~\ref{thm:double-exponential-error-rate}, we use
   coefficient-uniform fidelity bounds to prove that, throughout the BPQM
   success region, the symbol-error probability averaged over the random edge
   coefficients decays doubly exponentially with the number of BPQM
   iterations.
\item We construct a two-stage decoder for codes drawn from the random
   regular LDPC ensemble.
   \begin{enumerate}[label=(\roman*)]
   \item The first stage applies depth-\(\ell\) BPQM to the coordinates whose
      depth-\(\ell\) Tanner-graph neighbourhoods are trees.  With
      \(\ell=\Theta(\log\log N)\), the double-exponential symbol-error bound
      and the noncommutative union bound imply that the ensemble-average
      probability of any BPQM decoding error in this stage tends to zero.
   \item The second stage treats the coordinates with cyclic
      depth-\(\ell\) neighbourhoods as erasures.  With probability tending to
      one, their number is smaller than the minimum distance of the sampled
      code.  The restricted parity-check system therefore has a unique
      solution, which is obtained by Gaussian elimination.
   \end{enumerate}
   In Theorem~\ref{thm:vanishing_block_error_rate}, we combine these two
   stages and prove that the ensemble-average block-error probability of the
   complete decoder tends to zero as \(N\to\infty\) for every channel in the
   BPQM success region.
\item We compute the satisfaction ratios of DQI with the BPQM decoder for
   $D$-regular max-$k$-XORSAT and for $D$-regular max-$k$-LINSAT over
   $\mathbb{F}_q$ with $q\in\{3,5,7,11\}$, from BPQM density-evolution
   thresholds that satisfy the finite-iteration criterion of
   Theorem~\ref{thm:double-exponential-error-rate}. We compare them with DQI
   with classical BP, simulated annealing, Prange's algorithm and Regev's
   reduction with FGUM~\cite{shutty2026lqd}
   (Appendices~\ref{app:dqi-comparison}
   and~\ref{app:dqi-comparison-nonbinary}).
   We make the curious observation that the numerical estimates for the satisfaction ratios of DQI+BPQM and simulated annealing almost perfectly match.

\end{itemize}

The node unitaries and their adjoints also permit a coherent implementation
of the two-stage decoder.  By the principle of deferred measurement, each
intermediate BPQM measurement can be replaced by an isometry that coherently
encodes its outcome in an ancilla register, with subsequent operations
controlled by that register.  The Gaussian-elimination stage can likewise be
implemented as a reversible linear-algebra circuit.  The decoder can then
coherently compute the codeword estimate and subtract it from the codeword
register.  Theorem~\ref{thm:vanishing_block_error_rate} shows that the
ensemble-average probability of obtaining a nonzero value when this register
is checked tends to zero.  This provides the asymptotic block-error guarantee
required when BPQM is used in the decoding step of Regev's reduction and DQI.

Classical block-error analyses often impose a large-girth condition through
expurgation or a suitable graph construction
\cite{gallager1962low,hu2002irregular,pradhan_2016}. In conventional channel
coding, the code can be selected as part of the system design. In
coding-theoretic formulations of Regev's reduction and in DQI, the code is
determined by the computational problem instance, so imposing a large-girth
condition would restrict the class of instances covered by the analysis. The
BPQM decoder considered here works directly with codes drawn from
random LDPC ensembles. It applies BPQM to coordinates with tree neighbourhoods
and recovers the remaining coordinates using the parity constraints.
Appendix~\ref{app:irregular-ldpc} extends the argument to finite-support
irregular ensembles satisfying a linear-distance condition. In that case,
the density-evolution average includes both the random degrees and the
nonzero coefficients. A likelihood-ratio comparison between neighbourhoods
in the finite Tanner graph and the independent computation tree
extends the tree estimate to the finite ensemble.

By Theorem~\ref{thm:vanishing_block_error_rate}, the achievable satisfaction
ratio of DQI+BPQM is directly characterized by these thresholds. DQI+BPQM exceeds DQI with classical BP for every
instance considered, and its satisfaction ratio matches that of simulated
annealing within roughly a percent relative error. This indicates that DQI with the BPQM decoder can match the
performance of simulated annealing on these instances, and that improving
quantum decoding beyond BPQM is a promising route to a quantum advantage.
Since simulated annealing is run on finite instances with $10^6$ sweeps, a
conclusive comparison would require an asymptotic convergence analysis of simulated annealing
(see discussion in Appendix~\ref{app:dqi-discussion}).

\paragraph{Organization.}
The remainder of the paper is organized as follows.
Section~\ref{sec:preliminaries} reviews symmetric $q$-ary PSCs and the
error and fidelity quantities used in the analysis.
Section~\ref{sec:bpqm-tree-factor-graphs} develops BPQM on tree factor graphs,
including the local node operations, heralded-mixture representation, and
density-evolution recursion. Section~\ref{sec:bpqm-ldpc} introduces the
two-stage BPQM and erasure-recovery decoder for random regular LDPC codes.
Section~\ref{sec:double-exponential-decay} proves double-exponential decay of
the tree symbol-error probability, and
Section~\ref{sec:vanishing-block-error} combines this estimate with the local
tree structure and minimum-distance properties of the LDPC ensemble to prove
vanishing block-error probability. Appendices~\ref{app:dqi-comparison}--In~\ref{app:dqi-discussion}, we provide numerical comparison of DQI when paired with BPQM decoding with other classical and quantum benchmark algorithms for max-$k$-XORSAT and max-$k$-LINSAT problems. 

\section{Preliminaries}\label{sec:preliminaries}
For completeness, we recall the results on symmetric \(q\)-ary PSCs used
for the analysis.  Results taken directly from the \(q\)-ary BPQM
paper~\cite{mandal2026belief} are identified by citations in their
statements.  These results include the Gram-matrix characterization of
symmetric \(q\)-ary PSCs, formulas for the error probability of the pretty good
measurement and channel fidelity, the check- and bit-node channel-combining
rules, and the associated fidelity bounds used later.  We restate these results in the
notation of this paper and omit their existing proofs.  
\subsection{Notation}
For a positive integer $m$, we use
$[m]\coloneqq\{0,\ldots,m-1\}$.  Boldface
lowercase letters denote vectors, such as codewords $\bc$ and error vectors
$\be$.  We use $\mzero$
for an all-zero vector or matrix
when its dimension is clear from context.  For an event $\mathcal E$,
$\mone\{\mathcal E\}$ denotes its indicator.  The identity matrix is
denoted by $\mI$. For a set $\mathcal A$, $\mathcal A^m$
denotes the Cartesian product of $m$ copies of $\mathcal A$. For a matrix $M$, the entry in row $i$ and column $j$
is denoted by $M_{i,j}$.  We write $\ket{j}$ for the $j$th standard basis
vector whenever the ambient dimension is clear.

Throughout the paper, we assume that $q$ is prime and identify the input alphabet with $[q]=\{0,\ldots,q-1\}$. We also identify $[q]$ with the finite field $\mathbb{F}_q$ in the standard manner. Unless stated otherwise, all additions, subtractions, multiplications, and inverses of input symbols are over $\mathbb{F}_q$. In Fourier expressions, field elements are represented by their standard representatives in $[q]$.
\subsection{Symmetric Pure State Channels}

\begin{defn}
   A classical--quantum (CQ) channel $W\colon x\mapsto\rho_x$ maps each
   input $x$ in a finite alphabet $\cX$ to a density matrix $\rho_x$.
\end{defn}

\begin{defn}
   Two CQ channels $W$ and $W'$ with the same input alphabet are
   isometrically equivalent if there is an isometry $V$ such that
   $W'(j)=VW(j)V^\dagger$ for every input $j$.  If $V$ is unitary, the
   channels are unitarily equivalent.
\end{defn}

\begin{defn}
   A CQ channel $W\colon j\mapsto\rho_j$ is a pure-state channel (PSC) if
   $\rho_j=\ketbra{\psi_j}{\psi_j}$ for every input $j$.
\end{defn}

For states $\{\ket{\psi_u}\}_{u\in[q]}$, define the Gram matrix by
$G_{i,j}=\braket{\psi_i}{\psi_j}$, and denote its first row by
$[g_0,\dots,g_{q-1}]$.  If $G$ is circulant, then
$G_{i,j}=g_{j-i}$, where the index is evaluated in $\mathbb F_q$ and
represented in $[q]$.

\begin{defn}
   A PSC $W\colon u\mapsto\ketbra{\psi_u}{\psi_u}$ is a symmetric
   $q$-ary PSC if the Gram matrix of
   $\{\ket{\psi_u}\}_{u\in[q]}$ is circulant.
\end{defn}
Let $\omega\coloneqq e^{2\pi\mathrm{i}/q}$ be a primitive
$q$\textsuperscript{th} root of unity.  For all $i,j\in[q]$,
\begin{align}\label{Eq:fourier-orthogonality}
   \sum_{k\in[q]}\omega^{(i-j)k}=q\delta_{i,j}.
\end{align}
For $m\in[q]$, define the Fourier vector $\ket{v_m}$ by
$\braket{j}{v_m}=q^{-1/2}\omega^{jm}$ for every $j\in[q]$.
\eqref{Eq:fourier-orthogonality} implies
$\braket{v_m}{v_{m'}}=\delta_{m,m'}$.

\begin{lem}[{\cite[Lemma~5]{mandal2026belief}}]\label{lem:gram eigenvector}
   For $m\in[q]$, the Fourier vector $\ket{v_{m}}$ is an eigenvector of $G$ satisfying $G \ket{v_m} = \lambda_m \ket{v_m}$ where $\lambda_{m}=\sum_{j\in [q]}g_{j}\omega^{jm}$ $\forall m\in [q]$.
\end{lem}

We refer to the ordered eigenvalue list
$\blambda=[\lambda_0,\dots,\lambda_{q-1}]$ as the \emph{eigen list} of $G$; its
ordering is fixed by the Fourier vectors above.  The corresponding
canonical states are given by
\begin{align}\label{eq:psi-fourier-form}
   \ket{\psi_u} &\coloneqq \frac{1}{\sqrt{q}}\sum_{j\in [q]}\sqrt{\lambda_{j}}\omega^{-uj}\ket{v_{j}}.
\end{align}
For \(d\in[q]\), define
\begin{align}
   U_d
   \coloneqq
   \sum_{j\in[q]}
   \omega^{-dj}\ketbra{v_j}{v_j}.
   \label{eq:symmetric-psc-unitary}
\end{align}
It follows from~\eqref{eq:psi-fourier-form} that
\begin{align}
   U_d\ket{\psi_u}
   &=
   \ket{\psi_{u+d}},\\
   W(u+d)
   &=
   U_dW(u)U_d^\dagger .
   \label{eq:symmetric-psc-covariance}
\end{align}
\begin{lem}[{\cite[Lemma~6]{mandal2026belief}}]\label{lem:gram-trace-relation}
   For Gram matrix $G$ with eigen list $\blambda=[\lambda_0,\dots,\lambda_{q-1}]$, it holds $\sum_{u\in [q]}\lambda_u  =q$.
\end{lem}

It follows from Lemma~\ref{lem:gram-trace-relation} that $\bmu=[\mu_0,\dots,\mu_{q-1}]$, where $\mu_j=\lambda_j/q$ for every $j\in[q]$, is a probability distribution on $[q]$, which we call the \emph{normalized eigen list}.

\begin{lem}[{\cite[Lemma~7]{mandal2026belief}}]\label{lem:canonical state representation}
   All symmetric $q$-ary pure-state channels $W\colon u\rightarrow \ketbra{\psi_{u}}{\psi_u}$ with $u\in [q]$ are determined (up to isometric equivalence) by their eigen list $\blambda$.
\end{lem}

Throughout the rest of the paper, every symmetric $q$-ary PSC is used with
the uniform input distribution; that is, each symbol has probability
$1/q$.
For this input distribution, let \(I(W)\) denote the symmetric Holevo
information of \(W\).

For a CQ channel $W$, the channel fidelity is defined by~\cite{mandal2026belief}
\begin{align*}
   F(W)
   & =\frac{1}{q(q-1)}
   \sum_{\substack{u,u'\in[q]\\u\neq u'}}
   \Tr\left(\sqrt{\sqrt{W(u)}W(u')\sqrt{W(u)}}\right).
\end{align*}
For a symmetric PSC $W$, this reduces to the average absolute pairwise
overlap
\begin{align*}
    F(W)=\frac{1}{q(q-1)}\sum_{u,u'\in[q]:u\neq u'}\left|\braket{\psi_u}{\psi_{u'}}\right|.
\end{align*}

\begin{lem}[{\cite[Lemma~8]{mandal2026belief}}]\label{lem:Holevo information and channel fidelity}
Consider the symmetric  $q$-ary pure-state channel $W\colon u\rightarrow \ketbra{\psi_{u}}{\psi_u}$ for $u\in [q]$ whose Gram matrix $G$ is circulant with eigen list $\blambda=[\lambda_0,\dots,\lambda_{q-1}]$. 
The symmetric Holevo information and channel fidelity satisfy
\begin{align*}
 I(W) & = H(\bmu),\\
 % F(W)  &\  =\frac{1}{q(q-1)}\left(\sum_{u\in [q]}\lambda_{u}^2-q\right)
 F(W) & =\frac{1}{q-1}\sum_{u=1}^{q-1}|g_{u}|,
\end{align*}
where $\bmu=[\mu_0,\dots,\mu_{q-1}]$ is the normalized eigen list with
$\mu_j=\lambda_j/q$ for every $j\in[q]$, and \(H\) denotes the
base-\(2\) entropy.
\end{lem}

For a $q$-ary CQ channel $W\colon j\mapsto\rho_j$ with uniform input,
let $\bar\rho=q^{-1}\sum_{j\in[q]}\rho_j$.  The pretty good measurement
(PGM) (also called the square-root measurement) has operators
\begin{align*}
   M_j
   =
   \frac{1}{q}\bar\rho^{-1/2}\rho_j\bar\rho^{-1/2},
   \qquad j\in[q],
\end{align*}
where the inverse is taken on the support of $\bar\rho$
\cite{hausladen1994pretty,holevo1978asymptotically}. These operators form a
POVM on that support.  Since the output states of a symmetric PSC $W$ are
geometrically uniform, the PGM is optimal for minimum-error state
discrimination~\cite{eldar2002quantum}.
The following lemma expresses this optimal
symbol-error probability in terms of the eigen list of the Gram matrix.
\begin{lem}[{\cite[Lemma~9]{mandal2026belief}}]\label{lem:pgm}
Let $W\colon u\mapsto\ketbra{\psi_u}{\psi_u}$ be a symmetric $q$-ary
PSC with circulant Gram matrix $G$ and eigen list
$\blambda=[\lambda_0,\dots,\lambda_{q-1}]$.  For a uniform input, the
error probability $\perr(W)$ of the PGM satisfies
\begin{align*}
   \perr(W)=1-\left(\frac{1}{q}\sum_{u\in [q]}\sqrt{\lambda_u}\right)^{2}.
\end{align*}
\end{lem}

\begin{lem}\label{lem:pgm-error-fidelity}
Let $W$ be a symmetric $q$-ary PSC with Gram matrix first row
$[1,g_1,\dots,g_{q-1}]$. Then, we have
\begin{align*}
   \frac{F(W)^2}{4}
   \leq \perr(W)
   \leq \min\left\{(q-1)F(W),\,1-\frac{1}{q}\right\}.
\end{align*}
\end{lem}

\begin{proof}
Let $\bmu=[\mu_0,\dots,\mu_{q-1}]$ be the normalized eigen list, where
\(\mu_j=\lambda_j/q\), and set
\(A\coloneqq\sum_{j=0}^{q-1}\sqrt{\mu_j/q}\).  Since
\(\sum_{j=0}^{q-1}\mu_j=1\),
we have \(0\leq A\leq1\), and Lemma~\ref{lem:pgm} gives
\(\perr(W)=1-A^2\leq2(1-A)\).  Furthermore,
\(\lvert\sqrt{x}-\sqrt{y}\rvert^2\leq\lvert x-y\rvert\).  The
inverse Fourier relation gives
\(\mu_j-q^{-1}=q^{-1}\sum_{u=1}^{q-1}g_u\omega^{uj}\).  The triangle
inequality then gives
\begin{align*}
   2(1-A)
   &=\sum_{j=0}^{q-1}
     \left(\sqrt{\mu_j}-q^{-1/2}\right)^2\\
   &\leq\sum_{j=0}^{q-1}\left|\mu_j-q^{-1}\right|
    \leq\sum_{u=1}^{q-1}|g_u|
    =(q-1)F(W).
\end{align*}
For the uniform input distribution, we have
\(\perr(W)\leq1-q^{-1}\).  Conversely,
the Fourier relation
\(g_u=\sum_{j=0}^{q-1}(\mu_j-q^{-1})\omega^{-uj}\) and the triangle
inequality give
\(F(W)\leq\sum_{j=0}^{q-1}|\mu_j-q^{-1}|\).  Factoring each summand and
applying Cauchy--Schwarz yields
\begin{align*}
   F(W)
   &\leq\sum_{j=0}^{q-1}
   \left|\sqrt{\mu_j}-q^{-1/2}\right|
   \left(\sqrt{\mu_j}+q^{-1/2}\right)\\
   &\leq
   \sqrt{2(1-A)}\sqrt{2(1+A)}
   =2\sqrt{1-A^2}
   =2\sqrt{\perr(W)},
\end{align*}
which proves the lower bound.
\end{proof}

\section{BPQM on tree factor graphs}\label{sec:bpqm-tree-factor-graphs}
In this section, we formulate \(q\)-ary BPQM for symbol-wise decoding on a
tree Tanner graph.  For a selected code symbol, BPQM combines the incoming
quantum messages toward its variable node and applies the PGM to the output
of the resulting effective channel.  We recall the local channel-combining
rules from~\cite{mandal2026belief} and their unitary implementation using
check, bit, and multiplication nodes.  We identify heralded mixtures of
symmetric PSCs as a message class preserved by these operations and derive
the density-evolution recursion that characterizes the symbol-error
probability.

A factor graph is a bipartite graph associated with a codeword
\(\bc=(c_0,\ldots,c_{N-1})\).  It represents a function
\(f:[q]^N\rightarrow\mathbb R_{\geq0}\) in the form
\begin{align*}
   f(c_0,\ldots,c_{N-1})=\prod_{S\in\mathcal S}f_S(\bc_S),
\end{align*}
where \(\mathcal S\) is a collection of subsets of \([N]\), \(\bc_S\) is
the restriction of \(\bc\) to the coordinates in \(S\), and \(f_S\) depends
only on \(\bc_S\)~\cite{Kschischangjsac98}.  It contains a variable node for
each \(c_i\), a factor node for each \(f_S\), and an edge joining \(c_i\) to
\(f_S\) whenever \(i\in S\).  The graph describes the local dependence among
the variables and permits the computation of marginals from the factorized
function.
If the graph is a tree, removing an edge separates it into two
components.  Conditioned on the variable associated with that edge, the
factorization separates over these components.  Belief propagation computes
the desired marginal by combining the corresponding messages and this computation is exact on a tree.

For a linear code
\(\cC=\{\bc\in\mathbb F_q^N:H\bc=0\}\), the standard factor-graph
representation is its Tanner graph.
It has one variable node for every code symbol \(c_i\), one check factor
\(f_s\) for every row \(s\) of \(H\), and an edge between \(f_s\) and \(c_i\)
whenever \(H_{s,i}\ne0\).  Define the support of row \(s\) by
\begin{align*}
   S_s\coloneqq\{i\in[N]:H_{s,i}\ne0\}.
\end{align*}
Then \(f_s\) depends on \(\bc_{S_s}\) and enforces
\begin{align*}
   \sum_{i\in S_s}H_{s,i}c_i=0
   \qquad\text{in }\mathbb F_q.
\end{align*}
Thus, \(\partial f_s=S_s\) is the set of variable indices connected to
\(f_s\).
We represent the nonzero coefficient \(H_{s,i}\) by a multiplication
factor \(c\mapsto H_{s,i}c\) on the corresponding edge.  Hence the Tanner
graph used below contains check and multiplication nodes.  Bit-node
combining is a decoding operation: after the incoming check-to-variable
messages for a target symbol have been formed, it combines the resulting channels with the
physical channel output for that symbol.  All arithmetic below is over
\(\mathbb F_q\), identified with \([q]\).

\begin{wrapfigure}{r}{0.27\textwidth}
   \centering
   \vspace{-2mm}
   \begin{tikzpicture}[baseline=(check.base)]
  \node[draw,thick,rectangle,fill=checkcolor,minimum size=7mm,inner sep=0pt] (check) {$+$};
  \draw[thick] (check) -- ++(-1.25,0.65) node[left] {$c_1$};
  \draw[thick] (check) -- ++(-1.25,-0.65) node[left] {$c_2$};
  \draw[thick] (check) -- ++(1.35,0) node[right] {$c_3$};
\end{tikzpicture}
   \vspace{-3mm}
\end{wrapfigure}
A check node represents a parity constraint.  For the degree-three node shown
here,
the three edge variables are \(c_1,c_2,c_3\), and the local constraint is
\begin{align*}
   c_1-c_2=c_3.
\end{align*}
For a fixed \(c_3=l\), the pairs satisfying this constraint are
\((c_1,c_2)=(u,u-l)\), where \(u\in[q]\).  The check-node channel is defined
by averaging the corresponding channel outputs over \(u\).

\begin{wrapfigure}{r}{0.27\textwidth}
   \centering
   \vspace{-2mm}
   \begin{tikzpicture}[baseline=(bit.base)]
  \node[draw,thick,circle,fill=bitcolor,minimum size=7mm,inner sep=0pt] (bit) {$=$};
  \draw[thick] (bit) -- ++(-1.25,0.65) node[left] {$c_1$};
  \draw[thick] (bit) -- ++(-1.25,-0.65) node[left] {$c_2$};
  \draw[thick] (bit) -- ++(1.35,0) node[right] {$c_3$};
\end{tikzpicture}
   \vspace{-3mm}
\end{wrapfigure}
A bit node represents an equality constraint.  For the degree-three node
shown here, the three edge variables satisfy
\begin{align*}
   c_1=c_2=c_3.
\end{align*}
For a fixed \(c_3=l\), both input variables are fixed to
\(c_1=c_2=l\).  At a variable node, this operation combines the physical
channel output with the incoming check messages.

\begin{wrapfigure}{r}{0.27\textwidth}
   \centering
   \vspace{-2mm}
   \begin{tikzpicture}[baseline=(aut.base)]
  \node[draw,thick,isosceles triangle,isosceles triangle apex angle=60,
        fill=infocolor,minimum width=11mm,
        minimum height=12mm,inner sep=1pt] (aut) {$a$};
  \draw[thick] ($(aut.west)+(-1.05,0)$)
        node[left] {$c_1$} -- (aut.west);
  \draw[thick] (aut.apex) -- ++(1.05,0) node[right] {$c_2$};
\end{tikzpicture}
   \vspace{-3mm}
\end{wrapfigure}
A multiplication node is represented by an oriented triangle and applies
multiplication by a nonzero field element from its base to its tip.  For
\(a\in[q]\setminus\{0\}\), the node shown here maps \(c_1\) to \(c_2\), where
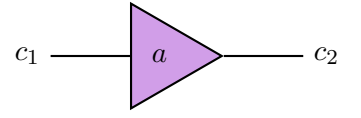
\begin{align*}
   c_2=ac_1.
\end{align*}
In the reverse direction, it maps \(c_2\) to \(c_1=a^{-1}c_2\).
Multiplication nodes describe the nonzero coefficients in a \(q\)-ary parity
check.

Following the check-node convention of~\cite{mandal2026belief}, we use
\(c_1-c_2=c_3\) above so that the associated channel-combining rule has the
form stated below.  Multiplication nodes make this convention equivalent to
the standard coding-theoretic form of a linear parity check.  For example,
\(c_1+c_2+c_3=0\) becomes
\(c_1-(-c_2)=-c_3\), with multiplication by \(-1\) on the second and third
edges.

We next define the channel-combining operations used by the decoder.

\begin{defn}[Check-node combining]\label{def:check node}
  For symmetric \(q\)-ary PSCs \(W_1\) and \(W_2\), the check-node combined
  CQ channel \(W_1\cnop W_2\) is defined, for every \(l\in[q]\), by
  \begin{align*}
      [W_1\cnop W_2](l)
      \coloneqq
      \frac{1}{q}\sum_{u\in[q]}W_1(u)\otimes W_2(u-l).
  \end{align*}
\end{defn}

\begin{defn}[Bit-node combining]\label{def:bit node}
  For symmetric \(q\)-ary PSCs \(W_1\) and \(W_2\), the bit-node combined CQ
  channel \(W_1\vnop W_2\) is defined, for every \(l\in[q]\), by
  \begin{align*}
      [W_1\vnop W_2](l)\coloneqq W_1(l)\otimes W_2(l).
  \end{align*}
\end{defn}

\begin{defn}[Multiplication Node]\label{def:multiplication node}
  For a symmetric \(q\)-ary PSC \(W\) and \(a\in[q]\setminus\{0\}\), the CQ
  channel \(W^{(a)}\) associated with the multiplication node is defined, for
  every \(l\in[q]\), by
  \begin{align*}
      W^{(a)}(l)\coloneqq W(a^{-1}l).
  \end{align*}
\end{defn}

BPQM associates each node operation with a unitary acting on the output
systems of its combined CQ channel.  At a check node, the unitary separates the output into
a symmetric \(q\)-ary PSC and an orthogonal herald register; the herald value
specifies the resulting PSC.  At a bit node, the two outputs conditioned on
the same input symbol are compressed into a symmetric \(q\)-ary PSC and an
input-independent ancilla.  At a multiplication node, a permutation of the
Fourier basis maps the output of \(W^{(a)}\) to its canonical symmetric-PSC
form and no herald register is introduced.  These outputs are instances of the heralded
mixtures defined in Section~\ref{sec:heralded-psc}, with the bit and
multiplication outputs corresponding to trivial heralds.  Up to an isometry,
the check-node output is characterized by its herald probabilities and
conditional eigen lists, whereas the bit- and multiplication-node outputs
are characterized by a single eigen list.  We next recall the check- and
bit-node eigen list updates, given as Lemmas~12 and~13 in the \(q\)-ary BPQM
paper~\cite{mandal2026belief}.  We then state the
multiplication-node update used for nonzero entries of the parity-check
matrix.

\begin{lem}[{\cite[Lemma~12]{mandal2026belief}}]\label{lem:check node}
   Let $W_1$ and $W_2$ be symmetric $q$-ary PSCs with eigen lists
   $\blambda_1=[\lambda_0^{(1)},\dots,\lambda_{q-1}^{(1)}]$ and
   $\blambda_2=[\lambda_0^{(2)},\dots,\lambda_{q-1}^{(2)}]$,
   respectively.  The channel $W_1\cnop W_2$ is isometrically equivalent
   to a heralded ensemble
   $\{p_m^{\cnop},W_m^{\cnop}\}_{m\in[q]}$ of symmetric $q$-ary PSCs.
   If
   $\blambda_m^{\cnop}=[\lambda_0^{(\cnop,m)},\dots,
   \lambda_{q-1}^{(\cnop,m)}]$ is the eigen list of $W_m^{\cnop}$, then
   \begin{align*}
       p_{m}^{\cnop} & =\frac{1}{q^2}\sum_{j\in [q]}\lambda_{(m+j)}^{(1)}\lambda_{-j}^{(2)}\\
       \lambda_{j}^{(\cnop,m)} & = \frac{1}{qp_{m}^{\cnop}}\lambda_{(m+j)}^{(1)}\lambda_{-j}^{(2)}.
   \end{align*}
   For the parity-check equation \(c_1+c_2=l\), the induced channel is
   \(W_1\cnop W_2^{(-1)}\)
   \cite{mandal2026qmp}, where $-1$ denotes edge label multiplied by $-1$.  Denote its herald probabilities and conditional
   eigen lists by
   \(\widetilde p_m^{\cnop}\) and
   \(\widetilde{\blambda}_m^{\cnop}\), respectively.  Then
   \begin{align*}
      \widetilde p_m^{\cnop}
      &=
      \frac{1}{q^2}
      \sum_{j\in[q]}
      \lambda_{m+j}^{(1)}\lambda_j^{(2)},\\
      \widetilde\lambda_j^{(\cnop,m)}
      &=
      \frac{
         \lambda_{m+j}^{(1)}\lambda_j^{(2)}
      }{
         q\widetilde p_m^{\cnop}
      }.
   \end{align*}
\end{lem}

   \begin{lem}[{\cite[Lemma~13]{mandal2026belief}}]\label{lem:bit node}
       Consider symmetric $q$-ary PSCs $W_{1}$ and $W_{2}$ with Gram matrices $G^{(1)}$ and $G^{(2)}$ having eigen lists $\blambda_{1}=[\lambda^{(1)}_{0},\dots,\lambda_{q-1}^{(1)}]$ and $\blambda_{2}=[\lambda^{(2)}_{0},\dots,\lambda_{q-1}^{(2)}]$, respectively. After bit-node combining, the channel $W_{1}\vnop W_{2}$ is isometrically equivalent to a symmetric $q$-ary PSC $W^{\vnop}$ whose Gram matrix has eigen list $\blambda^{\vnop}=[\lambda^{\vnop}_{0},\dots, \lambda^{\vnop}_{q-1}]$, where
       \begin{align*}
           \lambda^{\vnop}_{j}=\frac{1}{q}\sum_{k\in [q]}\lambdaa_{k}\lambdab_{j-k}.
       \end{align*}
   \end{lem}

Now, we describe how the multiplication node in
Definition~\ref{def:multiplication node} acts on the eigen list of a symmetric
\(q\)-ary PSC.

\begin{lem}\label{lem:multiplication-node}
   Let $W\colon u\rightarrow \ketbra{\psi_u}{\psi_u}$ be a symmetric $q$-ary PSC with Gram matrix first row $[g_0,\dots,g_{q-1}]$ and eigen list $\blambda=[\lambda_0,\dots,\lambda_{q-1}]$. For $a\in [q]\setminus\{0\}$, the multiplication channel $W^{(a)}$ is also a symmetric $q$-ary PSC. Its Gram matrix first row is given by
   \begin{align*}
       g^{(a)}_l = g_{a^{-1}l},\qquad l\in [q],
   \end{align*}
   and its eigen list $\blambda^{(a)}=[\lambda^{(a)}_0,\dots,\lambda^{(a)}_{q-1}]$ satisfies
   \begin{align*}
       \lambda^{(a)}_m=\lambda_{am},\qquad m\in [q].
   \end{align*}
   In particular, $\blambda^{(a)}$ is a permutation of $\blambda$, and
   \begin{align*}
      F(W^{(a)})=F(W),
      \qquad
      \perr(W^{(a)})=\perr(W).
   \end{align*}
\end{lem}

\begin{proof}
   The output states of $W^{(a)}$ are $\{\ket{\psi_{a^{-1}l}}\}_{l\in [q]}$. Hence the Gram matrix $G^{(a)}$ of $W^{(a)}$ satisfies
   \begin{align*}
       G^{(a)}_{i,j}
       &=
       \braket{\psi_{a^{-1}i}}{\psi_{a^{-1}j}}\\
       &=
       g_{a^{-1}(j-i)}.
   \end{align*}
   Thus $G^{(a)}$ is circulant, and its first row satisfies $g^{(a)}_l=g_{a^{-1}l}$. By Lemma~\ref{lem:gram eigenvector}, the $m\textsuperscript{th}$ eigenvalue of $G^{(a)}$ is
   \begin{align*}
       \lambda^{(a)}_m
       &=
       \sum_{l\in [q]}g^{(a)}_l\omega^{lm}\\
       &=
       \sum_{l\in [q]}g_{a^{-1}l}\omega^{lm}.
   \end{align*}
   Applying the change of variables $r=a^{-1}l$, equivalently $l=ar$, gives
   \begin{align*}
       \lambda^{(a)}_m
       &=
       \sum_{r\in [q]}g_r\omega^{arm}\\
       &=
       \lambda_{am}.
   \end{align*}
   Since $a\in [q]\setminus\{0\}$, the map $m\mapsto am$ is a permutation of $[q]$. Thus, $\blambda^{(a)}$ is a permutation of $\blambda$.
   The fidelity is unchanged because
   \(g_l^{(a)}=g_{a^{-1}l}\) permutes the nonzero-index terms in its
   definition.  The PGM error formula in Lemma~\ref{lem:pgm} depends on
   \(\blambda\) only through \(\sum_m\sqrt{\lambda_m}\), which is also
   invariant under this permutation.
\end{proof}

\subsection{BPQM unitaries}\label{sec:BPQM_Unitary}
The $q$-ary BPQM unitaries for check- and bit-node combining were
constructed in~\cite{mandal2026belief}.  We recall the required operators
to specify the forward transformations and their adjoints.  Define
\(\widetilde U^{\cnop}\) in the
Fourier basis by
\begin{align}
   \widetilde U^{\cnop}
   \left(\ket{v_j}\otimes\ket{v_{j'}}\right)
   =
   \ket{v_{j+j'}}\otimes\ket{v_{-j'}},
   \qquad j,j'\in[q].
   \label{eq:checknode unitary relation}
\end{align}
All indices in this relation are modulo \(q\).  Let \(F\) be the DFT satisfying
\(\ket{v_j}=F\ket{j}\).  With the symmetric-PSC output in the first
register and the check-node herald in the computational basis of the
second register, the check-node unitary is
\begin{align}
   U^{\cnop}
   =
   (\mI\otimes F^\dagger)
   \operatorname{SWAP}\widetilde U^{\cnop}.
   \label{eq:checknode-unitary}
\end{align}
Thus \(U^{\cnop}\) is independent of the input eigen lists.  For input
eigen lists \(\blambda_1,\blambda_2\), the bit-node unitary is
characterized on the input states by
\begin{align}\label{eq:bitnode unitary relation}
   U^{\vnop}_{\blambda_1,\blambda_2}
   \left(\ket{\psi^{(1)}_u}\otimes\ket{\psi^{(2)}_u}\right)
   =
   \ket{\psi^{\vnop}_u}\otimes\ket{0},
   \qquad u\in[q].
\end{align}
Thus, unlike the check-node unitary, the bit-node unitary depends on the
conditional input eigen lists.  For a multiplication node with
\(a\in[q]\setminus\{0\}\), we use the Fourier-basis permutation
\begin{align}
   U^{\times a}\ket{v_j}
   =
   \ket{v_{a^{-1}j}},
   \qquad j\in[q].
   \label{eq:multiplication-node-unitary}
\end{align}
It satisfies
\begin{align*}
   U^{\times a}\ket{\psi_{a^{-1}l}}
   =\ket{\psi^{(a)}_l},
   \qquad l\in[q].
\end{align*}
It implements the eigen list permutation
\(\lambda_m^{(a)}=\lambda_{am}\), depends only on \(a\), and introduces no
herald. For heralded mixtures, these operations are controlled by the incoming
herald values.  In particular, the bit-node operation is
\begin{align*}
   U^{\vnop}_{\mathrm{control}}
   =
   \sum_{\substack{x_1\in\cX_1\\x_2\in\cX_2}}
   U^{\vnop}_{\blambda_{x_1},\blambda_{x_2}}
   \otimes\ketbra{x_1}{x_1}\otimes\ketbra{x_2}{x_2},
\end{align*}
where \(\cX_1,\cX_2\) are the incoming herald alphabets.  The check- and
multiplication-node operations act on the PSC registers as
\(U^{\cnop}\) and \(U^{\times a}\), respectively, while preserving the
incoming herald registers.  These controlled operations are the local
unitaries used in Section~\ref{sec:operational-tree-bpqm}.

\subsection{Heralded Mixtures of Symmetric PSCs}\label{sec:heralded-psc}
A check-node operation generally produces classical side information
rather than a single symmetric PSC.  A closed message class permits
iteration of the node operations and expresses the reliability of the
resulting channel using finitely many eigen lists.  The
heralded-mixture representation provides this closure and also identifies
the PGM as the final optimal symbol measurement.  Specifically, the BPQM
check-node unitary produces a
symmetric PSC conditioned on an orthogonal classical value.  The
probability of each value and the eigen list of its conditional PSC specify
the message.  Bit and multiplication nodes correspond to the special case
with a single herald value.

\begin{defn}[Heralded mixture class]\label{defn:heralded-mixture-class}
Let $\cM_q$ denote the class of CQ channels $W_H$ that are isometrically
equivalent to a channel of the following form: there is a finite set $\cX$, a
probability distribution $\{p_x\}_{x\in\cX}$, and symmetric $q$-ary PSCs
$\{W_x\}_{x\in\cX}$ such that, for all $j\in[q]$,
\begin{align*}
   W_H(j)=\sum_{x\in\cX}p_x W_x(j)\otimes\ketbra{x}{x}.
\end{align*}
Here $W_x:j\mapsto\ketbra{\psi_j^x}{\psi_j^x}$ has circulant Gram matrix $G_x$ with eigen list $\blambda_x$, and $\{\ket{x}\}_{x\in\cX}$ is an orthonormal basis for the herald register. The distribution $\{p_x\}_{x\in\cX}$ is independent of $j$. We call every channel $W_H\in\cM_q$ a heralded mixture of symmetric $q$-ary PSCs.
\end{defn}
\noindent For such a channel, the channel fidelity $F(W_{H})$ \cite{mandal2026belief} satisfies
   \begin{align*}
     &  F(W_{H}) = \sum_{x\in \cX}p_x F(W_x) = \mathbb{E}_{x}(F(W_{x}))
   \end{align*}
For the canonical representation in~\eqref{eq:psi-fourier-form}, the PGM is
the projective measurement \(\{\Pi_j\}_{j\in[q]}\), where
\begin{align}
   \Pi_j=\ketbra{\Gamma_j}{\Gamma_j},
   \qquad
   \ket{\Gamma_j}
   =\frac{1}{\sqrt q}\sum_{m\in[q]}\omega^{-jm}\ket{v_m}.
   \label{eq:qary-pgm}
\end{align}
Thus, its operators do not depend on the eigen list
\cite[Proof of Lemma~9]{mandal2026belief}.

\begin{lem}
\label{lem:heralded-error-fidelity}
For a channel \(W_H\in\cM_q\) written in the canonical form of
Definition~\ref{defn:heralded-mixture-class} and with uniform input, the
measurement \(\{\Pi_j\otimes\mI\}_{j\in[q]}\) is optimal.  Its error probability
satisfies
\begin{align}
   \perr(W_H)
   =
   \sum_{x\in\cX}p_x\perr(W_x).                    \label{eq:heralded-optimal-error}
\end{align}
Moreover,
\begin{align*}
   \frac{F(W_H)^2}{4}
   \leq
   \perr(W_H)
   \leq
   (q-1)F(W_H).
\end{align*}
For an isometrically equivalent realization
\(W_H'(j)=V W_H(j)V^\dagger\), the corresponding optimal measurement on
the support of \(W_H'\) is
\(\{V(\Pi_j\otimes\mI)V^\dagger\}_{j\in[q]}\).
\end{lem}

\begin{proof}
For any POVM \(\{N_j\}_{j\in[q]}\), define
\(N_{j|x}=(\mI\otimes\bra{x})N_j(\mI\otimes\ket{x})\).
For each \(x\), \(\{N_{j|x}\}_{j\in[q]}\) is a POVM.  Optimality of
\(\{\Pi_j\}\) for every \(W_x\) gives
\begin{align*}
   \frac{1}{q}\sum_{j\in[q]}\Tr\!\left[N_jW_H(j)\right]
   &=
   \sum_{x\in\cX}p_x
   \frac{1}{q}\sum_{j\in[q]}
   \Tr\!\left[N_{j|x}W_x(j)\right] \\
   &\leq
   \sum_{x\in\cX}p_x
   \frac{1}{q}\sum_{j\in[q]}
   \Tr\!\left[\Pi_jW_x(j)\right] \\
   &=
   \frac{1}{q}\sum_{j\in[q]}
   \Tr\!\left[(\Pi_j\otimes\mI)W_H(j)\right].
\end{align*}
Hence, \(\{\Pi_j\otimes\mI\}\) is optimal and
\eqref{eq:heralded-optimal-error} follows.  By
Lemma~\ref{lem:pgm-error-fidelity},
\begin{align*}
   \perr(W_H)
   \leq(q-1)\sum_{x\in\cX}p_xF(W_x)
   =(q-1)F(W_H).
\end{align*}
For the other direction, the lower bound in
Lemma~\ref{lem:pgm-error-fidelity} and Jensen's inequality give
\begin{align*}
   \perr(W_H)
   &=
   \sum_{x\in\cX}p_x\perr(W_x)\\
   &\geq
   \frac{1}{4}\sum_{x\in\cX}p_xF(W_x)^2\\
   &\geq
   \frac{1}{4}
   \left(\sum_{x\in\cX}p_xF(W_x)\right)^2
   =
   \frac{F(W_H)^2}{4}.\qedhere
\end{align*}
\end{proof}

Moreover, the PGM in Lemma~\ref{lem:heralded-error-fidelity} has the same
conditional error probability for every input symbol.  Let
\(\blambda_x=[\lambda_{x,0},\ldots,\lambda_{x,q-1}]\) be the eigen list of
the Gram matrix of \(W_x\).  From~\eqref{eq:psi-fourier-form} and
\eqref{eq:qary-pgm}, for every \(j\in[q]\),
\begin{align}
   1-\Tr\!\left[(\Pi_j\otimes\mI)W_H(j)\right]
   &=
   \sum_{x\in\cX}p_x
   \left[
      1-
      \left(
         \frac{1}{q}\sum_{m\in[q]}\sqrt{\lambda_{x,m}}
      \right)^2
   \right] \notag\\
   &=
   \perr(W_H).
   \label{eq:heralded-symbol-independent-error}
\end{align}
Thus, the optimal average error probability equals the conditional error
probability for each input symbol.

Combining the three node updates gives the closure property required below.
\begin{theorem}[Closure of BPQM on trees]\label{thm:qary-tree-closure}
Let $\cM_q$ be the class of heralded mixtures of symmetric $q$-ary PSCs from Definition~\ref{defn:heralded-mixture-class}. Let $T$ be a tree factor graph whose local message updates are obtained from the check-node, bit-node, and multiplication-node BPQM updates in Lemmas~\ref{lem:check node}, \ref{lem:bit node}, and~\ref{lem:multiplication-node}, composed finitely many times. Fix a directed edge $e$ of $T$. If every leaf message in the computation tree associated with $e$ belongs to $\cM_q$, then the BPQM message produced on $e$ also belongs to $\cM_q$.
Equivalently, every directed message generated by exact BPQM on such a tree is a finite heralded mixture of symmetric $q$-ary PSCs.
For the message associated with a selected symbol, applying the PGM to the
quantum state associated with the corresponding root variable node is
optimal for estimating that symbol.
\end{theorem}

\begin{proof}
We prove the claim by induction on the depth of the computation tree for
$e$.  At depth zero, $e$ is a leaf edge, and the claim holds by assumption.

Assume the claim holds for all directed edges whose computation trees have depth at most $d$. Let $e$ have computation-tree depth $d+1$, and let the message on $e$ be computed from incoming directed edges $e_1,\dots,e_r$. By the induction hypothesis, each incoming message has the form
\begin{align*}
   W_i(l_i)
   =
   \sum_{x_i\in\cX_i}
   p_i(x_i) W_{i,x_i}(l_i)\otimes\ketbra{x_i}{x_i},
   \qquad l_i\in[q],
\end{align*}
where $\cX_i$ is finite and each $W_{i,x_i}$ is a symmetric $q$-ary PSC. Let
\begin{align*}
   \cX=\cX_1\times\cdots\times\cX_r,
   \qquad
   p(x)=\prod_{i=1}^{r}p_i(x_i),
\end{align*}
for $x=(x_1,\dots,x_r)\in\cX$. Conditioned on the herald value $x$, the incoming messages are symmetric $q$-ary PSCs. By Lemmas~\ref{lem:check node}, \ref{lem:bit node}, and~\ref{lem:multiplication-node}, and by finite composition of these local updates, the outgoing message conditioned on $x$ is a finite heralded mixture of symmetric $q$-ary PSCs. Hence, for each $x\in\cX$, there are a finite set $\cY_x$, a probability distribution $q_x$ on $\cY_x$, and symmetric $q$-ary PSCs $\{V_{x,y}\}_{y\in\cY_x}$ such that
\begin{align*}
   \widetilde{W}_{x}(l)
   =
   \sum_{y\in\cY_x}
   q_x(y)V_{x,y}(l)\otimes\ketbra{y}{y},
   \qquad l\in[q].
\end{align*}
The unconditional outgoing message on $e$ is
\begin{align*}
   W_e(l)
   &=
   \sum_{x\in\cX}
   p(x)\widetilde{W}_{x}(l)\otimes\ketbra{x}{x} \\
   &=
   \sum_{x\in\cX}\sum_{y\in\cY_x}
   p(x)q_x(y)V_{x,y}(l)\otimes\ketbra{x,y}{x,y},
   \qquad l\in[q].
\end{align*}
This is a finite heralded mixture of symmetric $q$-ary PSCs, so
\(W_e\in\cM_q\), which completes the induction.  For the message associated
with the selected symbol, write
\begin{align*}
   W_e(l)
   =
   \sum_{h\in\cH}p_hW_h(l)\otimes\ketbra{h}{h},
   \qquad l\in[q].
\end{align*}
The BPQM node unitaries implement the exact channel combinations in
Lemmas~\ref{lem:check node}, \ref{lem:bit node}, and
\ref{lem:multiplication-node}; hence \(W_e\) is the effective CQ channel for
the selected symbol.  Lemma~\ref{lem:heralded-error-fidelity} shows that the
PGM from Lemma~\ref{lem:pgm} on the PSC output register attains its
minimum average error probability.
\end{proof}

\subsection{Operational BPQM on a tree factor graph}
\label{sec:operational-tree-bpqm}
To connect the channel-combining rules with the finite-length decoder used
later, we describe their operational implementation on a $q$-ary tree
factor graph.  BPQM realizes these channel combinations by applying node
unitaries along the graph.  Detailed descriptions of the BPQM message-passing
procedure and its unitary implementation on tree factor graphs are available
in~\cite{Renes-njp17,rengaswamy2021belief,piveteau2022quantum}.  We use the
\(q\)-ary node operations developed
in~\cite{mandal2026belief,mandal2026qmp}.  Related binary and code-specific
constructions appear in
\cite{brandsen2022belief,mandal2023belief,mandal2024polar,piveteau2025efficient}.
Operationally, a BPQM message consists of a qudit output system, any
orthogonal herald registers introduced by check-node operations, and a
classical description of the corresponding conditional channels.
Conditional on the herald values, the qudit is the output of a symmetric
\(q\)-ary PSC specified up to isometric equivalence by its eigen list, as
stated in Lemma~\ref{lem:canonical state representation}.  These eigen
lists determine the subsequent bit-node unitaries.  The herald registers
remain part of the quantum message and are used coherently as controls.
For the factor graph in Fig.~\ref{fig:operational-bpqm-base}, consider
the variable-to-check computation associated with \(c_1\).  As in the
standard computation-tree construction, the variable-to-check message is
the outgoing message, and the other \(d_v-1\) messages at \(c_1\) are
incoming check-to-variable messages.
Normalize each check factor \(f_s\), whose
output is indexed by \(c_k\), as
\begin{align}
   c_k+\sum_{i\in\partial f_s\setminus\{k\}}a_i c_i=0,
   \qquad a_i\in[q]\setminus\{0\}.                            \label{eq:root-check}
\end{align}
The operation \(U^{\times a_i}\) permutes the eigen list indices by
\(\lambda_m\mapsto\lambda_{a_i m}\).  We absorb it into the adjacent
\(U^{\cnop}\).  Figure~\ref{fig:operational-bpqm-base} shows this computation,
with BPQM combining messages from the leaf variables at the bottom toward
the target variable \(c_1\) at the top.

\newcommand{\bpqmtree}[1]{%
  \node[var] (#1x1) at (0,8.0) {\(c_1\)};
  \node[factor] (#1c1) at (-2.3,6.5) {\(f_1\)};
  \node[factor] (#1c2) at (2.3,6.5) {\(f_2\)};

  \node[mult] (#1m2) at (-3.3,5.35) {\(a_2\)};
  \node[mult] (#1m3) at (-1.3,5.35) {\(a_3\)};
  \node[mult] (#1m4) at (1.3,5.35) {\(a_4\)};
  \node[mult] (#1m5) at (3.3,5.35) {\(a_5\)};
  \node[var] (#1x2) at (-3.3,4.2) {\(c_2\)};
  \node[var] (#1x3) at (-1.3,4.2) {\(c_3\)};
  \node[var] (#1x4) at (1.3,4.2) {\(c_4\)};
  \node[var] (#1x5) at (3.3,4.2) {\(c_5\)};

  \node[factor] (#1c3) at (-1.3,2.7) {\(f_3\)};
  \node[factor] (#1c4) at (3.3,2.7) {\(f_4\)};
  \node[mult] (#1m6) at (-2.2,1.55) {\(a_6\)};
  \node[mult] (#1m7) at (-0.4,1.55) {\(a_7\)};
  \node[mult] (#1m8) at (2.4,1.55) {\(a_8\)};
  \node[mult] (#1m9) at (4.2,1.55) {\(a_9\)};
  \node[var] (#1x6) at (-2.2,0.4) {\(c_6\)};
  \node[var] (#1x7) at (-0.4,0.4) {\(c_7\)};
  \node[var] (#1x8) at (2.4,0.4) {\(c_8\)};
  \node[var] (#1x9) at (4.2,0.4) {\(c_9\)};

  \path[edge] (#1x1)--(#1c1);
  \path[edge] (#1x1)--(#1c2);
  \path[edge] (#1c1)--(#1m2);
  \path[edge] (#1m2)--(#1x2);
  \path[edge] (#1c1)--(#1m3);
  \path[edge] (#1m3)--(#1x3);
  \path[edge] (#1c2)--(#1m4);
  \path[edge] (#1m4)--(#1x4);
  \path[edge] (#1c2)--(#1m5);
  \path[edge] (#1m5)--(#1x5);
  \path[edge] (#1x3)--(#1c3);
  \path[edge] (#1x5)--(#1c4);
  \path[edge] (#1c3)--(#1m6);
  \path[edge] (#1m6)--(#1x6);
  \path[edge] (#1c3)--(#1m7);
  \path[edge] (#1m7)--(#1x7);
  \path[edge] (#1c4)--(#1m8);
  \path[edge] (#1m8)--(#1x8);
  \path[edge] (#1c4)--(#1m9);
  \path[edge] (#1m9)--(#1x9);
}

\begin{figure}[H]
   \centering
   \resizebox{0.37\textwidth}{!}{\begin{tikzpicture}[scale=.92,transform shape,font=\small,
 var/.style={draw,thick,circle,fill=bitcolor,minimum size=8mm,inner sep=0pt},
 factor/.style={draw,thick,rectangle,fill=checkcolor,minimum size=7mm,inner sep=0pt},
 mult/.style={draw,thick,isosceles triangle,isosceles triangle apex angle=60,
 shape border rotate=90,fill=infocolor,minimum width=7.2mm,minimum height=6mm,inner sep=.2pt,
 font=\scriptsize},edge/.style={draw,thick}]
\path[use as bounding box] (-4.8,-0.1) rectangle (4.8,8.6);
\bpqmtree{A}
\end{tikzpicture}}
   \caption{Tree factor graph for decoding \(c_1\).  The yellow circles are
   variable nodes, the blue squares are check factors, and each purple triangle
   applies the multiplication \(c_i\mapsto a_ic_i\).  BPQM begins with the
   physical channel outputs at the leaf variable nodes and combines messages
   toward the node associated with \(c_1\).}
   \label{fig:operational-bpqm-base}
\end{figure}
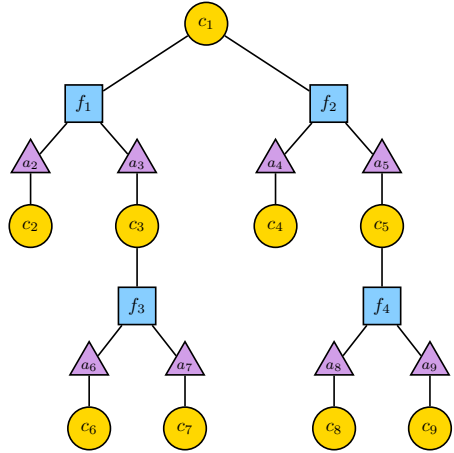

Removing an edge of the tree separates the factor graph into two
components.  The effective channel associated with the message across that
edge is induced by the component not containing \(c_1\), with input given by
the code symbol associated with that edge.  A node operation is applied after
all its incoming messages are available.  Operations on disjoint components
act on disjoint registers and can be applied in parallel.

At each check factor, the check-node unitary combines the incoming systems
from its child edges.  At a variable node, the incoming check-to-variable
channels and the physical channel have the same input code symbol.  BPQM
combines their output systems by successive applications of the bit-node
unitary.  In the notation below, \(\blambda_i\) denotes the eigen list of the
physical channel at \(c_i\), and \(\blambda_{f_s}\) denotes the eigen list of
the conditional PSC associated with the message from \(f_s\).  Concatenated
subscripts denote the eigen list obtained by combining the indicated messages.
For example, the
normalized equation at \(f_3\) is
\begin{align*}
   c_3+a_6c_6+a_7c_7=0.
\end{align*}
Hence, the physical systems for \(c_6,c_7\) are combined into a
\(c_3\)-indexed message.  The corresponding systems at \(f_4\) are combined
into a \(c_5\)-indexed message.  These messages are combined with the
physical systems for \(c_3,c_5\) using
\(U^{\vnop}_{\blambda_3,\blambda_{f_3}}\) and
\(U^{\vnop}_{\blambda_5,\blambda_{f_4}}\), respectively, with the herald
values selecting the conditional eigen lists.  Conditioned on these values,
each bit-node output is a symmetric \(q\)-ary PSC.

The resulting messages enter the check operations at \(f_1,f_2\), as shown
in Fig.~\ref{fig:operational-bpqm-forward}.  After applying
\(U^{\cnop}\) at these checks, both outputs are indexed by \(c_1\).  They
are combined by \(U^{\vnop}_{\blambda_{f_1},\blambda_{f_2}}\), followed by
\(U^{\vnop}_{\blambda_1,\blambda_{f_1f_2}}\) with the physical system for
\(c_1\).  These are the controlled node operations described in
Section~\ref{sec:BPQM_Unitary}.

\begin{figure}[H]
   \centering
   \resizebox{0.37\textwidth}{!}{\begin{tikzpicture}[scale=.92,transform shape,font=\small,
 var/.style={draw,thick,circle,fill=bitcolor,minimum size=8mm,inner sep=0pt},
 factor/.style={draw,thick,rectangle,fill=checkcolor,minimum size=7mm,inner sep=0pt},
 mult/.style={draw,thick,isosceles triangle,isosceles triangle apex angle=60,
 shape border rotate=90,fill=infocolor,minimum width=7.2mm,minimum height=6mm,inner sep=.2pt,
 font=\scriptsize},edge/.style={draw,thick},
 up/.style={->,line width=1.5pt,red!75!black},
 active/.style={draw=red!75!black,thick,rounded corners=3pt,inner sep=6pt}]
\path[use as bounding box] (-4.8,-0.1) rectangle (4.8,8.6);
\bpqmtree{D}
\node[active,fit=(Dc1)(Dm2)(Dm3)(Dx2)(Dx3)] {};
\node[active,fit=(Dc2)(Dm4)(Dm5)(Dx4)(Dx5)] {};
\draw[up] (Dc1.north)--(Dx1.south west);
\draw[up] (Dc2.north)--(Dx1.south east);
\end{tikzpicture}}
   \caption{Check-node stage at \(f_1\) and \(f_2\).  First, the incoming
   messages from the lower subtrees are combined with the physical systems
   at \(c_3\) and \(c_5\) by the corresponding bit-node unitaries.  Second,
   the multiplication permutations are incorporated into the incoming
   channels at \(f_1\) and \(f_2\).  Third, \(U^{\cnop}\) is applied at both
   checks in parallel.  The red arrows show the two resulting messages,
   each indexed by \(c_1\).}
   \label{fig:operational-bpqm-forward}
\end{figure}
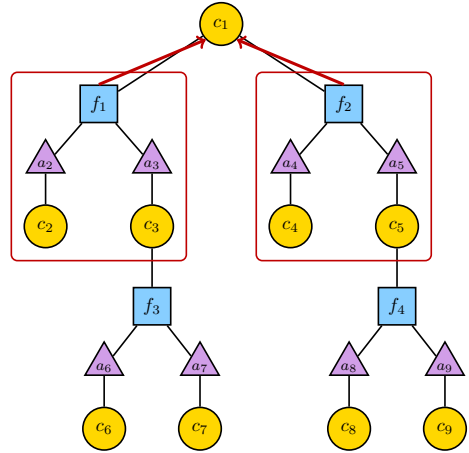

The final bit-node operation
\(U^{\vnop}_{\blambda_1,\blambda_{f_1f_2}}\) includes the physical system
for \(c_1\).  The PGM of Lemma~\ref{lem:heralded-error-fidelity} is then
applied to the final qudit, as shown in
Fig.~\ref{fig:operational-bpqm-measurement}.  Since the final effective channel
belongs to \(\cM_q\), this measurement is optimal and acts only on the final
qudit.  The herald registers are retained together with the remaining output
systems.  After the measurement outcome is obtained, these systems are the
inputs to the reverse circuit.

To decode another symbol, apply the node unitaries in reverse order with
each replaced by its adjoint.  The order is
\(\bigl(U^{\vnop}_{\blambda_1,\blambda_{f_1f_2}}\bigr)^\dagger\),
\(\bigl(U^{\vnop}_{\blambda_{f_1},\blambda_{f_2}}\bigr)^\dagger\), the checks
\(f_1,f_2\), the bit nodes at \(c_3,c_5\), and the checks
\(f_3,f_4\).  Adjoint operations within each pair can be applied in
parallel.  This undoes the forward basis transformation and restores the
register layout used for the next symbol.  It does not undo the disturbance
caused by the PGM measurement.

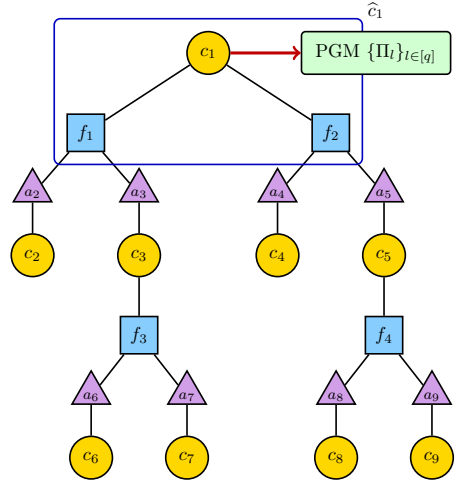
\begin{figure}[H]
   \centering
   \resizebox{0.37\textwidth}{!}{\begin{tikzpicture}[scale=.92,transform shape,font=\small,
 var/.style={draw,thick,circle,fill=bitcolor,minimum size=8mm,inner sep=0pt},
 factor/.style={draw,thick,rectangle,fill=checkcolor,minimum size=7mm,inner sep=0pt},
 mult/.style={draw,thick,isosceles triangle,isosceles triangle apex angle=60,
 shape border rotate=90,fill=infocolor,minimum width=7.2mm,minimum height=6mm,inner sep=.2pt,
 font=\scriptsize},edge/.style={draw,thick},
 up/.style={->,line width=1.5pt,red!75!black},
 active/.style={draw=blue!75!black,thick,rounded corners=3pt,inner sep=6pt}]
\path[use as bounding box] (-4.8,-0.1) rectangle (4.8,8.6);
\bpqmtree{F}
\node[active,fit=(Fx1)(Fc1)(Fc2)] {};
\node[draw,thick,rounded corners=2pt,fill=green!18,
 minimum width=28mm,minimum height=8mm] (pgm) at (3.15,8.0)
 {PGM \(\{\Pi_l\}_{l\in[q]}\)};
\draw[up] (Fx1.east)--(pgm.west);
\node[above=1mm of pgm] {\(\widehat{c}_1\)};
\end{tikzpicture}}
   \caption{Final operations for estimating \(c_1\).  First,
   \(U^{\vnop}_{\blambda_{f_1},\blambda_{f_2}}\) combines the messages from
   \(f_1\) and \(f_2\).  Second,
   \(U^{\vnop}_{\blambda_1,\blambda_{f_1f_2}}\) combines the result with the
   physical channel output at \(c_1\).  Finally, the PGM is applied to the
   retained qudit to obtain \(\widehat c_1\).}
   \label{fig:operational-bpqm-measurement}
\end{figure}

\subsection{Density evolution on a regular tree}
\label{sec:bpqm-density-evolution}
Density evolution tracks the distribution of messages under the
computation-tree recursion.  The incoming branches are independent, so the
check-node and variable-node updates determine the message distribution at
the next iteration.  Iterating these updates gives the asymptotic
symbol-error probability and identifies the BPQM success region in the
space of channels.

Consider a $(d_v,d_c)$-regular computation tree with physical channel $W$.
An outgoing message from a check node combines its $d_c-1$ incoming
variable-to-check messages.  The root variable node combines the physical
channel with its $d_v-1$ incoming check-to-variable messages.  With general
nonzero edge coefficients, the effective channels on these incoming
branches need not be identical because their descendant computation trees
can have different edge coefficients.  To retain this dependence, orient
every edge of the computation tree toward its root.  At each check
node, let $a_{\mathrm{out}}$ be the coefficient on the outgoing edge toward
the root and let $a_s$ be the coefficient on the $s$th incoming edge for
$s\in[d_c-1]$.  Dividing the parity-check equation by
$a_{\mathrm{out}}$ assigns the relative coefficient
$a_{\mathrm{out}}^{-1}a_s\in[q]\setminus\{0\}$ to the $s$th incoming
edge.  These relative coefficients label the incoming edges of the check
nodes in the depth-$t$ computation tree associated with a variable-to-check
message.

\begin{defn}[Coefficient trajectory]
\label{defn:coefficient-trajectory}
For each $t\geq0$, let $\cA_t$ denote the set of all assignments of
normalized nonzero edge labels to the incoming edges of the check
nodes in the depth-$t$ computation tree associated with a variable-to-check
message after $t$ BPQM iterations on a $(d_v,d_c)$-regular tree.  An element
$\ba_t\in\cA_t$ is one such edge-label assignment and is called a coefficient
trajectory.  At depth zero, no check-node update occurs, so set
$\cA_0\coloneqq\{\emptyset\}$.  For $t\geq0$, define
\begin{align*}
\cA_{t+1}
\coloneqq
\left(
   \left(
      \bigl([q]\setminus\{0\}\bigr)\times\cA_t
   \right)^{d_c-1}
\right)^{d_v-1}.
\end{align*}
The outer Cartesian power corresponds to the $d_v-1$ check-to-variable
messages entering the root variable node.  The inner Cartesian power
corresponds to the $d_c-1$ incoming variable-to-check messages at each
check node.
An element $\ba_{t+1}\in\cA_{t+1}$ is written as
\begin{align*}
\ba_{t+1}
&=
\left[
   \ba_{t+1,0},\ldots,\ba_{t+1,d_v-2}
\right],
\end{align*}
where, for each $r\in[d_v-1]$,
\begin{align*}
\ba_{t+1,r}
&=
\left[
   \left(a_{t+1,r,0},\ba_t^{(r,0)}\right),
   \ldots,
\left(a_{t+1,r,d_c-2},\ba_t^{(r,d_c-2)}\right)
\right].
\end{align*}
Here $r$ indexes one of the $d_v-1$ check-to-variable messages entering
the root variable node, and $s$ indexes one of the $d_c-1$ incoming
variable-to-check messages at the corresponding check node.  The first
component $a_{t+1,r,s}\in[q]\setminus\{0\}$ is the relative coefficient
defined above.  The second component $\ba_t^{(r,s)}\in\cA_t$ is the
label array of the depth-$t$ descendant subtree that produces the same
message.

An element of $\cA_1$ can therefore be identified with a
$(d_v-1)\times(d_c-1)$ array of normalized nonzero coefficients. At
greater depths, each array entry also specifies the coefficient
labels on its descendant subtree.
\end{defn}

For fixed $t$, sample independently and uniformly from
$[q]\setminus\{0\}$ all edge coefficients encountered during the first
$t$ BPQM iterations.  Define the random variable $\bA_t$ to be the nested
array obtained by replacing, at every check node, each incoming coefficient
$a_s$ with the relative coefficient $a_{\mathrm{out}}^{-1}a_s$ and arranging
these relative coefficients according to the recursion defining $\cA_t$.
Then $\bA_t$ takes values in $\cA_t$, and $\bA_0\coloneqq\emptyset$.

Write $A_{t+1,r,s}$ and $\bA_t^{(r,s)}$ for the corresponding components
of $\bA_{t+1}$.  At each check node, the map
\begin{align*}
\left(a_{\mathrm{out}},(a_s)_{s\in[d_c-1]}\right)
\longmapsto
\left(a_{\mathrm{out}},
\left(a_{\mathrm{out}}^{-1}a_s\right)_{s\in[d_c-1]}\right)
\end{align*}
is a bijection on
$\bigl([q]\setminus\{0\}\bigr)^{d_c}$.  Hence the relative coefficients
$\{A_{t+1,r,s}\}$ are independent and uniform on
$[q]\setminus\{0\}$.  The descendant branches are disjoint, so
$\{\bA_t^{(r,s)}\}$ are independent copies of $\bA_t$ and are independent
of these relative coefficients.  It follows by induction that $\bA_t$ is
uniform on $\cA_t$.  In particular, $\cA_t$ is the support of $\bA_t$.

For a fixed coefficient trajectory $\ba_t\in\cA_t$, let
$W_t^{\ba_t}$ denote the effective channel for estimating the root symbol
after $t$ rounds of BPQM.  For the unique depth-zero trajectory, set
\(W_0^{\emptyset}\coloneqq W\).
Suppose that $W_t^{\ba_t}$ has been defined for every
$\ba_t\in\cA_t$. Fix $t\geq0$, $\ba_{t+1}\in\cA_{t+1}$, and
$r\in[d_v-1]$. On branch $(r,s)$, the incoming variable-to-check channel
is the multiplication channel
$\left(W_t^{\ba_t^{(r,s)}}\right)^{(a_{t+1,r,s})}$ from
Definition~\ref{def:multiplication node}, with input
$u_s=a_{t+1,r,s}c_s$, where $c_s$ is the unscaled branch symbol.
Conditional on the symbol $c$ on the edge directed toward the root, the
$r$th check-to-variable channel is
\begin{equation*}
\begin{aligned}
&W_{t+1,r}^{\cnop\ba_{t+1,r}}(c)\\
&\quad\coloneqq
\frac{1}{q^{d_c-2}}
\sum_{\substack{u_0,\ldots,u_{d_c-2}\in[q]\\
c+\sum_{s=0}^{d_c-2}u_s=0}}
\bigotimes_{s=0}^{d_c-2}
\left(W_t^{\ba_t^{(r,s)}}\right)^{(a_{t+1,r,s})}(u_s),
\qquad c\in[q].
\end{aligned}
\end{equation*}
The constraint in the sum is the normalized parity-check equation
$c+\sum_{s=0}^{d_c-2}a_{t+1,r,s}c_s=0$, written using
$u_s=a_{t+1,r,s}c_s$.
For each $c\in[q]$, there are $q^{d_c-2}$ admissible tuples
$(u_0,\ldots,u_{d_c-2})$, and the factor $q^{-(d_c-2)}$ is the uniform
average over these tuples. The tensor product combines the channel outputs
from the disjoint descendant subtrees.
It is implemented by a sequence of the binary check-node and multiplication
operations in Lemma~\ref{lem:check node} and
Definition~\ref{def:multiplication node}.
The depth-$(t+1)$ effective channel associated with $\ba_{t+1}$ is
\begin{align*}
W_{t+1}^{\ba_{t+1}}
\coloneqq{}&
W \vnop
W_{t+1,0}^{\cnop\ba_{t+1,0}}
\vnop
\cdots
\vnop
W_{t+1,d_v-2}^{\cnop\ba_{t+1,d_v-2}}.
\end{align*}
% Removed line break -HP
Consequently, $W_t^{\bA_t}$ is the random effective channel for estimating
the root symbol after $t$ rounds of BPQM.  To display the dependence on the
physical channel, define
its average fidelity and symbol-error probability by
\begin{align*}
F_t(W)
&\coloneqq
\E_{\bA_t}\!\left[F\!\left(W_t^{\bA_t}\right)\right],\\
P_t(W)
&\coloneqq
\E_{\bA_t}\!\left[\perr\!\left(W_t^{\bA_t}\right)\right].
\end{align*}
When \(W\) is fixed, we write \(F_t\) and \(P_t\) for these quantities.
For every $\ba_t\in\cA_t$, Theorem~\ref{thm:qary-tree-closure} gives
$W_t^{\ba_t}\in\cM_q$, and Lemma~\ref{lem:heralded-error-fidelity} gives
\begin{align*}
\frac{
F\left(
W_t^{\ba_t}
\right)^2
}{4}\leq \perr\left(
W_t^{\ba_t}
\right)\ \leq (q-1) F\left(W_t^{\ba_t}
\right).
\end{align*}
Averaging over $\bA_t$ and applying Jensen's inequality to the lower
bound yields
\begin{align}\label{Eq:average-fidelity-error-bounds}
\frac{F_t^2}{4}
&\leq
P_t
\leq
(q-1)F_t.
\end{align}

For the $(d_v,d_c)$-regular computation tree with random nonzero edge
coefficients, define the \emph{BPQM success region}
\begin{align}
\mathrm{Reg}_{\mathrm{BPQM}}(d_v,d_c)
\coloneqq
\left\{
W:\,
\begin{array}{l}
W\text{ is a symmetric }q\text{-ary PSC},\\[1mm]
\displaystyle\lim_{t\to\infty}F_t(W)=0,\quad
\displaystyle\lim_{t\to\infty}P_t(W)=0
\end{array}
\right\}.
\label{eq:regular-bpqm-success-region}
\end{align}
The bounds in~\eqref{Eq:average-fidelity-error-bounds} show that the two
convergence conditions in the definition are equivalent.  Thus either
condition can be used to verify membership in
\(\mathrm{Reg}_{\mathrm{BPQM}}(d_v,d_c)\).  This ensemble-dependent
success region does not require a scalar ordering of the \(q\)-ary PSCs.

Appendix~\ref{app:numerical-bpqm-success-region} gives a numerical
visualization of \(\mathrm{Reg}_{\mathrm{BPQM}}(3,12)\) for symmetric
ternary PSCs.

\section{BPQM Decoding of Regular LDPC Codes}
\label{sec:bpqm-ldpc}
In this section, we extend the tree-level BPQM construction to finite
regular LDPC codes.  We first explain why the usual computation-tree
interpretation of classical BP does not directly define a quantum decoder
on a Tanner graph with cycles.  We then introduce a two-stage decoder that
applies exact BPQM on cycle-free neighbourhoods and recovers the remaining
coordinates from the parity constraints.

\subsection{Regular LDPC codes}

Let \(\cG=(V,C,E)\) be a bipartite Tanner graph with variable-node set
\(V=[N]\) and check-node set \(C=[m]\).  The socket construction below
produces a bipartite multigraph, so distinct edges can connect the same
variable node and check node.  Each edge \(e\in E\) connecting variable
node \(i\) and check node \(s\) has a coefficient
\(a_e\in[q]\setminus\{0\}\).  The corresponding entry of the parity-check
matrix is the sum of the coefficients of all edges connecting these two
nodes:
\begin{align}
   H_{s,i}
   \coloneqq
   \sum_{\substack{e\in E:\\e\text{ joins }i\text{ and }s}}a_e.
   \label{eq:aggregated-parity-check}
\end{align}
The associated code is
\begin{align}
   \cC=\left\{\bc\in\mathbb F_q^N:H\bc=0\right\}.
   \label{eq:regular-ldpc-code}
\end{align}
The Tanner graph is \((d_v,d_c)\)-regular if every variable node
has degree \(d_v\) and every check node has degree \(d_c\), with parallel
edges counted according to their multiplicity.  Attach \(d_v\) sockets to every
variable node and \(d_c\) sockets to every check node, and match the
\(Nd_v\) sockets on each side uniformly at random.  Independently assign
to every matched pair a coefficient sampled uniformly from
\([q]\setminus\{0\}\).  Hence
\(m=Nd_v/d_c\) and
\begin{align*}
   R=1-\frac{d_v}{d_c}.
\end{align*}
Thus, \(R\) is the design rate.  We consider ensembles with positive
design rate, so \(d_v<d_c\).
This is the standard regular LDPC construction
\cite{gallager1962low,richardson2008modern}.

BPQM is exact when the Tanner graph is a tree.  A regular LDPC graph typically
contains cycles, so its node operations do not in general decompose into a
single tree circuit.

\subsection{Cycles and approximate cloning}
Classical BP is applied to a Tanner graph with cycles by iterating the same
local message updates used on a tree.  After a fixed number of iterations,
each outgoing message is represented by a computation tree.  If the same
Tanner-graph variable occurs at several vertices of this tree, its channel
likelihood is used at every occurrence
\cite{richardson2008modern}.

A direct BPQM analogue would supply the single quantum channel output
associated with that variable to several branches of the computation tree.
For nonorthogonal channel outputs, the no-cloning theorem precludes this
operation~\cite{wootters1982single}.  In~\cite{piveteau2022quantum}, an adjoint
equality-node operation instead maps a repeated binary PSC output to
approximate product outputs with modified channel parameters.  Reversing
the operations after each BPQM measurement permits sequential decoding of
the remaining symbols.  The number and quality of the approximate outputs
depend on the occurrences of the variable in the computation tree.

A different construction is given in~\cite{piveteau2025efficient} for a
factor graph with a single cycle.  The variables on the cycle are merged
into larger factor nodes, producing an acyclic graph with enlarged message
alphabets.  On a general LDPC Tanner graph, approximate cloning changes the
channels supplied to the computation tree, whereas repeated node merging
increases the message dimension as the cyclic components grow.  These
constructions do not yield the exact tree-channel recursion used in the
BPQM density-evolution analysis.  This motivates a BPQM-based LDPC decoder
whose block-error probability can be related directly to that recursion.

\subsection{BPQM with erasure recovery}
In this work, we propose a two-stage BPQM-based LDPC decoder.  In the first
stage, the decoder applies BPQM to coordinates whose depth-$\ell$
neighbourhoods are cycle-free.  In the second stage, it treats the remaining
coordinates as erasure locations, substitutes the BPQM estimates into the
parity-check equations, and recovers the omitted values by Gaussian
elimination.  This construction separates errors in the local quantum
measurements from failures of the erasure-recovery step.

Fix the computation depth \(\ell\), and fix one outgoing edge $e_i$ at
every variable node $i\in[N]$.  Let \(\cN_\ell(i)\) be the
depth-\(\ell\) computation graph associated with the variable-to-check
message along $e_i$.  At variable node \(i\), this message combines the
physical observation $W(c_i)$ with the \(d_v-1\) incoming
check-to-variable messages along its remaining edges.
Edge multiplicities are inherited from \(\cG\); in
particular, two parallel edges form a cycle.  Define good and bad sets of symbols
\begin{align}
   G_\ell
   \triangleq
   \{i\in[N]:\cN_\ell(i)\text{ is a tree}\},\qquad
   B_\ell
   \triangleq
   [N]\setminus G_\ell .
   \label{eq:good-bad-coordinates}
\end{align}
In the first stage, the decoder estimates the coordinates in \(G_\ell\)
sequentially.  For each \(i\in G_\ell\), it applies the depth-\(\ell\)
BPQM circuit on \(\cN_\ell(i)\), measures the final qudit with the PGM to
obtain \(\hat c_i\), and applies the adjoints of the BPQM node unitaries in
reverse order before proceeding to the next coordinate.  This procedure ensures that
each physical channel output is used only once within any BPQM
computation, while channel outputs shared by different neighbourhoods are
available for successive computations.

For $i\in G_\ell$, consider the messages in $\cN_\ell(i)$ directed toward
variable node $i$ and normalize every
check equation so
that the coefficient of the variable associated with its outgoing
message is one.  The resulting
multiplication factors define a coefficient trajectory
$\ba_{\ell,i}\in\cA_\ell$ in the sense of
Definition~\ref{defn:coefficient-trajectory}.  The exact effective channel
used by the depth-$\ell$ BPQM circuit to estimate $c_i$ is
$W_\ell^{\ba_{\ell,i}}$.

In the second stage, the coordinates in \(B_\ell\) are treated as
erasures.  Let
\(H_{G_\ell}\) and \(H_{B_\ell}\) denote the submatrices of \(H\) formed
by the columns indexed by \(G_\ell\) and \(B_\ell\), respectively.  After
BPQM has been applied to every coordinate in \(G_\ell\), the remaining
symbols are determined by solving
\begin{align}
   H_{B_\ell}\bc_{B_\ell}
   =
   -H_{G_\ell}\widehat{\bc}_{G_\ell}.
   \label{eq:bpqm-ge-recovery}
\end{align}
If this system has a unique solution and the BPQM estimates on \(G_\ell\)
are correct, Gaussian elimination recovers the transmitted codeword.
In Section~\ref{sec:vanishing-block-error}, we show that, for every channel
\(W\in\mathrm{Reg}_{\mathrm{BPQM}}(d_v,d_c)\), the depth \(\ell\) can be chosen as a
function of \(N\) so that the ensemble-average block-error probability of
this decoder vanishes as \(N\to\infty\).

\subsection{Codeword Independence of local BPQM measurements}

The block-error analysis uses the error probability of the sequential BPQM
decoder for an arbitrary transmitted codeword.  The following covariance
property shows that the symbol error probabilities for BPQM operations, and hence the
block-error probability, are independent of this codeword.  It follows from
the channel symmetry and its preservation by the BPQM node operations.

\begin{lem}\label{lem:local-bpqm-projector-covariance}
Let \(i\in G_\ell\).  Let \(J_i\subseteq[N]\)
be the indices of the variable nodes in \(\cN_\ell(i)\), let
\(U_{\ell,i}\) be the depth-\(\ell\) BPQM unitary, and define its success
projector by
\begin{align}
   \Pi_i(c)
   \coloneqq
   U_{\ell,i}^{\dagger}
   \left(\Pi_c\otimes\mI\right)
   U_{\ell,i},
   \qquad c\in[q].
   \label{eq:local-bpqm-success-projector}
\end{align}
Let
\(\bc_{J_i}=(c_j)_{j\in J_i}\) and
\(\bc'_{J_i}=(c_j')_{j\in J_i}\) be two vectors in
\(\mathbb F_q^{J_i}\).  Suppose that each vector satisfies every
parity-check equation whose check node belongs to \(\cN_\ell(i)\), and set
\(d_j\coloneqq c_j'-c_j\) for \(j\in J_i\).
Then
\begin{align}
&\left(\bigotimes_{j\in J_i}U_{d_j}^{\dagger}\right)
\Pi_i(c_i')
\left(\bigotimes_{j\in J_i}U_{d_j}\right)
=
\Pi_i(c_i).
\label{eq:local-bpqm-projector-covariance}
\end{align}
The projectors and tensor products in
\eqref{eq:local-bpqm-projector-covariance} act as the identity on the
channel-output systems outside \(J_i\).
\end{lem}

\begin{proof}
\eqref{eq:symmetric-psc-unitary} and
\eqref{eq:multiplication-node-unitary} give
\begin{align*}
   U^{\times a}U_d
   &=
   U_{ad}U^{\times a}.
\end{align*}
From~\eqref{eq:checknode unitary relation} and
\eqref{eq:checknode-unitary}, direct evaluation on
\(\ket{v_j}\otimes\ket{v_{j'}}\) gives
\begin{align*}
   U^{\cnop}
   \left(U_{d_1}\otimes U_{d_2}\right)
   &=
   \left(U_{d_1-d_2}\otimes F^\dagger U_{d_1}F\right)U^{\cnop}.
\end{align*}
For every pair of conditional eigen lists, the bit-node unitary
constructed in~\cite{mandal2026belief} satisfies
\begin{align*}
   U^{\vnop}_{\blambda_1,\blambda_2}
   \left(U_d\otimes U_d\right)
   &=
   \left(U_d\otimes\mI\right)
   U^{\vnop}_{\blambda_1,\blambda_2}.
\end{align*}
Thus, although the bit-node unitary is selected by the incoming herald
values, its covariance action is $U_d\otimes U_d\mapsto U_d\otimes\mI$
for every conditional pair and does not depend on the herald registers.
The differences at every check satisfy the homogeneous parity-check
equation obtained by subtracting the equations for \(\bc_{J_i}\) and
\(\bc'_{J_i}\), while the differences entering a bit-node operation are equal.
The operators \(F^\dagger U_dF\) are diagonal in the check-herald basis
and commute with the subsequent herald-controlled bit-node operations.
Composing the three displayed identities over \(\cN_\ell(i)\) maps the
final qudit by \(U_{d_i}\), while all other resulting unitaries act on
registers on which the success projector is the identity.  Finally,
\eqref{eq:qary-pgm} and~\eqref{eq:symmetric-psc-unitary} give
\begin{align*}
   U_{d_i}^{\dagger}\Pi_{c_i'}U_{d_i}
   =
   \Pi_{c_i}.
\end{align*}
Using~\eqref{eq:local-bpqm-success-projector} proves
\eqref{eq:local-bpqm-projector-covariance}.
\end{proof}

If \(c_i'=c_i\), then \(d_i=0\), and
\eqref{eq:local-bpqm-projector-covariance} shows that the success
probability is the same for all vectors on \(J_i\) that satisfy the local
parity checks and have value \(c_i\) at variable node \(i\).

For codewords \(\bc,\bc'\in\cC\), set \(d_j=c_j'-c_j\) for
\(j\in[N]\).  Since \(\Pi_i(c)\) acts as the identity outside \(J_i\),
\eqref{eq:local-bpqm-projector-covariance} gives
\begin{align*}
&\left(\bigotimes_{j=1}^{N}U_{d_j}^{\dagger}\right)
\Pi_i(c_i')
\left(\bigotimes_{j=1}^{N}U_{d_j}\right)
=
\Pi_i(c_i)
\end{align*}
for every \(i\).  Moreover,~\eqref{eq:symmetric-psc-covariance} gives
\begin{align*}
   \bigotimes_{j=1}^{N}W(c_j')
   =
   \left(\bigotimes_{j=1}^{N}U_{d_j}\right)
   \left(\bigotimes_{j=1}^{N}W(c_j)\right)
   \left(\bigotimes_{j=1}^{N}U_{d_j}^{\dagger}\right).
\end{align*}
Substituting these covariance identities into the sequential measurement
probability shows that the decoder has the same block-error probability for
every transmitted codeword.

\section{Double-Exponential Decay of the Symbol-Error Probability on a Tree}\label{sec:double-exponential-decay}
In this section, we derive a quantitative convergence rate for the
symbol-error probability of BPQM on a regular computation tree.  We extend
the one-step channel-fidelity bounds of~\cite{mandal2026belief} to the
heralded messages and node degrees in the LDPC recursion.  The resulting
scalar inequality proves double-exponential decay of the symbol-error
probability throughout the BPQM success region.
\subsection{Fidelity Bounds}\label{sec:fidelity-bounds}
While exact density evolution propagates an entire distribution of
conditional eigen lists, a scalar bound is more convenient because the
number of conditional eigen lists grows exponentially with depth.
Extending the channel-fidelity inequalities to heralded messages and
arbitrary node degrees gives the recursion used below.
Bhattacharyya parameters control the asymptotic analysis of message passing
on classical channels, while channel fidelity plays the analogous role for
CQ channels~\cite{sason2006performance,wilde2012polar}.  In Lemma \ref{lem:fidelity check and bit node bounds}, we restate the fidelity bounds for check and bit node combined channels obtained in~\cite{mandal2026belief}. Later in Lemmas
\ref{lem fidelity bound for heralded PSC} and
\ref{lem:heralded PSC fidelity bound}, we lift these bounds to heralded mixtures
and arbitrary node degrees.  Theorem~\ref{thm:double-exponential-error-rate}
then uses the resulting recursion to establish double-exponential decay for the symbol error probability for the root node.
   \begin{lem}[{\cite[Lemma~16]{mandal2026belief}}]\label{lem:fidelity check and bit node bounds}
         Let $W_1$ and $W_2$ be symmetric $q$-ary PSCs with eigen lists
         $\blambda_1=[\lambda_0^{(1)},\dots,\lambda_{q-1}^{(1)}]$ and
         $\blambda_2=[\lambda_0^{(2)},\dots,\lambda_{q-1}^{(2)}]$,
         respectively.  The check- and bit-node combined channels satisfy
         \begin{align*}
         F(W_{1}\vnop W_{2}) & \leq  (q-1)F(W_1)F(W_{2})\\
              F(W_{1}\cnop W_{2}) & \leq F(W_1)+F(W_{2})+ (q-1)F(W_1)F(W_{2}).
             \end{align*}
   \end{lem}

\begin{lem}\label{lem fidelity bound for heralded PSC}
Let $W_H^{(1)},W_H^{(2)}\in\cM_q$ have the forms
   \begin{align*}
       W_H^{(1)}(j) & =\sum_{x_1\in \cX_1}p_{x_1}^{(1)}W_{x_1}^{(1)}(j)\otimes \ketbra{x_1}{x_1}\\
       W_H^{(2)}(j) & =\sum_{x_2\in \cX_2}p_{x_2}^{(2)}W_{x_2}^{(2)}(j)\otimes \ketbra{x_2}{x_2}
   \end{align*}
for every $j\in[q]$, where $\cX_1$ and $\cX_2$ are finite. Then, we have
\begin{align*}
   F(W_{H}^{(1)}\vnop W_{H}^{(2)}) & \leq (q-1)F(W_{H}^{(1)})F(W_{H}^{(2)})\\
   F(W_{H}^{(1)}\cnop W_{H}^{(2)}) & \leq F(W_H^{(1)})+F(W_{H}^{(2)})+(q-1)F(W_{H}^{(1)})F(W_{H}^{(2)}).
\end{align*}
\end{lem}
   \begin{proof}

Conditioning on the independent herald values and applying
Lemma~\ref{lem:fidelity check and bit node bounds} gives
\begin{align*}
   F(W_{H}^{(1)}\vnop W_{H}^{(2)}) & = \mathbb{E}_{x_1\in\cX_1,x_2\in \cX_2}[F(W_{x_1}^{(1)}\vnop W_{x_2}^{(2)})]\\
   & \leq (q-1)\mathbb{E}_{x_1\in\cX_1,x_2\in \cX_2}[F(W_{x_1}^{(1)})F (W_{x_2}^{(2)})]\\
    & = (q-1)F(W_{H}^{(1)})F(W_{H}^{(2)}).
\end{align*}
The check-node bound follows similarly from
\begin{align*}
F(W_{H}^{(1)}\cnop W_{H}^{(2)}) & = \mathbb{E}_{x_1\in \cX_1,x_2\in \cX_2}[F(W_{x_1}^{(1)}\cnop W_{x_2}^{(2)})]\\
   & \leq \mathbb{E}_{x_1\in \cX_1,x_2\in \cX_2}\Bigg[F(W_{x_1}^{(1)})+F(W_{x_2}^{(2)})+ (q-1)F(W_{x_1}^{(1)})F(W_{x_2}^{(2)})\Bigg]\\
& = F(W_H^{(1)})+F(W_{H}^{(2)})+(q-1)F(W_{H}^{(1)})F(W_{H}^{(2)}). \qedhere
\end{align*}
\end{proof}

\begin{lem}\label{lem:coefficient-uniform-node-fidelity}
Let \(W_1,W_2\in\cM_q\) and for all \(a,b\in[q]\setminus\{0\}\),
the multiplication node from Definition~\ref{def:multiplication node},
applied to every PSC component of each heralded mixture, satisfies
\begin{align*}
F\left(W_1^{(a)}\vnop W_2^{(b)}\right)
&\leq
(q-1)F(W_1)F(W_2),\\
F\left(W_1^{(a)}\cnop W_2^{(b)}\right)
&\leq
F(W_1)+F(W_2)+(q-1)F(W_1)F(W_2).
\end{align*}
\end{lem}

\begin{proof}
Multiplication permutes the eigen list of every PSC in each heralded
mixture. Hence, by Lemma~\ref{lem:multiplication-node},
\begin{align*}
F\left(W_1^{(a)}\right)=F(W_1),
\qquad
F\left(W_2^{(b)}\right)=F(W_2).
\end{align*}
Applying Lemma~\ref{lem fidelity bound for heralded PSC} to the
multiplied input channels proves both inequalities.
\end{proof}

\begin{lem}\label{lem:heralded PSC fidelity bound}
Let $W_1,\ldots,W_d\in\cM_q$.  Then the following bounds hold:
   \begin{align*}
       F(W_{1}\vnop \dots\vnop W_{d}) & \leq (q-1)^{d-1}\prod_{i=1}^{d}F(W_i)\\
       1+(q-1)F(W_{1}\cnop \dots\cnop W_{d}) & \leq \prod_{i=1}^{d}\left(1+(q-1)F(W_{i})\right).
   \end{align*}
\end{lem}
\begin{proof}
   Repeated application of Lemma~\ref{lem fidelity bound for heralded PSC}
   to bit-node combining gives
   \begin{align*}
       F(W_{1}\vnop \dots\vnop W_{d}) & \leq (q-1) F(W_1)F(W_2\vnop \dots\vnop W_d)\\
       & \leq (q-1)^{d-1}\prod_{i=1}^{d} F(W_i).
   \end{align*}
   For check-node combining, we obtain
   \begin{align*}
      (q-1) F(W_1\cnop W_2) &\leq (q-1)F(W_1)+(q-1)F(W_2)+(q-1)^2F(W_1)F(W_2).
   \end{align*}
   This implies
   \begin{align*}
       1+(q-1)F(W_1\cnop W_2)  \leq  (1+(q-1)F(W_1))(1+(q-1)F(W_2).
   \end{align*}
Iterating this inequality yields
   \begin{align*}
      1+(q-1) F(W_1\cnop \dots\cnop W_d ) & \leq (1+(q-1)F(W_1))(1+(q-1)F(W_2\cnop \dots \cnop W_d))\\
      & \leq \prod_{i=1}^{d}(1+(q-1)F(W_i)).\qedhere
   \end{align*}
\end{proof}
For nonzero factors \(a_1,\ldots,a_d\), the node inputs are
\(W_1^{(a_1)},\ldots,W_d^{(a_d)}\).  Their combination can depend on the
relative factors. Repeated application of
Lemma~\ref{lem:coefficient-uniform-node-fidelity} gives
\begin{align*}
F\left(W_1^{(a_1)}\vnop\cdots\vnop W_d^{(a_d)}\right)
&\leq
(q-1)^{d-1}\prod_{i=1}^d F(W_i),\\
1+(q-1)F\left(W_1^{(a_1)}\cnop\cdots\cnop W_d^{(a_d)}\right)
&\leq
\prod_{i=1}^d\left(1+(q-1)F(W_i)\right).
\end{align*}
Thus, the exact combined fidelity can depend on
\((a_1,\ldots,a_d)\), whereas the displayed upper bound is uniform over
these factors.
Averaging these bounds over the independent incoming coefficient
trajectories
gives the scalar recursion below.

\begin{theorem}[Double-exponential decay on a regular tree]
\label{thm:double-exponential-error-rate}
Let $F_t$ and $P_t$ be the average fidelity and symbol-error
probability defined in Section~\ref{sec:bpqm-density-evolution}. For a
$(d_v,d_c)$-regular computation tree with $d_v\geq3$, they satisfy
\begin{align}\label{Eq:regular-fidelity-recursion-bound}
F_{t+1}
&\leq
F(W)
\left(
\left(
1+(q-1)F_t
\right)^{d_c-1}
-1
\right)^{d_v-1}.
\end{align}
Moreover, if \(W\in\mathrm{Reg}_{\mathrm{BPQM}}(d_v,d_c)\), then there exist constants
$A,\beta>0$ and an integer $\ell\geq0$ such that, for every $s\geq0$,
\begin{align}\label{Eq:regular-symbol-error-double-exponential}
P_{\ell+s}
&\leq
A
\exp\left(
-\beta(d_v-1)^s
\right).
\end{align}
\end{theorem}
\begin{proof}
For each $r\in[d_v-1]$, let
\(W_{t+1,r}^{\cnop\bA_{t+1,r}}\) denote the \(r\)th random incoming
check-to-variable channel at iteration $t+1$. For every
realization of the incoming coefficient trajectories and multiplication
factors, repeated application of
Lemma~\ref{lem:coefficient-uniform-node-fidelity} gives
\begin{align*}
1+(q-1)F\left(W_{t+1,r}^{\cnop\bA_{t+1,r}}\right)
&\leq
\prod_{s\in[d_c-1]}
\left[
1+ (q-1) F\left( \left( W_t^{\bA_t^{(r,s)}}\right)^{(A_{t+1,r,s})} \right)\right]\\
&= \prod_{s\in[d_c-1]} \left[ 1+ (q-1) F\left(
W_t^{\bA_t^{(r,s)}} \right) \right].
\end{align*}
The equality follows from Lemma~\ref{lem:multiplication-node}. Taking
expectations and using the independence of the incoming coefficient
trajectories gives
\begin{align*}
1+ (q-1) \E\left[ F\left(W_{t+1,r}^{\cnop\bA_{t+1,r}}\right)
\right]
&\leq \prod_{s\in[d_c-1]}
\E\left[ 1+ (q-1) F\left( W_t^{\bA_t^{(r,s)}} \right)
\right]\\
&= \left( 1+(q-1)F_t
\right)^{d_c-1}.
\end{align*}
It follows that
\begin{align*}
\E\left[ F\left(W_{t+1,r}^{\cnop\bA_{t+1,r}}\right)\right]
&\leq \frac{ \left( 1+(q-1)F_t
\right)^{d_c-1} -1
}{q-1}.
\end{align*}
For every realization of $\bA_{t+1}$, the bit-node bound in
Lemma~\ref{lem:heralded PSC fidelity bound} gives
\begin{align*}
F\left(
W_{t+1}^{\bA_{t+1}}
\right)
&\leq
(q-1)^{d_v-1}
F(W)
\prod_{r\in[d_v-1]}
F\left(W_{t+1,r}^{\cnop\bA_{t+1,r}}\right).
\end{align*}
The channels
\(\{W_{t+1,r}^{\cnop\bA_{t+1,r}}:r\in[d_v-1]\}\) are independent and
identically distributed. Hence,
\begin{align*}
F_{t+1}
&\leq
(q-1)^{d_v-1}
F(W)
\prod_{r\in[d_v-1]}
\E\left[
F\left(W_{t+1,r}^{\cnop\bA_{t+1,r}}\right)
\right]\\
&=
(q-1)^{d_v-1}
F(W)
\left(
\E\left[
F\left(W_{t+1,0}^{\cnop\bA_{t+1,0}}\right)
\right]
\right)^{d_v-1}\\
&\leq
F(W)
\left(
\left(
1+(q-1)F_t
\right)^{d_c-1}
-1
\right)^{d_v-1}.
\end{align*}
This proves the stated scalar fidelity inequality.
Now assume that \(W\in\mathrm{Reg}_{\mathrm{BPQM}}(d_v,d_c)\), so that
$F_t\rightarrow0$. Since
\begin{align*}
\left(
1+(q-1)x
\right)^{d_c-1}
-1
&=
(q-1)(d_c-1)x
+
O(x^2)
\end{align*}
as $x\rightarrow0$, there exist constants $\epsilon>0$ and $C_0<\infty$
such that
\begin{align*}
\left(
1+(q-1)x
\right)^{d_c-1}
-1
&\leq
C_0x
\end{align*}
for every $x\in[0,\epsilon]$.  For $F_t\leq\epsilon$, the scalar fidelity
inequality gives
\begin{align*}
F_{t+1}
&\leq
F(W)
C_0^{d_v-1}
F_t^{d_v-1}.
\end{align*}
Set \(\Lambda\coloneqq F(W)C_0^{d_v-1}\).
If $\Lambda=0$, or if $F_t=0$ for some $t$, then the conclusion follows
immediately. Otherwise, because $F_t\rightarrow0$, there exists an
integer $\ell\geq0$ such that
\begin{align*}
F_\ell
&\leq \epsilon,\\
\Lambda^{1/(d_v-2)}F_\ell
&<1.
\end{align*}
An induction on $s$ gives
\begin{align*}
F_{\ell+s}
&\leq
\Lambda^{-1/(d_v-2)}
\left(
\Lambda^{1/(d_v-2)}
F_\ell
\right)^{(d_v-1)^s}
\end{align*}
for every $s\geq0$.  Set
\(\beta\coloneqq-\log(\Lambda^{1/(d_v-2)}F_\ell)>0\). Then
\begin{align*}
F_{\ell+s}
&\leq
\Lambda^{-1/(d_v-2)}
\exp\left(
-\beta(d_v-1)^s
\right).
\end{align*}
Finally, the averaged fidelity--error bound from
Section~\ref{sec:bpqm-density-evolution} gives
\begin{align*}
P_{\ell+s}
&\leq (q-1)F_{\ell+s}\\
&\leq
(q-1)
\Lambda^{-1/(d_v-2)}
\exp\left(
-\beta(d_v-1)^s
\right).
\end{align*}
Taking \(A\coloneqq(q-1)\Lambda^{-1/(d_v-2)}\) proves the claim.
\end{proof}

\section{Vanishing Block-Error Probability in the BPQM Success Region}
\label{sec:vanishing-block-error}
In this section, we prove that the proposed decoder has vanishing
ensemble-average block-error probability throughout the BPQM success
region.  Double-exponential symbol-error decay controls the sequential
BPQM measurements on cycle-free neighbourhoods.  We combine this estimate
with a bound on the number of cyclic neighbourhoods and minimum-distance
erasure recovery to obtain the block-error result.

For blocklength $N$, let $\cG_N$ be a random $(d_v,d_c)$-regular Tanner
graph with $N$ variable nodes and $m=\frac{Nd_v}{d_c}$ check nodes, obtained by
a uniform matching of the variable-node and check-node sockets.
Independently assign every edge a coefficient sampled uniformly from
\([q]\setminus\{0\}\), and let \(H\) be the parity-check matrix defined
by~\eqref{eq:aggregated-parity-check}
of the \(q\)-ary LDPC code \(\cC_N\subseteq\mathbb F_q^N\)
\cite{gallager1962low,richardson2008modern}.  By
Lemma~\ref{lem:local-bpqm-projector-covariance}, the decoder has the same
error probability for every transmitted codeword.  For a realized code
\(\cC_N\), define its conditional block-error probability by
\begin{align*}
   P(\cC_N,W)
   \coloneqq
   \Pr\!\left\{\widehat{\bc}\neq\bc\,\middle|\,\cC_N\right\},
\end{align*}

where the probability is over the decoder measurement outcomes for the
channel-output state induced by \(W\).  As the sampled code varies,
\(P(\cC_N,W)\) is a \([0,1]\)-valued random variable on the LDPC
ensemble.  Its ensemble average is
\begin{align}
   \pblk(N,W)
   \coloneqq
   \E_{\cC_N}\!\left[P(\cC_N,W)\right].
   \label{eq:regular-ensemble-block-error}
\end{align}

Recall the depth-$\ell$ neighbourhoods $\cN_\ell(i)$ and the sets
$G_\ell,B_\ell$ from \eqref{eq:good-bad-coordinates}.  The decoder in
Section~\ref{sec:bpqm-ldpc} applies depth-$\ell$ BPQM to the coordinates
in $G_\ell$ and treats the coordinates in $B_\ell$ as erasures.

For each realization of a code $\cC_N$ from the ensemble and each
$i\in G_\ell$, let
$\ba_{\ell,i}\in\cA_\ell$ be the coefficient trajectory obtained by
directing its messages toward the variable node $i$ and normalizing each
check equation with respect to its outgoing edge.  The exact effective channel for
estimating $c_i$ is $W_\ell^{\ba_{\ell,i}}$.

The event $i\in G_\ell$ depends only on the Tanner graph and is independent
of the nonzero edge coefficients.  Conditioned on $i\in G_\ell$, the
normalized coefficients in $\cN_\ell(i)$ are independent and uniform on
$[q]\setminus\{0\}$.  Consequently, the ensemble distribution of the
realizations $\ba_{\ell,i}$ satisfies
\begin{align*}
   \Pr\left\{
      \ba_{\ell,i}=\ba
      \mid i\in G_\ell
   \right\}
   =
   \Pr\left\{
      \bA_\ell=\ba
   \right\},
   \qquad \ba\in\cA_\ell.
\end{align*}
Thus, the identification with the density-evolution symbol-error
probability is made after averaging over the random edge coefficients.
\subsection{Erasure Recovery}
In the proposed decoder, we deliberately omit coordinates with cyclic local
neighbourhoods during the BPQM stage.  The following standard criterion
relates minimum distance to the rank of a parity-check submatrix and hence to
unique erasure recovery~\cite[Corollary~1.4.14]{huffman2003fundamentals}.
Let \(d_{\min}(\cC)\) denote the minimum Hamming distance of \(\cC\).  For
any \(S\subseteq[N]\), let \(H_S\) denote the submatrix of \(H\) formed by
the columns indexed by \(S\), and let \(\bc_S\) denote the restriction of
\(\bc\) to \(S\).

\begin{lem}\label{lem:minimum-distance-erasure-recovery}
Let $\cC\subseteq\mathbb F_q^N$ be a linear code with parity-check
matrix $H$, and let $B\subseteq[N]$ satisfy
\begin{align*}
   |B|<d_{\min}(\cC).
\end{align*}
Then $H_B$ has full column rank.  If the transmitted codeword is
$\bc\in\cC$ and the decoder correctly recovers $\bc_G$ on
$G=[N]\setminus B$, then $\bc_B$ is uniquely determined by $\bc_G$ and
can be recovered by solving
\begin{align*}
   H_B\bc_B=-H_G\bc_G.
\end{align*}
\end{lem}

The proof is included in
Appendix~\ref{app:minimum-distance-erasure-recovery}.
\subsection{Counting Tree Neighbourhoods}
The local tree property of random LDPC Tanner graphs is standard in
density-evolution analysis~\cite{gallager1962low,richardson2008modern}.
For the erasure-recovery step, we use the following finite-blocklength
estimate on the number of coordinates whose depth-\(\ell\) neighbourhoods
contain a cycle.  It shows that this number is sublinear at the depth used
in the BPQM error estimate.
\begin{lem}\label{lem:cyclic-neighbourhood-count}
Let $d_v\geq2$, define $\alpha\coloneqq(d_v-1)(d_c-1)$, and let
$\ell=\ell(N)$ satisfy $\alpha^{\ell(N)}=o(N)$.  Then
\begin{align}\label{Eq:expected-cyclic-neighbourhoods}
   \E|B_{\ell(N)}|
   =O\!\left(\alpha^{2\ell(N)}\right),
\end{align}
where the implicit constant depends only on $d_v$ and $d_c$ and is
independent of $N$ and $\ell(N)$.
\end{lem}
\begin{proof}
Write $\ell=\ell(N)$.  Fix a variable node $i$ and reveal its
depth-$\ell$ computation neighbourhood $\cN_\ell(i)$.  It is contained in
the Tanner subgraph induced by the vertices within graph distance
\(2\ell\) from \(i\).
The following bounds allow all $d_v$ branches at the initial variable node
and hence also apply to the computation graph formed from $d_v-1$ incoming
branches.
The number of check nodes is at most
\begin{align*}
N_C(\ell)
=
d_v\sum_{r=0}^{\ell-1}\alpha^r,
\end{align*}
and the number of variable nodes is at most
\begin{align*}
N_V(\ell)
=
1+d_v(d_c-1)\sum_{r=0}^{\ell-1}\alpha^r.
\end{align*}
Thus, the number of sockets attached to the explored vertices satisfies
\begin{align*}
M_\ell
&\le d_vN_V(\ell)+d_cN_C(\ell)\\
&= d_v+ \bigl(d_v^2(d_c-1)+d_vd_c\bigr) \sum_{r=0}^{\ell-1}\alpha^r .
\end{align*}
Since $\sum_{r=0}^{\ell-1}\alpha^r
\leq\alpha^\ell/(\alpha-1)$, we have
\begin{align*}
M_\ell \le \xi_{d_v,d_c}\alpha^\ell,
\end{align*}
where one valid explicit choice is
\begin{align*}
\xi_{d_v,d_c}
\triangleq
d_v+\frac{d_v^2(d_c-1)+d_vd_c}{\alpha-1}.
\end{align*}
The constant $\xi_{d_v,d_c}$ depends only on $d_v$ and $d_c$.  Since
$\alpha^\ell=o(N)$, the inequality $M_\ell\leq Nd_v/2$ holds for all
sufficiently large $N$.
Now consider the random regular LDPC ensemble.  During the layer-by-layer
exploration of $\cN_\ell(i)$, a cycle is formed if a newly matched socket
belongs to a vertex already present in the explored neighbourhood.  This
event also includes the creation of a parallel edge.  At every step, at
most $M_\ell$ sockets on the opposite side belong to vertices reached
earlier.  Since there are $Nd_v$ sockets on each side, the conditional
probability of this event is at most
$M_\ell/(Nd_v-M_\ell)\leq 2M_\ell/(Nd_v)$ whenever
$M_\ell\leq Nd_v/2$.  Each revealed edge uses a socket attached to an
explored vertex, so at most $M_\ell$ edges are revealed.  A union bound
gives
\begin{align*}
\Pr(i\in B_\ell) \leq \frac{2M_\ell^2}{Nd_v} \leq
\frac{2\xi_{d_v,d_c}^2\alpha^{2\ell}}{Nd_v}.
\end{align*}
Summing over $i\in[N]$ gives
\begin{align}\label{Eq:explicit-cyclic-neighbourhood-bound}
\E |B_\ell|
=
\sum_{i\in[N]} \Pr\{i\in B_\ell\}
\le
\kappa_{d_v,d_c}\alpha^{2\ell},
\end{align}
where
\begin{align*}
\kappa_{d_v,d_c}
\triangleq
\frac{2\xi_{d_v,d_c}^2}{d_v}.
\end{align*}
The constant $\kappa_{d_v,d_c}$ also depends only on $d_v$ and $d_c$.
\eqref{Eq:explicit-cyclic-neighbourhood-bound} proves
\eqref{Eq:expected-cyclic-neighbourhoods} with an implicit constant
independent of $N$ and $\ell(N)$.
\end{proof}

Linear minimum-distance estimates for random-coefficient regular
finite-field LDPC ensembles were established in
\cite{bennatan2004ldpc}; the sharper estimate in
\cite[Theorem~VI.2 and Remark~VI.3]{yang2011ldpc_qary_weight} applies to
the edge-coefficient ensemble considered here.  Indeed, the uniform
socket matching is the permutation construction used there, and
independent uniform nonzero coefficients give the corresponding random
finite-field maps; parallel edges contribute additively to the same
parity-check entry as in~\eqref{eq:aggregated-parity-check}.  In
particular, when
\(d_v\geq3\), there exists a constant \(\zeta>0\), independent of \(N\),
such that
\begin{align}
   \Pr\!\left\{
      d_{\min}(\cC_N)<\zeta N
   \right\}
   \longrightarrow 0 .
   \label{eq:regular-linear-minimum-distance}
\end{align}

\begin{lem}\label{lem:bad_symbols_error}
Let \(d_v\geq3\), and let \(\ell=\ell(N)\) satisfy
\(\alpha^{2\ell}=o(N)\).  Then
\begin{align*}
\frac{|B_\ell|}{N}\to 0
\end{align*}
in probability, and
\begin{align*}
\Pr\{|B_\ell|\ge d_{\min}(\cC_N)\}\to 0.
\end{align*}
\end{lem}

\begin{proof}
By Lemma~\ref{lem:cyclic-neighbourhood-count},
\eqref{Eq:explicit-cyclic-neighbourhood-bound} holds for all sufficiently
large $N$.  Markov's inequality gives, for every $\varepsilon>0$,
\begin{align*}
\Pr\{|B_\ell|\ge \varepsilon N\}
\le
\frac{\E |B_\ell|}{\varepsilon N}
\le
\frac{\kappa_{d_v,d_c}\alpha^{2\ell}}{\varepsilon N}.
\end{align*}
Thus \(|B_\ell|/N\to0\) in probability.  For the constant \(\zeta\)
in~\eqref{eq:regular-linear-minimum-distance},
\begin{align*}
\Pr\{|B_\ell|\geq d_{\min}(\cC_N)\}
&\leq
\Pr\{d_{\min}(\cC_N)<\zeta N\}
+\Pr\{|B_\ell|\geq\zeta N\}\\
&\leq
\Pr\{d_{\min}(\cC_N)<\zeta N\}
+ \frac{\kappa_{d_v,d_c}\alpha^{2\ell}}{\zeta N}.
\end{align*}
The first term vanishes by
\eqref{eq:regular-linear-minimum-distance}, and the second vanishes
because \(\alpha^{2\ell}=o(N)\).
\end{proof}

\subsection{Vanishing Block-Error Probability}
We now combine the tree-decoding and erasure-recovery analyses to prove the
block-error theorem.  We choose a depth
$\ell(N)=\Theta(\log\log N)$.  The double-exponential bound in
\eqref{Eq:regular-symbol-error-double-exponential}, together with Gao's
noncommutative union bound, controls all coordinates decoded by BPQM.
\eqref{Eq:expected-cyclic-neighbourhoods} leads to
Lemma~\ref{lem:bad_symbols_error}, which shows that the remaining set is
smaller than the minimum distance with high probability.
Lemma~\ref{lem:minimum-distance-erasure-recovery} then recovers this set
from the parity-check equations.  These steps establish vanishing
block-error probability for the two-stage decoder.
\begin{theorem}\label{thm:vanishing_block_error_rate}
For the \((d_v,d_c)\)-regular random \(q\)-ary LDPC ensemble with
\(d_v\geq3\), set
\(\ell(N)\coloneqq\lfloor \tau\log\log N\rfloor\), where
\(\tau>1/\log(d_v-1)\), and apply the decoder of
Section~\ref{sec:bpqm-ldpc}.  If
\(W\in\mathrm{Reg}_{\mathrm{BPQM}}(d_v,d_c)\), then its ensemble-average block-error
probability satisfies
\begin{align*}
   \lim_{N\rightarrow\infty} \pblk(N,W)=0 .
\end{align*}
\end{theorem}
\begin{proof}
For a realized LDPC code $\cC_N$ and $i\in G_{\ell(N)}$, let
$\hat c_i$ be the depth-$\ell(N)$ BPQM estimate of $c_i$, and let
$\ba_{\ell(N),i}\in\cA_{\ell(N)}$ be the coefficient trajectory of
$\cN_{\ell(N)}(i)$.  The effective channel for this
estimate is $W_{\ell(N)}^{\ba_{\ell(N),i}}$.
For the sequential measurement outcomes of this realized code, define the
random set of BPQM error locations
\begin{align*}
   \err_1(\cC_N)
   \coloneqq
   \left\{
      i\in G_{\ell(N)}
      \colon \hat c_i\neq c_i
   \right\}.
\end{align*}
Thus, a BPQM-stage error occurs exactly when
\(\err_1(\cC_N)\neq\varnothing\).  Over the
$(d_v,d_c)$-regular LDPC ensemble, define the erasure-failure event
\begin{align*}
   \err_2
   \coloneqq
   \left\{
      \cC_N
      \colon
      |B_{\ell(N)}|
      \geq
      d_{\min}(\cC_N)
   \right\}.
\end{align*}
The event \(\err_2\) consists of the sampled codes for which the number of
bad coordinates is at least the minimum distance.
The decoder can fail only if the BPQM error event occurs or the
erasure-recovery condition fails.  Indeed, fix a sampled code $\cC_N$.
If $\err_1(\cC_N)=\varnothing$ and $\cC_N\notin\err_2$, then all good
symbols are decoded correctly and the bad symbols form an erasure set of
size strictly smaller than $d_{\min}(\cC_N)$.  By
Lemma~\ref{lem:minimum-distance-erasure-recovery}, the transmitted
codeword is uniquely recovered.  Hence, the conditional block-error
probability satisfies
\begin{align*}
   P(\cC_N,W)
   \leq
   \Pr\!\left\{
      \err_1(\cC_N)\neq\varnothing
      \,\middle|\,\cC_N
   \right\}
   +
   \mone\{\cC_N\in\err_2\}.
\end{align*}
Averaging over the random code gives
\begin{align*}
   \pblk(N,W) 
   \leq
   \E_{\cC_N}
   \left[
      \Pr\!\left\{
         \err_1(\cC_N)\neq\varnothing
         \,\middle|\,\cC_N
      \right\}
   \right]
   +
   \Pr\!\left\{\cC_N\in\err_2\right\},
\end{align*}
where the probability in the expectation is over the sequential measurement
outcomes conditioned on the sampled code, whereas the last probability is
over the LDPC ensemble.

Choose an arbitrary ordering of the symbols in $G_{\ell(N)}$ and decode
them sequentially in this order.  Fix a transmitted codeword
$\bc\in\cC_N$, and let
\begin{align*}
   \rho_{\bc}
   \coloneqq
   \bigotimes_{j=1}^N W(c_j)
\end{align*}
be the channel output state.  For \(i\in G_{\ell(N)}\), let
\(\Pi_i(c_i)\) be the success projector
in~\eqref{eq:local-bpqm-success-projector}.  By
Lemma~\ref{lem:local-bpqm-projector-covariance}, its success probability
is the same for all assignments satisfying the parity checks in
\(\cN_{\ell(N)}(i)\) and having value \(c_i\) at variable node \(i\).
Averaging this common value over the assignments defining the effective
channel and applying~\eqref{eq:heralded-symbol-independent-error} gives
\begin{align*}
   \Tr\!\left(
      \Pi_i(c_i)\rho_{\bc}
   \right)
   =
   1-\perr\!\left(
      W_{\ell(N)}^{\ba_{\ell(N),i}}
   \right).
\end{align*}

Gao's noncommutative union bound~\cite{gao2015quantum}, applied to the
fixed code $\cC_N$, gives
\begin{align*}
   \Pr\!\left\{
      \err_1(\cC_N)\neq\varnothing
      \,\middle|\,\cC_N
   \right\}
   &\leq
   4
   \sum_{i\in G_{\ell(N)}}
   \Tr\left(
      \left(
         \mI-\Pi_i(c_i)
      \right)
      \rho_{\bc}
   \right)\\
   &=
   4
   \sum_{i\in G_{\ell(N)}}
   \perr\left(
      W_{\ell(N)}^{\ba_{\ell(N),i}}
   \right).
\end{align*}
Averaging over the LDPC ensemble yields
\begin{align*}
   \E_{\cC_N}
   \left[
      \Pr\!\left\{
         \err_1(\cC_N)\neq\varnothing
         \,\middle|\,\cC_N
      \right\}
   \right]
   &\leq
   4
   \sum_{i\in[N]}
   \E_{\cC_N}
   \left[
      \mone\{i\in G_{\ell(N)}\}
      \perr\left(
         W_{\ell(N)}^{\ba_{\ell(N),i}}
      \right)
   \right].
\end{align*}
Conditioned on $i\in G_{\ell(N)}$, the ensemble distribution of the
realizations $\ba_{\ell(N),i}$ coincides with the distribution of
$\bA_{\ell(N)}$.
Recalling that
\begin{align*}
   P_t
   =
   \E_{\bA_t}
   \left[
      \perr\left(W_t^{\bA_t}\right)
   \right],
\end{align*}
we obtain
\begin{align*}
   &\E_{\cC_N}
   \left[
      \mone\{i\in G_{\ell(N)}\}
      \perr\left(
         W_{\ell(N)}^{\ba_{\ell(N),i}}
      \right)
   \right]
   =
   \Pr\!\left\{
      i\in G_{\ell(N)}
   \right\}
   P_{\ell(N)}
   \leq
   P_{\ell(N)}.
\end{align*}
It follows that
\begin{align*}
   \E_{\cC_N}
   \left[
      \Pr\!\left\{
         \err_1(\cC_N)\neq\varnothing
         \,\middle|\,\cC_N
      \right\}
   \right]
   \leq
   4N P_{\ell(N)}.
\end{align*}

By Theorem~\ref{thm:double-exponential-error-rate}, there exist constants
$A,\beta>0$ and an integer $\ell_0\geq0$ such that
\begin{align*}
   P_{\ell(N)}
   \leq
   A\exp\left(
      -\beta(d_v-1)^{\ell(N)-\ell_0}
   \right)
\end{align*}
for all sufficiently large $N$.  Setting
\(\beta'\coloneqq\beta/(d_v-1)^{\ell_0}\),
we obtain
\begin{align*}
   \E_{\cC_N}
   \left[
      \Pr\!\left\{
         \err_1(\cC_N)\neq\varnothing
         \,\middle|\,\cC_N
      \right\}
   \right]
   \leq
   4NA
   \exp\left(
      -\beta'(d_v-1)^{\ell(N)}
   \right).
\end{align*}
Since
\begin{align*}
   (d_v-1)^{\ell(N)}
   &\geq
   \frac{1}{d_v-1}
   (\log N)^{\tau\log(d_v-1)},
\end{align*}
it follows that
\begin{align*}
   \E_{\cC_N}
   \left[
      \Pr\!\left\{
         \err_1(\cC_N)\neq\varnothing
         \,\middle|\,\cC_N
      \right\}
   \right]
   \leq
   4A
   \exp\left(
      \log N
      -
      \frac{\beta'}{d_v-1}
      (\log N)^{\tau\log(d_v-1)}
   \right).
\end{align*}
Because $\tau\log(d_v-1)>1$,
\begin{align*}
   \lim_{N\rightarrow\infty}
   \E_{\cC_N}
   \left[
      \Pr\!\left\{
         \err_1(\cC_N)\neq\varnothing
         \,\middle|\,\cC_N
      \right\}
   \right]
   =
   0.
\end{align*}
On the other hand, for fixed constant $\tau$, we have
\begin{align*}
   \alpha^{2\ell(N)}
   \leq
   \alpha^{2\tau\log\log N}
   =
   (\log N)^{2\tau\log \alpha}
   =o(N).
\end{align*}
Lemma~\ref{lem:bad_symbols_error} gives
\begin{align*}
   \lim_{N\to \infty}
   \Pr\!\left\{\cC_N\in\err_2\right\}=0.
\end{align*}
Combining these, we obtain
\begin{align*}
   \lim_{N\to \infty}\pblk(N, W)=0.
\end{align*}
\end{proof}

\begin{corol}\label{cor:regular-ensemble-convergence-in-probability}
Under the assumptions of
Theorem~\ref{thm:vanishing_block_error_rate}, the block-error probability
of a code sampled from the regular LDPC ensemble converges to zero in
probability.  In particular, for every \(\varepsilon>0\),
\begin{align*}
   \lim_{N\to\infty}
   \Pr\!\left\{
      P(\cC_N,W)>\varepsilon
   \right\}
   =0.
\end{align*}
\end{corol}
\begin{proof}
By Markov's inequality and~\eqref{eq:regular-ensemble-block-error},
\begin{align*}
   \Pr\!\left\{
      P(\cC_N,W)>\varepsilon
   \right\}
   \leq
   \frac{\pblk(N,W)}{\varepsilon}.
\end{align*}
The right-hand side tends to zero by
Theorem~\ref{thm:vanishing_block_error_rate}.
\end{proof}

Thus, asymptotically almost every code sampled from the ensemble has
vanishing conditional block-error probability. More precisely, the
probability of sampling a code whose conditional block-error probability
exceeds any fixed positive value tends to zero as $N\to\infty$.

\appendix
\section{Irregular LDPC Codes}
\label{app:irregular-ldpc}
We extend the regular-ensemble analysis to
finite-support irregular degree distributions.  We average the BPQM
messages over the random computation-tree degrees to obtain the irregular
fidelity recursion and its decay estimate.  Under a linear-distance
condition, the neighbourhood-counting and erasure-recovery arguments give
vanishing block-error probability for the irregular ensemble.

An irregular LDPC ensemble is specified by the distributions of the
variable- and check-node degrees
\cite{richardson2008modern,pfister2014density}.  Let
\begin{align}
   L(x)=\sum_{i\geq1}L_i x^i,
   \qquad
   P(x)=\sum_{j\geq1}P_j x^j,                              \label{eq:irregular-node-distributions}
\end{align}
where \(L_i\) and \(P_j\) are the fractions of variable and check nodes
of degrees \(i\) and \(j\), respectively.  Thus,
\(L(1)=P(1)=1\).  For \(N\) variable nodes and \(m\) check nodes, equality
of the numbers of sockets on the two sides requires
\begin{align*}
   NL'(1)=mP'(1).
\end{align*}
The design rate is
\begin{align}
   R
   =
   1-\frac{m}{N}
   =
   1-\frac{L'(1)}{P'(1)}.                                  \label{eq:irregular-design-rate}
\end{align}
The corresponding edge-perspective degree distributions are
\begin{align}
   \eta(x)=
   \frac{L'(x)}{L'(1)}
   =
   \sum_{i\geq1}\eta_i x^{i-1},\qquad
   \gamma(x)=
   \frac{P'(x)}{P'(1)}
   =
   \sum_{j\geq1}\gamma_j x^{j-1},                           \label{eq:irregular-edge-distributions}
\end{align}
where
\begin{align*}
   \eta_i=\frac{iL_i}{L'(1)},
   \qquad
   \gamma_j=\frac{jP_j}{P'(1)}.
\end{align*}
Hence, the variable-node and check-node endpoints of a uniformly chosen
edge have degrees \(i\) and \(j\) with probabilities \(\eta_i\) and
\(\gamma_j\), respectively.  Equivalently,
\begin{align*}
   R
   =
   1-
   \frac{\int_0^1\gamma(x)\,dx}
        {\int_0^1\eta(x)\,dx}.
\end{align*}

Along blocklengths for which the prescribed node counts are integral, let
\(\mathsf E_N(L,P)\) denote the standard irregular ensemble obtained by
assigning the degree sequences uniformly to the variable and check nodes
and matching the two socket sets uniformly at random.  For a \(q\)-ary
ensemble, every matched
edge is assigned, independently and uniformly, an element of
\([q]\setminus\{0\}\).  The corresponding parity-check entry is the sum
of the coefficients of all edges joining the same variable and check
nodes, as in~\eqref{eq:aggregated-parity-check}.  After normalizing a check
equation by the coefficient of its outgoing edge, the remaining
coefficients induce the multiplication channels of
Definition~\ref{def:multiplication node}.  In the local weak limit,
the neighbourhood of a uniformly chosen edge is a computation tree whose
variable-node degrees are distributed according to \(\eta\) and whose
check-node degrees are distributed according to \(\gamma\).  A uniformly
chosen root variable node instead has degree distribution \(L\).  These
edge- and node-perspective computation trees are the objects used in
density evolution for irregular LDPC ensembles.

\subsection{Fidelity evolution on an irregular tree}

We now incorporate the edge-perspective degree distributions into the
exact channel recursion and the scalar fidelity bound.
Consider the irregular computation tree with physical channel \(W\).
Let \(\mathsf D_t\) denote the depth-\(t\)
edge-perspective degree tree generated by \(\eta\) and \(\gamma\).
Conditional on \(\mathsf D_t\), let \(\bA_t\) be the random variable whose
realizations are the normalized coefficient trajectories induced by the
sampled edge coefficients.  Its relative-coefficient components are
independent and uniform on \([q]\setminus\{0\}\).  Write
\(W_t(\mathsf D_t,\bA_t)\) for the resulting exact effective channel.
Define
\begin{align}
   F_{\gamma,\eta}(W_t)
   \triangleq
   \E_{\mathsf D_t}
   \E_{\bA_t\mid\mathsf D_t}
   \left[
      F\left(W_t(\mathsf D_t,\bA_t)\right)
   \right],
   \qquad
   F_{\gamma,\eta}(W_0)=F(W),
   \label{eq:irregular-average-fidelity}
\end{align}
where the outer expectation is over the degree tree and the inner
expectation is over the coefficient-trajectory random variable.  For every realization,
Theorem~\ref{thm:qary-tree-closure} gives
\(W_t(\mathsf D_t,\bA_t)\in\cM_q\), so its fidelity already contains the herald
expectation from Definition~\ref{defn:heralded-mixture-class}.
The additional outer expectation over the degree tree is the distinction
from the regular ensemble.  The same distinction appears in
the Bhattacharyya-parameter analysis of classical irregular LDPC ensembles
\cite{richardson2008modern,pfister2014density,sason2006performance}.

Let \(W_t^{\cnop}\) denote the check-to-variable channel after the check-node
half-iteration at level \(t\), with
\(F_{\gamma,\eta}(W_t^{\cnop})\) defined by the corresponding
computation-tree average.  Conditional on check degree \(j\), let
\(\mathsf D_t^{(1)},\ldots,\mathsf D_t^{(j-1)}\) be its independent
incoming degree trees, let
\(\bA_t^{(1)},\ldots,\bA_t^{(j-1)}\) be their coefficient trajectories,
and let \(A_1,\ldots,A_{j-1}\) be the normalized factors associated with
the incoming variable-to-check messages.  Its exact output is
\[
W_t(\mathsf D_t^{(1)},\bA_t^{(1)})^{(A_1)}
\cnop\cdots\cnop
W_t(\mathsf D_t^{(j-1)},\bA_t^{(j-1)})^{(A_{j-1})}.
\]
The coefficient-uniform check-node bound following
Lemma~\ref{lem:heralded PSC fidelity bound} gives
\begin{align*}
&1+(q-1)\E\!\left[
F\!\left(
W_t(\mathsf D_t^{(1)},\bA_t^{(1)})^{(A_1)}
\cnop\cdots\cnop
W_t(\mathsf D_t^{(j-1)},\bA_t^{(j-1)})^{(A_{j-1})}
\right)\,\middle|\,j
\right]\\
&\qquad\leq
\left(1+(q-1)F_{\gamma,\eta}(W_t)\right)^{j-1}.
\end{align*}
Averaging over \(\gamma_j\) yields
\begin{align}
   1+(q-1)F_{\gamma,\eta}(W_t^{\cnop})
   &\leq
   \sum_{j\geq1}\gamma_j
   \left(1+(q-1)F_{\gamma,\eta}(W_t)\right)^{j-1} \notag\\
   &=
   \gamma\!\left(1+(q-1)F_{\gamma,\eta}(W_t)\right).
   \label{eq:irregular-check-fidelity}
\end{align}
The \(i-1\) check-node outputs entering a variable node of degree \(i\)
are also independent.  Applying the bit-node bound and averaging over
\(\eta_i\) yields
\begin{align}
   F_{\gamma,\eta}(W_{t+1})
   &\leq
   F(W)\sum_{i\geq1}\eta_i
   \left((q-1)F_{\gamma,\eta}(W_t^{\cnop})\right)^{i-1}
   \notag\\
   &=
   F(W)\eta\!\left(
   (q-1)F_{\gamma,\eta}(W_t^{\cnop})
   \right) \notag\\
   &\leq
   F(W)\eta\!\left(
   \gamma\!\left(1+(q-1)F_{\gamma,\eta}(W_t)\right)-1
   \right).
   \label{eq:irregular-fidelity-recursion}
\end{align}
Here the averages in both half-iterations include the random degrees,
normalized factors, and incoming trees.  Averaging
Lemma~\ref{lem:heralded-error-fidelity} over the same random-coefficient
computation-tree ensemble, define
\begin{align*}
   P_{\gamma,\eta,t}(W)
   &\triangleq
   \E_{\mathsf D_t}
   \E_{\bA_t\mid\mathsf D_t}
   \left[
      \perr\left(W_t(\mathsf D_t,\bA_t)\right)
   \right].
\end{align*}
The pointwise fidelity--error bounds and Jensen's inequality give
\begin{align}
   \frac{F_{\gamma,\eta}(W_t)^2}{4}
   \leq
   P_{\gamma,\eta,t}(W)
   \leq
   (q-1)F_{\gamma,\eta}(W_t).
   \label{eq:irregular-average-fidelity-error-bounds}
\end{align}
Accordingly, define the BPQM success region of the irregular ensemble by
\begin{align}
\mathrm{Reg}_{\mathrm{BPQM}}(\gamma,\eta)
\coloneqq
\left\{
W:\,
\begin{array}{l}
W\text{ is a symmetric }q\text{-ary PSC},\\[1mm]
\displaystyle\lim_{t\to\infty}F_{\gamma,\eta}(W_t)=0,\quad
\displaystyle\lim_{t\to\infty}P_{\gamma,\eta,t}(W)=0
\end{array}
\right\}.
\label{eq:irregular-bpqm-success-region}
\end{align}
By~\eqref{eq:irregular-average-fidelity-error-bounds}, the two limits are
equivalent, so either one alone characterizes membership in
\(\mathrm{Reg}_{\mathrm{BPQM}}(\gamma,\eta)\).

\begin{theorem}[Double-exponential decay on an irregular tree]
\label{thm:irregular-double-exponential}
Let $W$ be a symmetric $q$-ary PSC, and let
\begin{align*}
   d_\eta
   \triangleq
   \min\{i:\eta_i>0\},
\end{align*}
and assume \(d_\eta\geq3\).  If
\(W\in\mathrm{Reg}_{\mathrm{BPQM}}(\gamma,\eta)\),
then there exist \(A,\beta>0\) and an integer \(\ell\geq0\) such that,
for every \(s\geq0\),
\begin{align}
   \E_{\mathsf D_{\ell+s}}
   \E_{\bA_{\ell+s}\mid\mathsf D_{\ell+s}}
   \left[
      \perr\left(
         W_{\ell+s}(\mathsf D_{\ell+s},\bA_{\ell+s})
      \right)
   \right]
   \leq
   A\exp\!\left(-\beta(d_\eta-1)^s\right).
   \label{eq:irregular-double-exponential}
\end{align}
\end{theorem}

\begin{proof}
Set
\begin{align*}
   r\triangleq d_\eta-1,
   \qquad
   g(x)\triangleq
   \gamma\!\left(1+(q-1)x\right)-1.
\end{align*}
Since \(g(0)=0\) and \(g\) is a polynomial, there exist
\(\epsilon>0\) and \(C<\infty\) such that
\begin{align*}
   g(x)\leq Cx,
   \qquad 0\leq x\leq\epsilon.
\end{align*}
Moreover, \(\eta_i=0\) for \(i<d_\eta\), and hence
\begin{align*}
   \eta(x)
   =
   \sum_{i\geq d_\eta}\eta_i x^{i-1}
   \leq x^r,
   \qquad 0\leq x\leq1.
\end{align*}
After reducing \(\epsilon\) so that \(C\epsilon\leq1\),
\eqref{eq:irregular-fidelity-recursion} gives
\begin{align*}
   F_{\gamma,\eta}(W_{t+1})
   \leq
   \Lambda F_{\gamma,\eta}(W_t)^r,
   \qquad
   \Lambda\triangleq F(W)C^r,
\end{align*}
whenever \(F_{\gamma,\eta}(W_t)\leq\epsilon\).
The success-region assumption places this recursion in the range
\(\Lambda^{1/(r-1)}F_{\gamma,\eta}(W_t)<1\) after finitely many
iterations.  The induction used in
Theorem~\ref{thm:double-exponential-error-rate}, with
\(d_v-1\) replaced by \(r\), then gives
\(F_{\gamma,\eta}(W_{\ell+s})\leq A'e^{-\beta r^s}\)
for some \(A',\beta>0\).  Averaging
Lemma~\ref{lem:heralded-error-fidelity} over the computation-tree
ensemble proves \eqref{eq:irregular-double-exponential}.
\end{proof}

\subsection{Vanishing Block-Error Probability}

We complete the irregular analysis by combining its computation-tree
estimate with neighbourhood counting and minimum-distance recovery.
For \(\cC_N\sim\mathsf E_N(L,P)\), define the ensemble-average block-error
probability by
\begin{align}
   \pblk^{L,P}(N,W)
   \triangleq
   \E_{\cC_N\sim\mathsf E_N(L,P)}
   \bigl[P(\cC_N,W)\bigr].
   \label{eq:irregular-block-error}
\end{align}

Let
\begin{align*}
   d_{v,\max}
   &\triangleq \max\{i:L_i>0\},&
   d_{c,\max}
   &\triangleq \max\{j:P_j>0\},&
   \alpha_{\max}
   &\triangleq
   \max\!\left\{
      1,\,
      (d_{v,\max}-1)(d_{c,\max}-1)
   \right\}.
\end{align*}

For the depth-\(t\) variable-to-check computation associated with a
uniformly chosen variable node, fix one of its outgoing edges as in
Section~\ref{sec:bpqm-ldpc}.  The initial variable-node degree has
distribution \(L\).  Conditional on degree \(i\), the message along the
fixed edge combines the physical observation with the \(i-1\) incoming
check-to-variable messages along the remaining edges.
All subsequent variable- and check-node degrees have distributions
\(\eta\) and \(\gamma\), respectively.  Let
\(\mathsf D_t^{\mathrm r}\) denote this degree tree.  Conditional on
\(\mathsf D_t^{\mathrm r}\), let \(\bA_t^{\mathrm r}\) be the random
variable whose realizations are the corresponding normalized coefficient
trajectories.

\begin{lem}\label{lem:irregular-tree-comparison}
Let \(I\) be uniformly distributed on \([N]\), independently of
\(\cC_N\sim\mathsf E_N(L,P)\).  There are constants
\(C_{L,P},C'_{L,P}<\infty\) such that the number of sockets exposed by a
depth-\(\ell\) neighbourhood is at most
\[
M_\ell\coloneqq C_{L,P}\alpha_{\max}^{\ell},
\]
and, whenever \(M_\ell=o(N)\), every nonnegative functional \(\varphi\)
of a depth-\(\ell\) degree tree and its coefficient trajectory satisfies
\begin{align}
&\E_{\cC_N,I}\!\left[
   \mone\{I\in G_\ell\}
   \varphi\!\left(\cN_\ell(I)\right)
\right]\notag\\
&\qquad\leq
\exp\!\left(
   C'_{L,P}\frac{M_\ell^2}{N}
\right)
\E_{\mathsf D_\ell^{\mathrm r}}
\E_{\bA_\ell^{\mathrm r}\mid\mathsf D_\ell^{\mathrm r}}
\left[
   \varphi\!\left(
      \mathsf D_\ell^{\mathrm r},\bA_\ell^{\mathrm r}
   \right)
\right].
\label{eq:irregular-tree-comparison}
\end{align}
\end{lem}

\begin{proof}
Expose the neighbourhood of \(I\) one socket at a time.  The root degree
has distribution \(L\).  Thereafter, conditional on the sockets exposed so
far, the degree of a newly reached variable or check node is
sampled without replacement from the corresponding size-biased socket
population.  Since the degree distributions have finite support, every
degree in their support has \(\Theta(N)\) sockets.  Until at most
\(M_\ell\) sockets have been exposed, the ratio between each
without-replacement transition probability and the corresponding
\(\eta\)- or \(\gamma\)-probability is at most
\(1+C M_\ell/N\), where \(C\) depends only on \(L\) and \(P\).
It follows that the likelihood ratio of any cycle-free degree-tree
realization is at most
\[
\left(1+C\frac{M_\ell}{N}\right)^{M_\ell}
\leq
\exp\!\left(C'_{L,P}\frac{M_\ell^2}{N}\right).
\]
The nonzero edge coefficients, and hence the relative coefficients, have
the same independent uniform distribution in both models.  Multiplying the
likelihood-ratio bound for each realization by \(\varphi\) and summing over
all degree-tree and coefficient-trajectory pairs
proves~\eqref{eq:irregular-tree-comparison}.
\end{proof}

\begin{lem}\label{lem:irregular-bad-symbols}
With \(C_{L,P}\) as in
Lemma~\ref{lem:irregular-tree-comparison}, there is a constant
\(C^{\mathrm{bad}}_{L,P}<\infty\), independent of \(N\) and \(\ell\), such that,
whenever
\begin{align*}
   C_{L,P}\alpha_{\max}^{\ell}\leq \frac{NL'(1)}{2},
\end{align*}
the number of coordinates whose depth-\(\ell\) neighbourhood is not a tree
satisfies
\begin{align*}
   \E_{\mathsf E_N(L,P)}|B_\ell|
   \leq
   C^{\mathrm{bad}}_{L,P}\alpha_{\max}^{2\ell}.
\end{align*}
\end{lem}

\begin{proof}
Fix a variable node \(i\) and expose its depth-\(\ell\) neighbourhood
during the random socket matching. Since all degrees are bounded by
\(d_{v,\max}\) and
\(d_{c,\max}\), the number of sockets adjacent to exposed vertices is at
most
\begin{align*}
   M_\ell\leq C_{L,P}\alpha_{\max}^{\ell}
\end{align*}
for a constant \(C_{L,P}\) depending only on \(L\) and \(P\).  The
ensemble has \(NL'(1)\) sockets on each side.  Conditional on an
acyclic exposure up to a given pairing, at most \(M_\ell\) sockets on the
opposite side are connected to an already exposed vertex.  Thus, under the
stated condition, the probability that this pairing creates a collision is
at most
\begin{align*}
   \frac{M_\ell}{NL'(1)-M_\ell}
   \leq
   \frac{2M_\ell}{NL'(1)}.
\end{align*}
At most \(M_\ell\) pairings are exposed.  A union bound gives
\begin{align*}
   \Pr_{\mathsf E_N(L,P)}\{i\in B_\ell\}
   \leq
   \frac{2C_{L,P}^{\,2}\alpha_{\max}^{2\ell}}{NL'(1)}.
\end{align*}
Summing over the \(N\) variable nodes proves the claim with
\(C^{\mathrm{bad}}_{L,P}=2C_{L,P}^{\,2}/L'(1)\).
\end{proof}

Define
\begin{align}
   \perr^{L,\gamma,\eta}(W_t)
   \triangleq
   \E_{\mathsf D_t^{\mathrm r}}
   \E_{\bA_t^{\mathrm r}\mid\mathsf D_t^{\mathrm r}}
   \left[
      \perr\left(
         W_t(\mathsf D_t^{\mathrm r},\bA_t^{\mathrm r})
      \right)
   \right]
   \label{eq:irregular-root-error}
\end{align}
as the average error probability of this depth-\(t\) BPQM measurement.
The bit-node fidelity bound and
\eqref{eq:irregular-check-fidelity} give
\begin{align}
   \perr^{L,\gamma,\eta}(W_{t+1})
   &\leq
   (q-1)F(W)
   \sum_{i\geq1}L_i
   \left[
       \gamma\!\left(
           1+(q-1)F_{\gamma,\eta}(W_t)
       \right)-1
   \right]^{i-1}.
   \label{eq:irregular-root-error-bound}
\end{align}
The quantity in square brackets tends to zero by
Theorem~\ref{thm:irregular-double-exponential} and hence belongs to
\([0,1]\) for all sufficiently large \(t\).  Since \(\eta_i>0\) if and
only if \(L_i>0\), the sum is then bounded by the
\((d_\eta-1)\)st power of this quantity.  Applying
Theorem~\ref{thm:irregular-double-exponential} gives
constants \(A_L,\beta_L>0\) and an integer \(\ell_0\) such that
\begin{align}
   \perr^{L,\gamma,\eta}(W_{\ell_0+s})
   \leq
   A_L\exp\!\left(-\beta_L(d_\eta-1)^s\right),
   \qquad s\geq0.
   \label{eq:irregular-root-double-exponential}
\end{align}

\begin{theorem}[Vanishing Block-Error Probability for Irregular Ensembles]
\label{thm:irregular-vanishing-block-error}
Let \(W\) be a symmetric \(q\)-ary PSC.  Consider a \(q\)-ary irregular LDPC
ensemble with finite-support degree
distributions \(L\) and \(P\), and let \(d_\eta\geq3\).  Assume that
the ensemble has linear minimum distance with high probability.  Set
\begin{align*}
   \ell(N)=\lfloor \tau\log\log N\rfloor,
   \qquad
   \tau>\frac{1}{\log(d_\eta-1)}.
\end{align*}
If \(W\in\mathrm{Reg}_{\mathrm{BPQM}}(\gamma,\eta)\),
then
the decoder of Section~\ref{sec:bpqm-ldpc} satisfies
\begin{align*}
   \lim_{N\to\infty}\pblk^{L,P}(N,W)=0.
\end{align*}
\end{theorem}

\begin{proof}
Use the sets \(G_{\ell(N)}\) and \(B_{\ell(N)}\) from
\eqref{eq:good-bad-coordinates}, and set
\[
\kappa_N
\coloneqq
\exp\!\left(
   C'_{L,P}\frac{M_{\ell(N)}^2}{N}
\right).
\]
Applying the sequential-measurement argument from
Theorem~\ref{thm:vanishing_block_error_rate}, averaging over a uniformly
chosen coordinate, and then using
Lemma~\ref{lem:irregular-tree-comparison} with
\(\varphi(\mathsf D,\ba)=\perr(W_{\ell(N)}(\mathsf D,\ba))\) gives
\begin{align}
   \E_{\cC_N\sim\mathsf E_N(L,P)}
   \Pr\!\left\{
      \text{a BPQM error occurs in }G_{\ell(N)}
      \,\middle|\,\cC_N
   \right\}
   &\leq
   4N\kappa_N
   \perr^{L,\gamma,\eta}(W_{\ell(N)}) \notag\\
   &\leq
   4N\kappa_N A_L
   \exp\!\left(
       -\beta_L'
       (d_\eta-1)^{\ell(N)}
   \right),
   \label{eq:irregular-good-coordinate-error}
\end{align}
for some \(\beta_L'>0\).  Since
\(M_{\ell(N)}=O(\alpha_{\max}^{\ell(N)})\) is polylogarithmic in \(N\),
\(\kappa_N\to1\).  The right-hand side tends to zero because
\((d_\eta-1)^{\ell(N)}
\geq(d_\eta-1)^{-1}(\log N)^{\tau\log(d_\eta-1)}\) and
\(\tau\log(d_\eta-1)>1\).

Since \(\ell(N)\leq \tau\log\log N\), one has
\(\alpha_{\max}^{2\ell(N)}=o(N)\).  Lemma~\ref{lem:irregular-bad-symbols}
gives
\begin{align*}
   \E_{\mathsf E_N(L,P)}|B_{\ell(N)}|
   \leq C^{\mathrm{bad}}_{L,P}\alpha_{\max}^{2\ell(N)}
   =o(N).
\end{align*}
Markov's inequality and the linear-distance assumption give
\begin{align*}
   \Pr\!\left\{
       |B_{\ell(N)}|\geq d_{\min}(\cC_N)
   \right\}\longrightarrow0.
\end{align*}
Lemma~\ref{lem:minimum-distance-erasure-recovery} completes the symbols
in \(B_{\ell(N)}\) whenever this event does not occur.  Combining the
two error terms as in Theorem~\ref{thm:vanishing_block_error_rate}
proves the claim.
\end{proof}

For \(\eta(x)=x^{d_v-1}\) and \(\gamma(x)=x^{d_c-1}\), the degree
distributions select degrees \(d_v\) and \(d_c\) with probability one.
Substitution in~\eqref{eq:irregular-fidelity-recursion} gives the
regular-tree bound in Theorem~\ref{thm:double-exponential-error-rate}.

\section{Numerical Visualization of the BPQM Success Region}
\label{app:numerical-bpqm-success-region}

We specialize to symmetric ternary PSCs and the
\((d_v,d_c)=(3,12)\) regular ensemble.  Its design rate is
\(R=1-d_v/d_c=3/4\).  Recall that \(F_t(W)\) and \(P_t(W)\), defined in
Section~\ref{sec:bpqm-density-evolution}, are the average fidelity and
symbol-error probability after \(t\) density-evolution iterations for the
physical channel \(W\).  Theorem~\ref{thm:double-exponential-error-rate}
specializes to
\begin{align}\label{eq:q3-312-fidelity-certificate}
F_{t+1}(W)
&\leq
F(W)\left(\left(1+2F_t(W)\right)^{11}-1\right)^2\nonumber\\
&\leq
\left(\left(1+2F_t(W)\right)^{11}-1\right)^2,
\end{align}
where the second inequality uses \(F(W)\leq1\).  Define
\(T(x)\coloneqq\left((1+2x)^{11}-1\right)^2\).  The binomial expansion
shows that
\begin{align*}
\frac{\left(1+2x\right)^{11}-1}{x}
\end{align*}
is strictly increasing for \(x>0\).  Consequently,
\begin{align*}
\frac{T(x)}{x}
&=
x\left(
\frac{\left(1+2x\right)^{11}-1}{x}
\right)^2
\end{align*}
is also strictly increasing.  Moreover, \(T(x)/x\) tends to zero as
\(x\downarrow0\) and tends to infinity as \(x\to\infty\).  Hence, there is a
unique positive fixed point \(\delta_\star\) satisfying
\begin{align*}
T(\delta_\star)
&=\delta_\star.
\end{align*}
Numerically, \(\delta_\star\approx1.98552\times10^{-3}\).  Every
\(\delta\in(0,\delta_\star)\) satisfies \(T(\delta)/\delta<1\).  We choose
\(\delta\coloneqq1.9\times10^{-3}\) and define
\begin{align*}
\nu_\delta
&\coloneqq
\left(
\frac{\left(1+2\delta\right)^{11}-1}{\delta}
\right)^2.
\end{align*}
Since \(\delta<\delta_\star\),
\(\nu_\delta\delta=T(\delta)/\delta\approx0.955286<1\).
The monotonicity of
\(\bigl((1+2x)^{11}-1\bigr)/x\) implies that, whenever
\(0\leq F_t(W)\leq\delta\),
\eqref{eq:q3-312-fidelity-certificate} gives
\begin{align*}
F_{t+1}(W)
&\leq \nu_\delta F_t(W)^2
\leq
\left(\nu_\delta\delta\right)F_t(W)
\leq
F_t(W)
\leq
\delta.
\end{align*}
Thus, the interval \([0,\delta]\) is preserved by the scalar bound, and
the inequality can be iterated.  Induction then gives, for every integer
\(s\geq0\),
\begin{align*}
F_{t+s}(W)
&\leq
\nu_\delta^{-1}
\left(\nu_\delta F_t(W)\right)^{2^s}\\
&\leq
\nu_\delta^{-1}
\exp\!\left(-\beta_\delta 2^s\right),
\end{align*}
where \(\beta_\delta\coloneqq-\log(\nu_\delta\delta)>0\).  Thus the
finite-iteration condition \(F_t(W)\leq\delta\) certifies
double-exponential decay of the average fidelity.  The upper bound
in~\eqref{Eq:average-fidelity-error-bounds} then gives \(P_t(W)\to0\), and
hence \(W\in\mathrm{Reg}_{\mathrm{BPQM}}(3,12)\).

For the numerical experiment, we evaluate \(2701\) eigen lists on a uniform
triangular grid whose edges are divided into \(72\) equal intervals.  Let
\(\widehat{F}_t(W)\) denote the population density-evolution estimate of
\(F_t(W)\), averaged over three independent runs with population size \(1200\).
After \(40\) iterations, we include a sampled channel \(W\) in the estimated
success region when
\begin{align}\label{eq:numerical-bpqm-tail-criterion}
\max_{33\leq t\leq40}\widehat{F}_t(W)
&\leq \delta.
\end{align}

The requirement over the last eight iterations reduces sensitivity to Monte
Carlo fluctuations and does not modify the asymptotic definition in
\eqref{eq:regular-bpqm-success-region}.  Since
\(\widehat{F}_t(W)\) is a finite-population estimate evaluated on a finite
grid, the blue set in Figure~\ref{fig:bpqm-success-region-q3} is an empirical
approximation to the set of channels satisfying
\(\max_{33\leq t\leq40}F_t(W)\leq\delta\).  This is a sufficient condition for
BPQM convergence, but the blue set does not rigorously characterize the full
BPQM success region.

\begin{figure}[H]
   \centering
   \captionsetup{font=small}
   \resizebox{0.5\textwidth}{!}{%
      \input{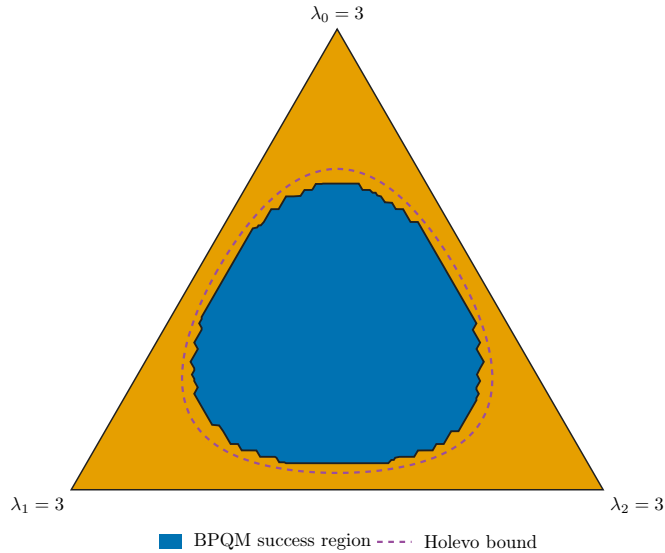}%
   }
   \caption{Numerical estimate of the BPQM success region for symmetric ternary
   PSCs and the \((d_v,d_c)=(3,12)\) regular ensemble with design rate
   \(R=3/4\).  Each point in the simplex represents an eigen list
   \(\blambda=[\lambda_0,\lambda_1,\lambda_2]\) satisfying
   \(\lambda_j\geq0\) and \(\sum_j\lambda_j=3\).  For
   \(\bmu=\blambda/3\), the uniform-input Holevo information is
   \(I(W)=H(\bmu)\).  The blue set satisfies the finite-iteration fidelity
   criterion in~\eqref{eq:numerical-bpqm-tail-criterion}.  The dashed curve is
   the design-rate Holevo boundary
   \(I(W)=R\log_2 3=(3/4)\log_2 3\).}
   \label{fig:bpqm-success-region-q3}
\end{figure}

\section{Minimum-Distance Criterion for Erasure Recovery}
\label{app:minimum-distance-erasure-recovery}

\begin{proof}[Proof of Lemma~\ref{lem:minimum-distance-erasure-recovery}]
Suppose first that $H_B$ does not have full column rank.  There is then a
nonzero vector $\bz_B\in\mathbb F_q^{|B|}$ satisfying
\begin{align*}
   H_B\bz_B=0.
\end{align*}
Extend $\bz_B$ to $\bz\in\mathbb F_q^N$ by setting $z_i=0$ for every
$i\notin B$.  It follows that
\begin{align*}
   H\bz=H_B\bz_B=0,
\end{align*}
so $\bz$ is a nonzero codeword in $\cC$.  Its Hamming weight satisfies
\begin{align*}
   \mathrm{wt}(\bz)
   \leq
   |B|
   <
   d_{\min}(\cC),
\end{align*}
contradicting the definition of minimum distance.  Hence $H_B$ has full
column rank.

Since the transmitted word satisfies $H\bc=0$, its restrictions obey
\begin{align*}
   H_G\bc_G+H_B\bc_B=0.
\end{align*}
Thus, $\bc_B$ is a solution of
\begin{align*}
   H_B\bc_B=-H_G\bc_G.
\end{align*}
Full column rank of $H_B$ implies that this system has at most one
solution.  Because the transmitted restriction $\bc_B$ is a solution, it
is the unique solution and can be recovered by Gaussian elimination.
\end{proof}

\section{Satisfaction Ratios of DQI with BP and BPQM for Regular max-$k$-XORSAT}
\label{app:dqi-comparison}

Table~\ref{tab:dqi-comparison} lists satisfaction ratios for $D$-regular max-$k$-XORSAT with $3\leq k<D\leq 11$. In DQI the decoding problem for such an instance is the decoding of a $(d_v,d_c)=(k,D)$ regular LDPC code~\cite{jordan2025optimization,shutty2026lqd}. With independent and identically distributed amplitudes the induced channel is the binary symmetric pure-state channel
\begin{align*}
\ket{\psi_u}=\sqrt{1-\omega}\,\ket{u}+\sqrt{\omega}\,\ket{u\oplus 1},\qquad u\in\{0,1\}.
\end{align*}
For a decoder with threshold $\omega^\star$ we obtain the satisfaction ratio from the semicircle law as in~\cite{jordan2025optimization,shutty2026lqd,chailloux_opixsoftdecoders_2025},
\begin{align*}
s=\frac12+\sqrt{\omega^\star(1-\omega^\star)}.
\end{align*}

The columns of Table~\ref{tab:dqi-comparison} are obtained as follows.
\begin{itemize}
\item DQI+BP uses the threshold of classical BP on the binary symmetric channel with crossover probability $\omega$, which results from measuring each channel output in the computational basis.
\item DQI+BPQM uses the threshold of BPQM obtained from the density-evolution recursion of Section~\ref{sec:bpqm-density-evolution} with $q=2$.
\item The DQI+BP threshold is estimated by population dynamics with population size $2\times10^5$ and at most $300$ iterations. A value of $\omega$ is counted as a success when the estimated error probability falls below $10^{-7}$, and the threshold is located by bisection in $\omega$ with tolerance $2\times10^{-5}$.
\item The DQI+BPQM threshold is estimated by population dynamics with population size $10^6$ and three seeds on a grid of values with spacing at most $4\times10^{-4}$ in the satisfaction ratio, with at most $2000$ iterations and a run counted as a failure only when the estimated error stays above $10^{-3}$ and decreases by less than $0.1\%$ over $200$ iterations after iteration $300$. The table lists the largest grid value at which all three runs converge.
\item Regev+FGUM is the satisfaction ratio of Regev's reduction with fine-grained unambiguous measurements, computed from the expression in~\cite{shutty2026lqd}.
\item Prange is the satisfaction ratio $\frac12+\frac{k}{2D}$ of Prange's algorithm~\cite{shutty2026lqd}.
\item SA is the satisfaction ratio of simulated annealing on finite instances drawn from Gallager's ensemble, taken from~\cite{shutty2026lqd,jordan2025optimization}.
\end{itemize}
The last two columns give the differences $\Delta_1=\text{DQI+BPQM}-\text{DQI+BP}$ and $\Delta_2=\text{SA}-\text{DQI+BPQM}$. Pairs with $k/D\leq 0.7$ are listed in the upper block of the table and the remaining pairs in the lower block. We highlight again that the BP and BPQM values are numerical estimates of asymptotic density-evolution thresholds, whereas the SA values are obtained on finite-size instances.

\begingroup
\small
\renewcommand{\arraystretch}{1.15}
\setlength{\tabcolsep}{4pt}
\captionsetup{font={small}}
\begin{longtable}{|c|c||c|c||c|c|c||c|c|}
\hline
\multirow{2}{*}{$k$} & \multirow{2}{*}{$D$} & \multicolumn{2}{c||}{Classical algorithms} & \multicolumn{3}{c||}{Quantum algorithms} & \multicolumn{2}{c|}{Differences} \\*
\cline{3-9}
 & & SA & Prange & DQI+BP & DQI+BPQM & Regev+FGUM & $\Delta_1$ & $\Delta_2$ \\*
\hline\hline
\endfirsthead
\multicolumn{9}{l}{Table~\ref{tab:dqi-comparison} (continued)} \\
\hline
\multirow{2}{*}{$k$} & \multirow{2}{*}{$D$} & \multicolumn{2}{c||}{Classical algorithms} & \multicolumn{3}{c||}{Quantum algorithms} & \multicolumn{2}{c|}{Differences} \\*
\cline{3-9}
 & & SA & Prange & DQI+BP & DQI+BPQM & Regev+FGUM & $\Delta_1$ & $\Delta_2$ \\*
\hline\hline
\endhead
\multicolumn{9}{|l|}{$k/D\leq 0.7$} \\*
\hline
\multirow{7}{*}{$3$} & $5$ & $0.9005$ & $0.8000$ & $0.8176$ & $0.8943$ & $0.8379$ & $0.0767$ & $0.0061$ \\*
 & $6$ & $0.8712$ & $0.7500$ & $0.7775$ & $0.8627$ & $0.7857$ & $0.0852$ & $0.0085$ \\*
 & $7$ & $0.8492$ & $0.7143$ & $0.7474$ & $0.8371$ & $0.7621$ & $0.0897$ & $0.0121$ \\*
 & $8$ & $0.8287$ & $0.6875$ & $0.7239$ & $0.8160$ & $0.7312$ & $0.0921$ & $0.0127$ \\*
 & $9$ & $0.8109$ & $0.6667$ & $0.7053$ & $0.7984$ & $0.7183$ & $0.0931$ & $0.0125$ \\*
 & $10$ & $0.7956$ & $0.6500$ & $0.6897$ & $0.7833$ & $0.6976$ & $0.0936$ & $0.0123$ \\*
 & $11$ & $0.7832$ & $0.6364$ & $0.6770$ & $0.7703$ & $0.6895$ & $0.0933$ & $0.0129$ \\
\hline
\multirow{6}{*}{$4$} & $6$ & $0.9024$ & $0.8333$ & $0.8214$ & $0.8980$ & $0.8616$ & $0.0766$ & $0.0044$ \\*
 & $7$ & $0.8771$ & $0.7857$ & $0.7908$ & $0.8733$ & $0.8267$ & $0.0825$ & $0.0038$ \\*
 & $8$ & $0.8587$ & $0.7500$ & $0.7663$ & $0.8523$ & $0.7905$ & $0.0860$ & $0.0064$ \\*
 & $9$ & $0.8391$ & $0.7222$ & $0.7464$ & $0.8340$ & $0.7708$ & $0.0877$ & $0.0050$ \\*
 & $10$ & $0.8238$ & $0.7000$ & $0.7299$ & $0.8186$ & $0.7463$ & $0.0887$ & $0.0052$ \\*
 & $11$ & $0.8111$ & $0.6818$ & $0.7157$ & $0.8047$ & $0.7337$ & $0.0890$ & $0.0064$ \\
\hline
\multirow{4}{*}{$5$} & $8$ & $0.8740$ & $0.8125$ & $0.7894$ & $0.8717$ & $0.8441$ & $0.0823$ & $0.0023$ \\*
 & $9$ & $0.8559$ & $0.7778$ & $0.7690$ & $0.8538$ & $0.8180$ & $0.0848$ & $0.0021$ \\*
 & $10$ & $0.8416$ & $0.7500$ & $0.7520$ & $0.8381$ & $0.7901$ & $0.0861$ & $0.0035$ \\*
 & $11$ & $0.8265$ & $0.7273$ & $0.7374$ & $0.8240$ & $0.7733$ & $0.0866$ & $0.0025$ \\
\hline
\multirow{3}{*}{$6$} & $9$ & $0.8677$ & $0.8333$ & $0.7841$ & $0.8667$ & $0.8637$ & $0.0826$ & $0.0010$ \\*
 & $10$ & $0.8536$ & $0.8000$ & $0.7667$ & $0.8508$ & $0.8324$ & $0.0841$ & $0.0028$ \\*
 & $11$ & $0.8390$ & $0.7727$ & $0.7520$ & $0.8371$ & $0.8114$ & $0.0851$ & $0.0019$ \\
\hline
\multirow{2}{*}{$7$} & $10$ & $0.8613$ & $0.8500$ & $0.7777$ & $0.8603$ & $0.8743$ & $0.0827$ & $0.0010$ \\*
 & $11$ & $0.8471$ & $0.8182$ & $0.7624$ & $0.8466$ & $0.8492$ & $0.0842$ & $0.0005$ \\
\hline
\hline
\multicolumn{9}{|l|}{$k/D> 0.7$} \\*
\hline
\multirow{1}{*}{$3$} & $4$ & $0.9366$ & $0.8750$ & $0.8730$ & $0.9344$ & $0.8930$ & $0.0614$ & $0.0022$ \\
\hline
\multirow{1}{*}{$4$} & $5$ & $0.9279$ & $0.9000$ & $0.8606$ & $0.9274$ & $0.9216$ & $0.0668$ & $0.0005$ \\
\hline
\multirow{2}{*}{$5$} & $6$ & $0.9190$ & $0.9167$ & $0.8443$ & $0.9161$ & $0.9312$ & $0.0718$ & $0.0029$ \\*
 & $7$ & $0.8965$ & $0.8571$ & $0.8142$ & $0.8921$ & $0.8853$ & $0.0779$ & $0.0043$ \\
\hline
\multirow{2}{*}{$6$} & $7$ & $0.9051$ & $0.9286$ & $0.8292$ & $0.9045$ & $0.9427$ & $0.0753$ & $0.0006$ \\*
 & $8$ & $0.8875$ & $0.8750$ & $0.8046$ & $0.8844$ & $0.8962$ & $0.0798$ & $0.0031$ \\
\hline
\multirow{2}{*}{$7$} & $8$ & $0.8941$ & $0.9375$ & $0.8157$ & $0.8935$ & $0.9481$ & $0.0778$ & $0.0006$ \\*
 & $9$ & $0.8763$ & $0.8889$ & $0.7951$ & $0.8760$ & $0.9092$ & $0.0809$ & $0.0003$ \\
\hline
\multirow{3}{*}{$8$} & $9$ & $0.8832$ & $0.9444$ & $0.8037$ & $0.8832$ & $0.9546$ & $0.0795$ & $0.0000$ \\*
 & $10$ & $0.8688$ & $0.9000$ & $0.7857$ & $0.8678$ & $0.9162$ & $0.0821$ & $0.0010$ \\*
 & $11$ & $0.8544$ & $0.8636$ & $0.7709$ & $0.8538$ & $0.8869$ & $0.0828$ & $0.0007$ \\
\hline
\multirow{2}{*}{$9$} & $10$ & $0.8756$ & $0.9500$ & $0.7928$ & $0.8736$ & $0.9581$ & $0.0808$ & $0.0020$ \\*
 & $11$ & $0.8598$ & $0.9091$ & $0.7776$ & $0.8599$ & $0.9246$ & $0.0823$ & $-0.0001$ \\
\hline
\multirow{1}{*}{$10$} & $11$ & $0.8638$ & $0.9545$ & $0.7834$ & $0.8649$ & $0.9623$ & $0.0815$ & $-0.0011$ \\
\hline
\caption{Satisfaction ratios for $D$-regular max-$k$-XORSAT with $3\leq k<D\leq 11$. Pairs with $k/D\leq 0.7$ are listed in the upper block. DQI+BP and DQI+BPQM are obtained from the respective density-evolution thresholds through the semicircle law~\cite{jordan2025optimization,shutty2026lqd,chailloux_opixsoftdecoders_2025}. SA values are from~\cite{shutty2026lqd,jordan2025optimization}, and Regev+FGUM values are computed as in~\cite{shutty2026lqd}. The differences are $\Delta_1=\text{DQI+BPQM}-\text{DQI+BP}$ and $\Delta_2=\text{SA}-\text{DQI+BPQM}$.}\label{tab:dqi-comparison} \\
\end{longtable}
\endgroup

\section{Satisfaction Ratios for Non-Binary max-$k$-LINSAT}
\label{app:dqi-comparison-nonbinary}

Tables~\ref{tab:nonbinary-q3-r1}--\ref{tab:nonbinary-q11-r4} extend the comparison of Appendix~\ref{app:dqi-comparison} to $D$-regular max-$k$-LINSAT over $\mathbb{F}_q$ with $q\in\{3,5,7,11\}$ and $3\leq k<D\leq 8$. Constraint $i$ is satisfied if $\mathbf{b}_i\cdot\mathbf{x}-v_i\in F$ with $F=\{0,\dots,r-1\}$, so each constraint has $r$ satisfying values. In DQI the decoding problem is the decoding of a $(d_v,d_c)=(k,D)$ regular LDPC code over $\mathbb{F}_q$ with nonzero edge coefficients~\cite{jordan2025optimization}.

For a target satisfaction ratio $s\in[r/q,1]$ we use the per-constraint amplitudes $P(y)=\sqrt{s/r}$ for $y\in F$ and $P(y)=\sqrt{(1-s)/(q-r)}$ otherwise, so that a constraint is satisfied with probability $s$. The decoder then receives the symmetric $q$-ary pure-state channel $\ket{\psi_x}=\sum_{z\in\mathbb{F}_q}a(z)\ket{x+z}$, where $a$ is the Fourier transform of $P$. For $q=2$ and $r=1$ this reduces to the channel and the satisfaction ratio of Appendix~\ref{app:dqi-comparison}. 
Because of \cite[Theorem 1 and proposition 2]{chailloux_opixsoftdecoders_2025}, solving this quantum decoding problem with asympotically vanishing error rate leads to an achieved satisfaction ration of $s$.
Concretely, this means that the largest achievable satisfaction ratio is precisely the largest $s$ at which its density-evolution recursion converges.

The columns of Tables~\ref{tab:nonbinary-q3-r1}--\ref{tab:nonbinary-q11-r4} are obtained as follows.
\begin{itemize}
\item DQI+BP uses density evolution to evaluate the threshold of $q$-ary BP on the classical channel obtained by measuring each channel output in the computational basis, which adds noise $z$ with probability $|a(z)|^2$.
\item DQI+BPQM uses the density-evolution recursion of Section~\ref{sec:bpqm-density-evolution} for the eigen list of the channel, with uniformly random nonzero edge coefficients.
\item Both thresholds are estimated by population dynamics with population size $10^5$ and at most $300$ iterations. A value of $s$ is counted as a success when the estimated error probability falls below $10^{-7}$, and as a failure when the estimated error stays above $10^{-3}$ and decreases by less than $0.2\%$ over $30$ iterations. The threshold is located by bisection in $s$ with tolerance $10^{-4}$.
\item The DQI+BPQM values are then recomputed with population size $10^6$ and three seeds on a grid of $s$ values with spacing at most $4\times10^{-4}$, with at most $2000$ iterations and a run counted as a failure only when the estimated error stays above $10^{-3}$ and decreases by less than $0.1\%$ over $200$ iterations after iteration $300$. The tables list the largest grid value at which all three runs converge. These values differ from those obtained with population size $10^5$ by at most $5\times10^{-4}$. Each listed DQI+BPQM value satisfies the finite-iteration criterion of Theorem~\ref{thm:double-exponential-error-rate}, verified as described in Appendix~\ref{app:dqi-comparison}.
\item Prange is the satisfaction ratio $\frac{k}{D}+\left(1-\frac{k}{D}\right)\frac{r}{q}$ of Prange's algorithm.
\item SA is the satisfaction ratio of simulated annealing on instances with $n=2520$ variables drawn from the $(k,D)$ Gallager ensemble with uniformly random nonzero coefficients. We use the schedule of~\cite{jordan2025optimization,shutty2026lqd}, in which the inverse temperature increases linearly from $0$ to $5$ over $10^6$ sweeps, and propose a uniformly random new value for one variable at a time. Each entry is the mean over four instances of the best value among four runs. On identical instances our implementation reproduces the publicly available xortools programs \texttt{anneal} for $q=2$ and \texttt{bigpanneal} with the same schedule for $q\in\{3,5,7\}$~\cite{jordan2025optimization} to within run-to-run variation. The standard deviation over the four instances is at most $0.0025$ for every tabulated pair, with median $0.0007$. These SA values are computed by us with the protocol above, whereas the binary SA values in Table~\ref{tab:dqi-comparison} are taken from published data as described in Appendix~\ref{app:dqi-comparison}.
\end{itemize}
For the remaining computed values of $r$, DQI+BPQM reaches satisfaction ratio $1$ for most or all pairs with $D\leq 8$, namely for $(q,r)=(3,2)$ on 11 of 15 pairs, $(q,r)=(5,3)$ on 14 of 15 pairs, $(q,r)=(5,4)$ on 15 of 15 pairs, $(q,r)=(7,4)$ on 14 of 15 pairs, $(q,r)=(7,5)$ on 15 of 15 pairs, $(q,r)=(7,6)$ on 15 of 15 pairs, $(q,r)=(11,5)$ on 10 of 15 pairs, $(q,r)=(11,6)$ on 15 of 15 pairs, $(q,r)=(11,7)$ on 15 of 15 pairs, $(q,r)=(11,8)$ on 15 of 15 pairs, $(q,r)=(11,9)$ on 15 of 15 pairs, and $(q,r)=(11,10)$ on 15 of 15 pairs. These cases are not tabulated. We highlight again that the BP and BPQM values are numerical estimates of asymptotic density-evolution thresholds, whereas the SA values are obtained on finite-size instances.

\vspace{0.4em}
\noindent\begin{minipage}{\textwidth}
\captionsetup{type=table}
\centering
\footnotesize
\renewcommand{\arraystretch}{1.0}
\setlength{\tabcolsep}{4pt}
\begin{tabular}{|c|c||c|c||c|c||c|c|}
\hline
\multirow{2}{*}{$k$} & \multirow{2}{*}{$D$} & \multicolumn{2}{c||}{Classical algorithms} & \multicolumn{2}{c||}{Quantum algorithms} & \multicolumn{2}{c|}{Differences} \\
\cline{3-8}
 & & SA & Prange & DQI+BP & DQI+BPQM & $\Delta_1$ & $\Delta_2$ \\
\hline\hline
\multirow{5}{*}{$3$} & $4$ & $0.909$ & $0.833$ & $0.813$ & $0.903$ & $0.090$ & $0.006$ \\
 & $5$ & $0.858$ & $0.733$ & $0.738$ & $0.847$ & $0.109$ & $0.010$ \\
 & $6$ & $0.817$ & $0.667$ & $0.685$ & $0.803$ & $0.119$ & $0.013$ \\
 & $7$ & $0.783$ & $0.619$ & $0.645$ & $0.768$ & $0.124$ & $0.014$ \\
 & $8$ & $0.755$ & $0.583$ & $0.614$ & $0.740$ & $0.125$ & $0.015$ \\
\hline
\multirow{4}{*}{$4$} & $5$ & $0.893$ & $0.867$ & $0.793$ & $0.891$ & $0.098$ & $0.002$ \\
 & $6$ & $0.855$ & $0.778$ & $0.740$ & $0.850$ & $0.110$ & $0.005$ \\
 & $7$ & $0.821$ & $0.714$ & $0.699$ & $0.815$ & $0.116$ & $0.005$ \\
 & $8$ & $0.792$ & $0.667$ & $0.666$ & $0.786$ & $0.119$ & $0.006$ \\
\hline
\multirow{3}{*}{$5$} & $6$ & $0.877$ & $0.889$ & $0.769$ & $0.874$ & $0.105$ & $0.003$ \\
 & $7$ & $0.843$ & $0.809$ & $0.729$ & $0.840$ & $0.111$ & $0.003$ \\
 & $8$ & $0.815$ & $0.750$ & $0.695$ & $0.811$ & $0.116$ & $0.004$ \\
\hline
\multirow{2}{*}{$6$} & $7$ & $0.858$ & $0.905$ & $0.748$ & $0.857$ & $0.108$ & $0.001$ \\
 & $8$ & $0.830$ & $0.833$ & $0.715$ & $0.828$ & $0.113$ & $0.002$ \\
\hline
\multirow{1}{*}{$7$} & $8$ & $0.841$ & $0.917$ & $0.730$ & $0.841$ & $0.111$ & $0.000$ \\
\hline
\end{tabular}
\captionsetup{font={small}}
\caption{Satisfaction ratios for $D$-regular max-$k$-LINSAT over $\mathbb{F}_{3}$ with $r=1$ satisfying values per constraint. The differences are $\Delta_1=\text{DQI+BPQM}-\text{DQI+BP}$ and $\Delta_2=\text{SA}-\text{DQI+BPQM}$.}
\label{tab:nonbinary-q3-r1}
\end{minipage}
\vspace{1.0em}

\vspace{0.4em}
\noindent\begin{minipage}{\textwidth}
\captionsetup{type=table}
\centering
\footnotesize
\renewcommand{\arraystretch}{1.0}
\setlength{\tabcolsep}{4pt}
\begin{tabular}{|c|c||c|c||c|c||c|c|}
\hline
\multirow{2}{*}{$k$} & \multirow{2}{*}{$D$} & \multicolumn{2}{c||}{Classical algorithms} & \multicolumn{2}{c||}{Quantum algorithms} & \multicolumn{2}{c|}{Differences} \\
\cline{3-8}
 & & SA & Prange & DQI+BP & DQI+BPQM & $\Delta_1$ & $\Delta_2$ \\
\hline\hline
\multirow{5}{*}{$3$} & $4$ & $0.875$ & $0.800$ & $0.752$ & $0.870$ & $0.118$ & $0.005$ \\
 & $5$ & $0.809$ & $0.680$ & $0.659$ & $0.799$ & $0.140$ & $0.010$ \\
 & $6$ & $0.756$ & $0.600$ & $0.596$ & $0.744$ & $0.148$ & $0.012$ \\
 & $7$ & $0.715$ & $0.543$ & $0.548$ & $0.701$ & $0.153$ & $0.015$ \\
 & $8$ & $0.680$ & $0.500$ & $0.513$ & $0.665$ & $0.153$ & $0.014$ \\
\hline
\multirow{4}{*}{$4$} & $5$ & $0.853$ & $0.840$ & $0.724$ & $0.852$ & $0.129$ & $0.001$ \\
 & $6$ & $0.801$ & $0.733$ & $0.658$ & $0.799$ & $0.140$ & $0.002$ \\
 & $7$ & $0.759$ & $0.657$ & $0.609$ & $0.755$ & $0.146$ & $0.004$ \\
 & $8$ & $0.723$ & $0.600$ & $0.570$ & $0.719$ & $0.149$ & $0.005$ \\
\hline
\multirow{3}{*}{$5$} & $6$ & $0.828$ & $0.867$ & $0.693$ & $0.829$ & $0.136$ & $0.000$ \\
 & $7$ & $0.787$ & $0.771$ & $0.643$ & $0.785$ & $0.142$ & $0.001$ \\
 & $8$ & $0.751$ & $0.700$ & $0.603$ & $0.749$ & $0.146$ & $0.002$ \\
\hline
\multirow{2}{*}{$6$} & $7$ & $0.804$ & $0.886$ & $0.666$ & $0.805$ & $0.140$ & $-0.001$ \\
 & $8$ & $0.770$ & $0.800$ & $0.625$ & $0.769$ & $0.144$ & $0.001$ \\
\hline
\multirow{1}{*}{$7$} & $8$ & $0.783$ & $0.900$ & $0.641$ & $0.783$ & $0.142$ & $0.000$ \\
\hline
\end{tabular}
\captionsetup{font={small}}
\caption{Satisfaction ratios for $D$-regular max-$k$-LINSAT over $\mathbb{F}_{5}$ with $r=1$ satisfying values per constraint. The differences are $\Delta_1=\text{DQI+BPQM}-\text{DQI+BP}$ and $\Delta_2=\text{SA}-\text{DQI+BPQM}$.}
\label{tab:nonbinary-q5-r1}
\end{minipage}
\vspace{1.0em}

\vspace{0.4em}
\noindent\begin{minipage}{\textwidth}
\captionsetup{type=table}
\centering
\footnotesize
\renewcommand{\arraystretch}{1.0}
\setlength{\tabcolsep}{4pt}
\begin{tabular}{|c|c||c|c||c|c||c|c|}
\hline
\multirow{2}{*}{$k$} & \multirow{2}{*}{$D$} & \multicolumn{2}{c||}{Classical algorithms} & \multicolumn{2}{c||}{Quantum algorithms} & \multicolumn{2}{c|}{Differences} \\
\cline{3-8}
 & & SA & Prange & DQI+BP & DQI+BPQM & $\Delta_1$ & $\Delta_2$ \\
\hline\hline
\multirow{5}{*}{$3$} & $4$ & $0.995$ & $0.850$ & $0.942$ & $1.000$ & $0.058$ & $-0.005$ \\
 & $5$ & $0.976$ & $0.760$ & $0.874$ & $0.977$ & $0.104$ & $-0.001$ \\
 & $6$ & $0.945$ & $0.700$ & $0.820$ & $0.943$ & $0.123$ & $0.002$ \\
 & $7$ & $0.915$ & $0.657$ & $0.778$ & $0.911$ & $0.133$ & $0.004$ \\
 & $8$ & $0.889$ & $0.625$ & $0.745$ & $0.883$ & $0.138$ & $0.006$ \\
\hline
\multirow{4}{*}{$4$} & $5$ & $0.991$ & $0.880$ & $0.918$ & $0.993$ & $0.075$ & $-0.003$ \\
 & $6$ & $0.967$ & $0.800$ & $0.868$ & $0.968$ & $0.100$ & $-0.001$ \\
 & $7$ & $0.941$ & $0.743$ & $0.827$ & $0.941$ & $0.114$ & $0.000$ \\
 & $8$ & $0.916$ & $0.700$ & $0.793$ & $0.915$ & $0.122$ & $0.001$ \\
\hline
\multirow{3}{*}{$5$} & $6$ & $0.978$ & $0.900$ & $0.893$ & $0.979$ & $0.086$ & $-0.002$ \\
 & $7$ & $0.955$ & $0.829$ & $0.853$ & $0.956$ & $0.103$ & $-0.001$ \\
 & $8$ & $0.932$ & $0.775$ & $0.820$ & $0.932$ & $0.112$ & $0.000$ \\
\hline
\multirow{2}{*}{$6$} & $7$ & $0.963$ & $0.914$ & $0.870$ & $0.965$ & $0.095$ & $-0.002$ \\
 & $8$ & $0.942$ & $0.850$ & $0.837$ & $0.943$ & $0.106$ & $-0.001$ \\
\hline
\multirow{1}{*}{$7$} & $8$ & $0.950$ & $0.925$ & $0.850$ & $0.951$ & $0.101$ & $-0.001$ \\
\hline
\end{tabular}
\captionsetup{font={small}}
\caption{Satisfaction ratios for $D$-regular max-$k$-LINSAT over $\mathbb{F}_{5}$ with $r=2$ satisfying values per constraint. The differences are $\Delta_1=\text{DQI+BPQM}-\text{DQI+BP}$ and $\Delta_2=\text{SA}-\text{DQI+BPQM}$.}
\label{tab:nonbinary-q5-r2}
\end{minipage}
\vspace{1.0em}

\vspace{0.4em}
\noindent\begin{minipage}{\textwidth}
\captionsetup{type=table}
\centering
\footnotesize
\renewcommand{\arraystretch}{1.0}
\setlength{\tabcolsep}{4pt}
\begin{tabular}{|c|c||c|c||c|c||c|c|}
\hline
\multirow{2}{*}{$k$} & \multirow{2}{*}{$D$} & \multicolumn{2}{c||}{Classical algorithms} & \multicolumn{2}{c||}{Quantum algorithms} & \multicolumn{2}{c|}{Differences} \\
\cline{3-8}
 & & SA & Prange & DQI+BP & DQI+BPQM & $\Delta_1$ & $\Delta_2$ \\
\hline\hline
\multirow{5}{*}{$3$} & $4$ & $0.855$ & $0.786$ & $0.719$ & $0.852$ & $0.133$ & $0.003$ \\
 & $5$ & $0.781$ & $0.657$ & $0.618$ & $0.773$ & $0.155$ & $0.009$ \\
 & $6$ & $0.723$ & $0.571$ & $0.549$ & $0.712$ & $0.162$ & $0.012$ \\
 & $7$ & $0.676$ & $0.510$ & $0.499$ & $0.664$ & $0.165$ & $0.012$ \\
 & $8$ & $0.640$ & $0.464$ & $0.461$ & $0.626$ & $0.165$ & $0.014$ \\
\hline
\multirow{4}{*}{$4$} & $5$ & $0.830$ & $0.829$ & $0.686$ & $0.830$ & $0.145$ & $0.000$ \\
 & $6$ & $0.771$ & $0.714$ & $0.614$ & $0.771$ & $0.156$ & $0.001$ \\
 & $7$ & $0.725$ & $0.633$ & $0.562$ & $0.722$ & $0.160$ & $0.003$ \\
 & $8$ & $0.685$ & $0.571$ & $0.520$ & $0.682$ & $0.162$ & $0.003$ \\
\hline
\multirow{3}{*}{$5$} & $6$ & $0.800$ & $0.857$ & $0.651$ & $0.802$ & $0.151$ & $-0.002$ \\
 & $7$ & $0.753$ & $0.755$ & $0.597$ & $0.754$ & $0.157$ & $-0.001$ \\
 & $8$ & $0.714$ & $0.679$ & $0.554$ & $0.714$ & $0.159$ & $0.000$ \\
\hline
\multirow{2}{*}{$6$} & $7$ & $0.774$ & $0.878$ & $0.621$ & $0.776$ & $0.155$ & $-0.002$ \\
 & $8$ & $0.734$ & $0.786$ & $0.577$ & $0.735$ & $0.158$ & $-0.001$ \\
\hline
\multirow{1}{*}{$7$} & $8$ & $0.749$ & $0.893$ & $0.594$ & $0.751$ & $0.157$ & $-0.002$ \\
\hline
\end{tabular}
\captionsetup{font={small}}
\caption{Satisfaction ratios for $D$-regular max-$k$-LINSAT over $\mathbb{F}_{7}$ with $r=1$ satisfying values per constraint. The differences are $\Delta_1=\text{DQI+BPQM}-\text{DQI+BP}$ and $\Delta_2=\text{SA}-\text{DQI+BPQM}$.}
\label{tab:nonbinary-q7-r1}
\end{minipage}
\vspace{1.0em}

\vspace{0.4em}
\noindent\begin{minipage}{\textwidth}
\captionsetup{type=table}
\centering
\footnotesize
\renewcommand{\arraystretch}{1.0}
\setlength{\tabcolsep}{4pt}
\begin{tabular}{|c|c||c|c||c|c||c|c|}
\hline
\multirow{2}{*}{$k$} & \multirow{2}{*}{$D$} & \multicolumn{2}{c||}{Classical algorithms} & \multicolumn{2}{c||}{Quantum algorithms} & \multicolumn{2}{c|}{Differences} \\
\cline{3-8}
 & & SA & Prange & DQI+BP & DQI+BPQM & $\Delta_1$ & $\Delta_2$ \\
\hline\hline
\multirow{5}{*}{$3$} & $4$ & $0.984$ & $0.821$ & $0.896$ & $0.989$ & $0.093$ & $-0.005$ \\
 & $5$ & $0.942$ & $0.714$ & $0.809$ & $0.941$ & $0.132$ & $0.001$ \\
 & $6$ & $0.896$ & $0.643$ & $0.745$ & $0.894$ & $0.149$ & $0.002$ \\
 & $7$ & $0.858$ & $0.592$ & $0.696$ & $0.853$ & $0.157$ & $0.005$ \\
 & $8$ & $0.824$ & $0.554$ & $0.657$ & $0.818$ & $0.161$ & $0.006$ \\
\hline
\multirow{4}{*}{$4$} & $5$ & $0.964$ & $0.857$ & $0.863$ & $0.968$ & $0.105$ & $-0.004$ \\
 & $6$ & $0.927$ & $0.762$ & $0.801$ & $0.929$ & $0.128$ & $-0.002$ \\
 & $7$ & $0.892$ & $0.694$ & $0.751$ & $0.892$ & $0.140$ & $0.000$ \\
 & $8$ & $0.859$ & $0.643$ & $0.712$ & $0.859$ & $0.147$ & $0.000$ \\
\hline
\multirow{3}{*}{$5$} & $6$ & $0.943$ & $0.881$ & $0.830$ & $0.947$ & $0.117$ & $-0.003$ \\
 & $7$ & $0.910$ & $0.796$ & $0.782$ & $0.913$ & $0.130$ & $-0.002$ \\
 & $8$ & $0.879$ & $0.732$ & $0.741$ & $0.881$ & $0.140$ & $-0.002$ \\
\hline
\multirow{2}{*}{$6$} & $7$ & $0.923$ & $0.898$ & $0.802$ & $0.926$ & $0.124$ & $-0.002$ \\
 & $8$ & $0.893$ & $0.821$ & $0.762$ & $0.895$ & $0.134$ & $-0.003$ \\
\hline
\multirow{1}{*}{$7$} & $8$ & $0.902$ & $0.911$ & $0.776$ & $0.906$ & $0.130$ & $-0.003$ \\
\hline
\end{tabular}
\captionsetup{font={small}}
\caption{Satisfaction ratios for $D$-regular max-$k$-LINSAT over $\mathbb{F}_{7}$ with $r=2$ satisfying values per constraint. The differences are $\Delta_1=\text{DQI+BPQM}-\text{DQI+BP}$ and $\Delta_2=\text{SA}-\text{DQI+BPQM}$.}
\label{tab:nonbinary-q7-r2}
\end{minipage}
\vspace{1.0em}

\vspace{0.4em}
\noindent\begin{minipage}{\textwidth}
\captionsetup{type=table}
\centering
\footnotesize
\renewcommand{\arraystretch}{1.0}
\setlength{\tabcolsep}{4pt}
\begin{tabular}{|c|c||c|c||c|c||c|c|}
\hline
\multirow{2}{*}{$k$} & \multirow{2}{*}{$D$} & \multicolumn{2}{c||}{Classical algorithms} & \multicolumn{2}{c||}{Quantum algorithms} & \multicolumn{2}{c|}{Differences} \\
\cline{3-8}
 & & SA & Prange & DQI+BP & DQI+BPQM & $\Delta_1$ & $\Delta_2$ \\
\hline\hline
\multirow{5}{*}{$3$} & $4$ & $0.997$ & $0.857$ & $0.986$ & $1.000$ & $0.014$ & $-0.003$ \\
 & $5$ & $0.995$ & $0.771$ & $0.931$ & $1.000$ & $0.069$ & $-0.005$ \\
 & $6$ & $0.985$ & $0.714$ & $0.881$ & $0.987$ & $0.107$ & $-0.002$ \\
 & $7$ & $0.963$ & $0.673$ & $0.839$ & $0.963$ & $0.124$ & $0.000$ \\
 & $8$ & $0.940$ & $0.643$ & $0.805$ & $0.939$ & $0.133$ & $0.002$ \\
\hline
\multirow{4}{*}{$4$} & $5$ & $0.997$ & $0.886$ & $0.964$ & $1.000$ & $0.036$ & $-0.003$ \\
 & $6$ & $0.994$ & $0.809$ & $0.921$ & $0.998$ & $0.077$ & $-0.005$ \\
 & $7$ & $0.979$ & $0.755$ & $0.883$ & $0.981$ & $0.099$ & $-0.002$ \\
 & $8$ & $0.960$ & $0.714$ & $0.850$ & $0.962$ & $0.112$ & $-0.001$ \\
\hline
\multirow{3}{*}{$5$} & $6$ & $0.995$ & $0.905$ & $0.941$ & $1.000$ & $0.059$ & $-0.005$ \\
 & $7$ & $0.988$ & $0.837$ & $0.905$ & $0.990$ & $0.085$ & $-0.002$ \\
 & $8$ & $0.971$ & $0.786$ & $0.874$ & $0.973$ & $0.099$ & $-0.002$ \\
\hline
\multirow{2}{*}{$6$} & $7$ & $0.991$ & $0.918$ & $0.920$ & $0.994$ & $0.075$ & $-0.003$ \\
 & $8$ & $0.978$ & $0.857$ & $0.888$ & $0.980$ & $0.091$ & $-0.002$ \\
\hline
\multirow{1}{*}{$7$} & $8$ & $0.982$ & $0.929$ & $0.900$ & $0.984$ & $0.084$ & $-0.003$ \\
\hline
\end{tabular}
\captionsetup{font={small}}
\caption{Satisfaction ratios for $D$-regular max-$k$-LINSAT over $\mathbb{F}_{7}$ with $r=3$ satisfying values per constraint. The differences are $\Delta_1=\text{DQI+BPQM}-\text{DQI+BP}$ and $\Delta_2=\text{SA}-\text{DQI+BPQM}$.}
\label{tab:nonbinary-q7-r3}
\end{minipage}
\vspace{1.0em}

\vspace{0.4em}
\noindent\begin{minipage}{\textwidth}
\captionsetup{type=table}
\centering
\footnotesize
\renewcommand{\arraystretch}{1.0}
\setlength{\tabcolsep}{4pt}
\begin{tabular}{|c|c||c|c||c|c||c|c|}
\hline
\multirow{2}{*}{$k$} & \multirow{2}{*}{$D$} & \multicolumn{2}{c||}{Classical algorithms} & \multicolumn{2}{c||}{Quantum algorithms} & \multicolumn{2}{c|}{Differences} \\
\cline{3-8}
 & & SA & Prange & DQI+BP & DQI+BPQM & $\Delta_1$ & $\Delta_2$ \\
\hline\hline
\multirow{5}{*}{$3$} & $4$ & $0.830$ & $0.773$ & $0.681$ & $0.830$ & $0.149$ & $0.000$ \\
 & $5$ & $0.748$ & $0.636$ & $0.572$ & $0.742$ & $0.170$ & $0.006$ \\
 & $6$ & $0.686$ & $0.545$ & $0.499$ & $0.676$ & $0.177$ & $0.010$ \\
 & $7$ & $0.635$ & $0.480$ & $0.447$ & $0.624$ & $0.177$ & $0.011$ \\
 & $8$ & $0.595$ & $0.432$ & $0.408$ & $0.583$ & $0.175$ & $0.012$ \\
\hline
\multirow{4}{*}{$4$} & $5$ & $0.801$ & $0.818$ & $0.644$ & $0.804$ & $0.160$ & $-0.003$ \\
 & $6$ & $0.737$ & $0.697$ & $0.567$ & $0.738$ & $0.170$ & $-0.001$ \\
 & $7$ & $0.684$ & $0.610$ & $0.510$ & $0.685$ & $0.174$ & $0.000$ \\
 & $8$ & $0.642$ & $0.545$ & $0.468$ & $0.641$ & $0.174$ & $0.001$ \\
\hline
\multirow{3}{*}{$5$} & $6$ & $0.768$ & $0.849$ & $0.606$ & $0.772$ & $0.166$ & $-0.004$ \\
 & $7$ & $0.716$ & $0.740$ & $0.547$ & $0.719$ & $0.171$ & $-0.002$ \\
 & $8$ & $0.671$ & $0.659$ & $0.502$ & $0.674$ & $0.172$ & $-0.003$ \\
\hline
\multirow{2}{*}{$6$} & $7$ & $0.736$ & $0.870$ & $0.572$ & $0.742$ & $0.169$ & $-0.006$ \\
 & $8$ & $0.693$ & $0.773$ & $0.526$ & $0.697$ & $0.171$ & $-0.004$ \\
\hline
\multirow{1}{*}{$7$} & $8$ & $0.709$ & $0.886$ & $0.544$ & $0.714$ & $0.170$ & $-0.005$ \\
\hline
\end{tabular}
\captionsetup{font={small}}
\caption{Satisfaction ratios for $D$-regular max-$k$-LINSAT over $\mathbb{F}_{11}$ with $r=1$ satisfying values per constraint. The differences are $\Delta_1=\text{DQI+BPQM}-\text{DQI+BP}$ and $\Delta_2=\text{SA}-\text{DQI+BPQM}$.}
\label{tab:nonbinary-q11-r1}
\end{minipage}
\vspace{1.0em}

\vspace{0.4em}
\noindent\begin{minipage}{\textwidth}
\captionsetup{type=table}
\centering
\footnotesize
\renewcommand{\arraystretch}{1.0}
\setlength{\tabcolsep}{4pt}
\begin{tabular}{|c|c||c|c||c|c||c|c|}
\hline
\multirow{2}{*}{$k$} & \multirow{2}{*}{$D$} & \multicolumn{2}{c||}{Classical algorithms} & \multicolumn{2}{c||}{Quantum algorithms} & \multicolumn{2}{c|}{Differences} \\
\cline{3-8}
 & & SA & Prange & DQI+BP & DQI+BPQM & $\Delta_1$ & $\Delta_2$ \\
\hline\hline
\multirow{5}{*}{$3$} & $4$ & $0.953$ & $0.795$ & $0.835$ & $0.960$ & $0.124$ & $-0.006$ \\
 & $5$ & $0.892$ & $0.673$ & $0.733$ & $0.893$ & $0.160$ & $-0.001$ \\
 & $6$ & $0.837$ & $0.591$ & $0.659$ & $0.834$ & $0.175$ & $0.002$ \\
 & $7$ & $0.789$ & $0.532$ & $0.604$ & $0.786$ & $0.182$ & $0.004$ \\
 & $8$ & $0.750$ & $0.489$ & $0.562$ & $0.745$ & $0.183$ & $0.005$ \\
\hline
\multirow{4}{*}{$4$} & $5$ & $0.925$ & $0.836$ & $0.794$ & $0.931$ & $0.136$ & $-0.006$ \\
 & $6$ & $0.874$ & $0.727$ & $0.721$ & $0.878$ & $0.157$ & $-0.005$ \\
 & $7$ & $0.829$ & $0.649$ & $0.664$ & $0.832$ & $0.168$ & $-0.003$ \\
 & $8$ & $0.790$ & $0.591$ & $0.620$ & $0.792$ & $0.172$ & $-0.001$ \\
\hline
\multirow{3}{*}{$5$} & $6$ & $0.896$ & $0.864$ & $0.755$ & $0.901$ & $0.147$ & $-0.006$ \\
 & $7$ & $0.853$ & $0.766$ & $0.698$ & $0.857$ & $0.159$ & $-0.005$ \\
 & $8$ & $0.815$ & $0.693$ & $0.653$ & $0.818$ & $0.166$ & $-0.003$ \\
\hline
\multirow{2}{*}{$6$} & $7$ & $0.868$ & $0.883$ & $0.721$ & $0.874$ & $0.153$ & $-0.006$ \\
 & $8$ & $0.832$ & $0.795$ & $0.675$ & $0.836$ & $0.161$ & $-0.004$ \\
\hline
\multirow{1}{*}{$7$} & $8$ & $0.843$ & $0.898$ & $0.692$ & $0.849$ & $0.157$ & $-0.006$ \\
\hline
\end{tabular}
\captionsetup{font={small}}
\caption{Satisfaction ratios for $D$-regular max-$k$-LINSAT over $\mathbb{F}_{11}$ with $r=2$ satisfying values per constraint. The differences are $\Delta_1=\text{DQI+BPQM}-\text{DQI+BP}$ and $\Delta_2=\text{SA}-\text{DQI+BPQM}$.}
\label{tab:nonbinary-q11-r2}
\end{minipage}
\vspace{1.0em}

\vspace{0.4em}
\noindent\begin{minipage}{\textwidth}
\captionsetup{type=table}
\centering
\footnotesize
\renewcommand{\arraystretch}{1.0}
\setlength{\tabcolsep}{4pt}
\begin{tabular}{|c|c||c|c||c|c||c|c|}
\hline
\multirow{2}{*}{$k$} & \multirow{2}{*}{$D$} & \multicolumn{2}{c||}{Classical algorithms} & \multicolumn{2}{c||}{Quantum algorithms} & \multicolumn{2}{c|}{Differences} \\
\cline{3-8}
 & & SA & Prange & DQI+BP & DQI+BPQM & $\Delta_1$ & $\Delta_2$ \\
\hline\hline
\multirow{5}{*}{$3$} & $4$ & $0.989$ & $0.818$ & $0.936$ & $1.000$ & $0.064$ & $-0.011$ \\
 & $5$ & $0.968$ & $0.709$ & $0.849$ & $0.973$ & $0.124$ & $-0.005$ \\
 & $6$ & $0.929$ & $0.636$ & $0.781$ & $0.930$ & $0.150$ & $-0.001$ \\
 & $7$ & $0.890$ & $0.584$ & $0.727$ & $0.890$ & $0.163$ & $0.001$ \\
 & $8$ & $0.856$ & $0.545$ & $0.686$ & $0.854$ & $0.168$ & $0.001$ \\
\hline
\multirow{4}{*}{$4$} & $5$ & $0.983$ & $0.855$ & $0.898$ & $0.990$ & $0.092$ & $-0.007$ \\
 & $6$ & $0.952$ & $0.758$ & $0.834$ & $0.957$ & $0.123$ & $-0.005$ \\
 & $7$ & $0.920$ & $0.688$ & $0.782$ & $0.923$ & $0.141$ & $-0.003$ \\
 & $8$ & $0.887$ & $0.636$ & $0.740$ & $0.890$ & $0.150$ & $-0.003$ \\
\hline
\multirow{3}{*}{$5$} & $6$ & $0.966$ & $0.879$ & $0.862$ & $0.971$ & $0.109$ & $-0.005$ \\
 & $7$ & $0.935$ & $0.792$ & $0.811$ & $0.940$ & $0.129$ & $-0.004$ \\
 & $8$ & $0.905$ & $0.727$ & $0.768$ & $0.909$ & $0.141$ & $-0.004$ \\
\hline
\multirow{2}{*}{$6$} & $7$ & $0.946$ & $0.896$ & $0.830$ & $0.950$ & $0.121$ & $-0.005$ \\
 & $8$ & $0.917$ & $0.818$ & $0.788$ & $0.922$ & $0.134$ & $-0.005$ \\
\hline
\multirow{1}{*}{$7$} & $8$ & $0.926$ & $0.909$ & $0.802$ & $0.931$ & $0.128$ & $-0.004$ \\
\hline
\end{tabular}
\captionsetup{font={small}}
\caption{Satisfaction ratios for $D$-regular max-$k$-LINSAT over $\mathbb{F}_{11}$ with $r=3$ satisfying values per constraint. The differences are $\Delta_1=\text{DQI+BPQM}-\text{DQI+BP}$ and $\Delta_2=\text{SA}-\text{DQI+BPQM}$.}
\label{tab:nonbinary-q11-r3}
\end{minipage}
\vspace{1.0em}

\vspace{0.4em}
\noindent\begin{minipage}{\textwidth}
\captionsetup{type=table}
\centering
\footnotesize
\renewcommand{\arraystretch}{1.0}
\setlength{\tabcolsep}{4pt}
\begin{tabular}{|c|c||c|c||c|c||c|c|}
\hline
\multirow{2}{*}{$k$} & \multirow{2}{*}{$D$} & \multicolumn{2}{c||}{Classical algorithms} & \multicolumn{2}{c||}{Quantum algorithms} & \multicolumn{2}{c|}{Differences} \\
\cline{3-8}
 & & SA & Prange & DQI+BP & DQI+BPQM & $\Delta_1$ & $\Delta_2$ \\
\hline\hline
\multirow{5}{*}{$3$} & $4$ & $0.995$ & $0.841$ & $0.990$ & $1.000$ & $0.010$ & $-0.005$ \\
 & $5$ & $0.994$ & $0.746$ & $0.929$ & $1.000$ & $0.071$ & $-0.006$ \\
 & $6$ & $0.982$ & $0.682$ & $0.872$ & $0.986$ & $0.114$ & $-0.005$ \\
 & $7$ & $0.957$ & $0.636$ & $0.824$ & $0.959$ & $0.135$ & $-0.002$ \\
 & $8$ & $0.931$ & $0.602$ & $0.786$ & $0.932$ & $0.146$ & $0.000$ \\
\hline
\multirow{4}{*}{$4$} & $5$ & $0.995$ & $0.873$ & $0.964$ & $1.000$ & $0.036$ & $-0.005$ \\
 & $6$ & $0.991$ & $0.788$ & $0.914$ & $0.997$ & $0.083$ & $-0.006$ \\
 & $7$ & $0.974$ & $0.727$ & $0.870$ & $0.978$ & $0.108$ & $-0.004$ \\
 & $8$ & $0.952$ & $0.682$ & $0.833$ & $0.956$ & $0.123$ & $-0.003$ \\
\hline
\multirow{3}{*}{$5$} & $6$ & $0.993$ & $0.894$ & $0.935$ & $1.000$ & $0.065$ & $-0.007$ \\
 & $7$ & $0.983$ & $0.818$ & $0.894$ & $0.987$ & $0.093$ & $-0.004$ \\
 & $8$ & $0.964$ & $0.761$ & $0.857$ & $0.968$ & $0.110$ & $-0.004$ \\
\hline
\multirow{2}{*}{$6$} & $7$ & $0.988$ & $0.909$ & $0.909$ & $0.992$ & $0.083$ & $-0.004$ \\
 & $8$ & $0.971$ & $0.841$ & $0.873$ & $0.975$ & $0.102$ & $-0.004$ \\
\hline
\multirow{1}{*}{$7$} & $8$ & $0.976$ & $0.920$ & $0.885$ & $0.980$ & $0.095$ & $-0.004$ \\
\hline
\end{tabular}
\captionsetup{font={small}}
\caption{Satisfaction ratios for $D$-regular max-$k$-LINSAT over $\mathbb{F}_{11}$ with $r=4$ satisfying values per constraint. The differences are $\Delta_1=\text{DQI+BPQM}-\text{DQI+BP}$ and $\Delta_2=\text{SA}-\text{DQI+BPQM}$.}
\label{tab:nonbinary-q11-r4}
\end{minipage}
\vspace{1.0em}

\section{Discussion of the Numerical Comparison}
\label{app:dqi-discussion}

Tables~\ref{tab:dqi-comparison} and~\ref{tab:nonbinary-q3-r1}--\ref{tab:nonbinary-q11-r4} compare DQI with the BPQM decoder against DQI with classical BP, simulated annealing, Prange's algorithm and, in the binary case, Regev's reduction with FGUM.

\paragraph{DQI+BPQM and DQI+BP.}
DQI+BPQM attains a strictly higher satisfaction ratio than DQI+BP for every pair in the tables. For $r=1$ the difference increases with the alphabet size, from $0.061$--$0.094$ for $q=2$ to $0.090$--$0.125$ for $q=3$, $0.118$--$0.153$ for $q=5$, $0.133$--$0.165$ for $q=7$ and $0.149$--$0.177$ for $q=11$. For larger $r$ both values approach $1$ and the difference decreases.

\paragraph{DQI+BPQM and simulated annealing.}
The satisfaction ratios of DQI+BPQM and simulated annealing match very closely, and differ by at most $0.015$ for all pairs in the tables. Simulated annealing is ahead mainly for $r=1$ and small $k$, by up to $0.015$. For $q\in\{5,7,11\}$ and $r\geq2$, DQI+BPQM is ahead on most pairs, by up to $0.011$. This includes pairs at which DQI+BPQM reaches satisfaction ratio $1$ while simulated annealing remains below $1$.

\paragraph{DQI+BPQM, Prange's algorithm and Regev+FGUM.}
Prange's algorithm exceeds DQI+BPQM only for pairs with $k/D\geq 5/7$ and $r\leq 2$. For the binary pairs with $k/D\leq 0.7$, DQI+BPQM exceeds Prange's algorithm for all pairs and Regev+FGUM for all pairs except $(k,D)\in\{(7,10),(7,11)\}$. Regev+FGUM exceeds DQI+BPQM for $13$ of the $36$ binary pairs, all of which have $k\geq5$ and $k/D>0.6$.

\paragraph{Interpretation.}
Our results strongly suggest that DQI+BPQM and SA achieve essentially identical satisfaction ratios.
The observed discrepancy between the two methods are very small and can likely be explained through numerical inaccuracies and finite-size effects.
Indeed, DQI+BP and DQI+BPQM values are density-evolution thresholds, which describe codes in the limit of large blocklength.
Simulated annealing, in contrast, is run on finite instances with $n=2520$ variables for $10^6$ sweeps, and its convergence is not established.
To establish a more rigorous connection between these two techniques, it would be necessary to derive a rigorous asymptotic analysis of the SA satisfaction ratio, ideally showing that it exactly matches the BPQM density evolution threshold.
We leave the investigation of this conjecture to future work.

\section*{AI Statement}

The authors used ChatGPT (GPT-5.6 Sol and earlier GPT model versions) and Claude (Opus 5.5)  to
improve the clarity and presentation of the manuscript.  They also used
Codex and Claude Code to refine scripts used to generate the numerical results, specifically Appendix~\ref{app:dqi-comparison} and Appendix~\ref{app:dqi-comparison-nonbinary}.  The main
ideas and results, including the doubly exponential decay of the PGM error
under BPQM and the vanishing block-error probability of the decoder combining
BPQM with Gaussian elimination, were conceived and developed by the authors.
The authors
independently verified all mathematical arguments, numerical results, and
interpretations and take full responsibility for the content of the
manuscript.

\section*{Acknowledgements}

The authors thank Stephen Jordan and Noah Shutty for helpful discussions related to this work. JMR acknowledges support from the Swiss State Secretariat for Education, Research and Innovation (SERI) under contract No.
20QU-1\_225224 and the ETH Quantum Center. 
\printbibliography
\end{document}